\documentclass[11pt,a4paper]{amsart}
\usepackage{amsmath,amsfonts,amsthm,amssymb}
\usepackage{color}
\usepackage{epsfig}
\usepackage{graphicx}
\usepackage{cite,url}
\usepackage{hyperref}
\usepackage{xspace}
\usepackage{enumerate}
\usepackage[left=1in,right=1in,top=1in,bottom=1.1in,foot=0.5in]{geometry}

\theoremstyle{plain}
\newtheorem{theorem}{Theorem}[section]
\newtheorem{lemma}[theorem]{Lemma}
\newtheorem{proposition}[theorem]{Proposition}
\newtheorem{corollary}[theorem]{Corollary}
\newtheorem{claim}[theorem]{Claim}

\theoremstyle{definition}

\newtheorem{remark}[theorem]{Remark}

\newtheorem{example}[theorem]{Example}

\numberwithin{equation}{section}

\definecolor{airforceblue}{rgb}{0.36, 0.54, 0.66}

\DeclareMathOperator{\dom}{dom}
\DeclareMathOperator{\sign}{sign}
\DeclareMathOperator{\tg}{tag}
\DeclareMathOperator{\supp}{supp}
\DeclareMathOperator{\ex}{ex}

\DeclareMathOperator{\tail}{tail}
\DeclareMathOperator{\Ima}{Im}
\DeclareMathOperator{\RS}{RS}
\DeclareMathOperator{\vcdim}{VC-dim}

\newcommand{\uoX}{\ensuremath{X}\xspace}
\newcommand{\uX}{\ensuremath{\underline{X}}\xspace}
\newcommand{\oX}{\ensuremath{\overline{X}}\xspace}

\newcommand{\dG}{\ensuremath{\overrightarrow{G}}\xspace}
\newcommand{\dGo}{\ensuremath{\overrightarrow{G_o}}\xspace}

\newcommand{\U}{\ensuremath{U}\xspace}
\newcommand{\cupdot}{\mathbin{\mathaccent\cdot\cup}}

\newcommand{\R}{\mathbb R}

\begin{document}

\title[Compression schemes and corner peelings for ample and maximum
classes]{Unlabeled sample compression schemes and corner peelings for
  ample and maximum classes$^*$}\thanks{$^*$An extended
  abstract~\cite{CCMW-icalp} of this paper has appeared in the
  proceedings of ICALP 2019.}

\author[J.\ Chalopin]{J\'er\'emie Chalopin}
\address{CNRS, Aix-Marseille Universit\'e, Universit\'e de Toulon,
  LIS, Marseille, France}
\email{jeremie.chalopin@lis-lab.fr}

\author[V. Chepoi]{Victor Chepoi}
\address{Aix-Marseille Universit\'e, CNRS, Universit\'e de Toulon,
  LIS, Marseille, France}
\email{victor.chepoi@lis-lab.fr}

\author[S. Moran]{Shay Moran}
\address{Department of Mathematics, Technion and Google Research
  % Department of Computer Science, Princeton University,
  % Princeton, USA
}
\email{shaym@cs.princeton.edu}

\author[M. K. Warmuth]{Manfred K. Warmuth}
\address{Visiting Google Brain, Mountain View, California, USA.
Formerly Computer Science Department, University of California, Santa Cruz, USA}
\email{manfredwarmuth57@gmail.com}


\maketitle
\begin{abstract}
  We examine connections between combinatorial notions that arise in
  machine learning and topological notions in cubical/simplicial
  geometry.  These connections enable to export results from geometry
  to machine learning.  Our first main result is based on a geometric
  construction by Tracy Hall (2004) of a partial shelling of the
  cross-polytope which can not be extended.  We use it to derive a
  maximum class of VC dimension~3 that has no corners.  This refutes
  several previous works in machine learning from the past 11 years.
  In particular, it implies that the previous constructions of optimal
  unlabeled sample compression schemes for maximum classes are
  erroneous.

  On the positive side we present a new construction of an optimal
  unlabeled sample compression scheme for maximum classes.  We leave
  as open whether our unlabeled sample compression scheme extends to
  ample (a.k.a.~lopsided or extremal) classes, which represent a
  natural and far-reaching generalization of maximum classes.  Towards
  resolving this question, we provide a geometric characterization in
  terms of unique sink orientations of the 1-skeletons of associated
  cubical complexes.
\end{abstract}

\section{Introduction}\label{s:intro}
The Sauer-Shelah-Perles Lemma~\cite{Sauer,Shelah,VaCh} is arguably the
most basic fact in VC theory; it asserts that any class
$C\subseteq\{0,1\}^n$ satisfies
$\lvert C\rvert \leq {\binom{n} {\leq d}}$, where $d=\vcdim(C)$.  A
beautiful generalization of Sauer-Shelah-Perles's inequality asserts
that $\lvert C\rvert\leq\lvert \oX(C)\rvert$, where $\oX(C)$ is the
family of subsets that are shattered by $C$.\footnote{Note that this
  inequality indeed implies the Sauer-Shelah-Perles Lemma, since
  $ \oX(C) \subseteq {\binom{[n]}{\leq d}}$.}  The latter
inequality is a part of the Sandwich Lemma~\cite{AnRoSa,BoRa,Dr,Pa},
which also provides a lower bound for $\lvert C\rvert$ (and thus
``sandwiches'' $\lvert C\rvert$) in terms of the number of its
\emph{strongly shattered subsets} (see Section~\ref{s:prel}).  A class
$C$ is called \emph{maximum/ample} if the Sauer-Shelah-Perles/Sandwich
upper bounds are tight (respectively).  Every maximum class is ample,
but not vice versa.

Maximum classes were studied mostly in discrete geometry and machine
learning, e.g.~\cite{Welzl,GaWe,Fl,FlWa,KuWa}.  The history of ample
classes is more interesting as they were discovered independently by
several works in disparate
contexts~\cite{AnRoSa,La,BoRa,Moran,BaChDrKo,Dr,Wiedemann}.
Consequently, they received different names such as lopsided
classes~\cite{La}, extremal classes~\cite{BoRa,Moran}, and ample
classes~\cite{BaChDrKo,Dr}.  Lawrence~\cite{La} was the first to
define them for the investigation of the possible sign patterns
realized by points of a convex set of $\R^d$.  Interestingly,
Lawrence's definition of these classes does not use the notion of
shattering nor the Sandwich Lemma.  In this context, these classes
were discovered by Bollob\'{a}s and Radcliffe~\cite{BoRa} and Bandelt
et al.~\cite{BaChDrKo}, and the equivalence between the two
definitions appears in~\cite{BaChDrKo}.  Ample classes admit a
multitude of combinatorial and geometric characterizations
\cite{BaChDrKo,BaChDrKo_geo,BoRa,La} and comprise many natural
examples arising from discrete geometry, combinatorics, graph theory,
and geometry of groups~\cite{BaChDrKo,La}.

\subsection{Main Results}
	
\subsubsection{Corner Peelings.}
A \emph{corner} in an ample class $C$ is any concept $c \in C$ that
belongs to a unique maximal cube of $C$ (equivalently, $c$ is a corner
if $C\setminus\{c\}$ is also ample, see Lemma~\ref{ample_bis}).  A
sequence of corner removals leading to a single concept is called a
\emph{corner peeling}; corner peeling is a strong version of
\emph{collapsibility}.  Wiedemann~\cite{Wiedemann} and independently
Chepoi (unpublished, 1996) asked whether every ample class has a
corner.  The machine learning community studied this question
independently in the context of \emph{sample compression schemes} for
maximum classes: Rubinstein and Rubinstein~\cite{RuRu} showed that
corner peelings lead to optimal \emph{unlabeled sample compression
  schemes (USCS)}.
% Kuzmin and Warmuth~\cite{KuWa} showed in 2007 that
% corner peelings lead to optimal \emph{sample compression} (see below)
% and conjectured that every maximum class has a corner.

In Theorem~\ref{thm:no_corner_bis} we refute this conjecture.  The
crux of the proof is an equivalence between corner peelings and
partial shellings of the cross-polytope. This equivalence translates
the question whether corners always exist to the question whether
partial shellings can always be extended. The latter was an open
question in Ziegler's book on polytopes~\cite{Zi}, and was resolved in
Tracy Hall's PhD thesis~\cite{Hall} where an interesting
counterexample is presented.  The ample class resulting from Hall's
construction yields a maximum class without corners.
	
\subsubsection{Sample Compression.}  Sample compression is a powerful
technique to derive generalization bounds in statistical learning.
Littlestone and Warmuth~\cite{LiWa} introduced it and asked if every
class of VC-dimension $d<\infty$ has a sample compression scheme of a
finite size.  This question was later relaxed by Floyd and
Warmuth~\cite{FlWa,Wa} to \emph{whether a sample compression scheme of
  size~$O(d)$ exists}.  The first question was recently resolved
by~\cite{MoYe} who exhibited an $\exp(d)$ sample compression scheme.  The
second question however remains one of the oldest open problems in
machine learning (for more background we refer the reader
to~\cite{MoWa} and the books~\cite{ShaBen,Wig}).

Rubinstein and Rubinstein~\cite[Theorem~16]{RuRu} showed that the
existence of a corner peeling for a maximum class $C$ implies a
\emph{representation map} for $C$ (see Section~\ref{s:ample} for a
definition), which is known to yield an optimal unlabeled sample
compression scheme of
size~$\vcdim(C)$~\cite{KuWa}.\footnote{P\'alv\"olgyi and
  Tardos~\cite{PaTa} recently exhibited a (non-ample) class $C$ with
  no USCS of size $\vcdim(C)$.}  They claim, using an interesting
topological approach, that maximum classes admit corner
peelings. Unfortunately, our Theorem~\ref{thm:no_corner_bis} shows
that this does not hold.

% Kuzmin and Warmuth~\cite{KuWa} established an interesting connection
% between compression schemes and corner peelings for maximum classes
% $C$.  In particular, they showed that a corner peeling implies a
% \emph{representation map} (see Section~\ref{s:ample} for a definition)
% and that a representation map implies an unlabeled sample compression
% scheme of size~$\vcdim(C).  $\footnote{P\'alv\"olgyi and
%   Tardos~\cite{PaTa} recently exhibited a (not ample) class $C$ with
%   no unlabeled sample compression scheme of size $\vcdim(C)$.}
	
While our Theorem~\ref{thm:no_corner_bis} rules out the program of
deriving representation maps from corner peelings, in
Theorem~\ref{th-maximum} we provide an alternative derivation of
representation maps for maximum classes and therefore also of
unlabeled sample compression schemes for them.

\subsubsection{Sample Compression and Unique Sink Orientations.}
We next turn to construction of representation maps for ample classes.
In Theorem~\ref{t:local-to-global} we present geometric
characterizations of such maps via \emph{unique sink orientations}: an
orientation of the edges of a cube $B$ is a \emph{unique sink
  orientation (USO)} if any subcube $B'\subseteq B$ has a unique sink.
Szab\'{o} and Welzl~\cite{SzWe} showed that any USO of $B$ leads to a
representation map for $B$.  We extend this bijection to ample classes
$C$ by proving that representation maps are equivalent to orientations
$r$ of $C$ such that (i) $r$ is a USO on each subcube $B\subseteq C$,
and (ii) for each $c\in C$ the edges outgoing from $c$ belong to a
subcube $B\subseteq C$.  We further show that any ample class admits
orientations satisfying each one of those conditions.  However,
\emph{the question whether all ample classes admit representation maps
  remains open.}

\subsubsection{Implications on Previous Works.}

Our Theorem~\ref{thm:no_corner_bis} establishes the existence of
maximum classes without any corners, thus countering several previous
results in machine learning:
\begin{itemize}
% \item Rubinstein and Rubinstein~\cite{RuRu} used an interesting
%   topological approach to argue that maximum classes admit a corner.
%   This is unfortunately false, as witnessed by
%   Theorem~\ref{thm:no_corner_bis}.
\item Rubinstein and Rubinstein~\cite[Theorem 32]{RuRu} showed that
  any maximum class can be represented by a simple arrangement of
  piecewise-linear hyperplanes. In~\cite[Theorem 39]{RuRu}, they claim
  that sweeping such an arrangement leads to a corner peeling of the
  corresponding maximum class.  This is unfortunately false, as
  witnessed by Theorem~\ref{thm:no_corner_bis}.
\item Kuzmin and Warmuth~\cite{KuWa} constructed unlabeled sample
  compression schemes for maximum classes based on the presumed
  uniqueness of a certain matching (their Theorem 10).  This theorem
  is wrong (as explained in Section~\ref{KuWaerror}) as it implies the
  existence of corners and Hall's counterexample does not have
  corners.  However their conclusion is correct: In our
  Theorem~\ref{th-maximum} we show that such unlabeled compression
  schemes always exist based on a different construction and proof
  method.
\item Theorem 3 by Samei, Yang, and Zilles~\cite{SaYaZi} is built on a
  generalization of Theorem 10 from~\cite{KuWa} to the multiclass case
  which is also incorrect.
\item Theorem 26 by Doliwa et al.~\cite{DoFaSiZi} uses the result
  in~\cite{RuRu} to show that the Recursive Teaching Dimension (RTD)
  of maximum classes equals to their VC dimension.  However the VC
  dimension $3$ maximum class from Theorem~\ref{thm:no_corner_bis} has
  RTD at least~$4$.  It remains open whether the RTD of every maximum
  class $C$ is bounded by~$O(\vcdim(C))$.
\end{itemize}

\subsubsection{An Optimal Proper PAC Learner for Maximum Classes.}
In a recent work, Bousquet, Hanneke, Moran, and Zhivotovskiy~\cite{Bousquet20stablecomp}
showed that a special type of sample compression schemes, termed {\it stable compression schemes},
achieve the optimal learning rate in PAC learning. They further noticed
that any sample compression scheme which is defined by a representation map is stable.
Thus, using the compression scheme constructed in this paper, Bousquet et al.\ conclude that
every maximum class can be properly learned by an algorithm achieving the optimal learning rate.

\subsection{Organization.}
Section~\ref{s:prel} presents the main definitions and
notations. Section~\ref{s:ample} reviews characterizations of
ample/maximum classes and characteristic examples.
Section~\ref{s:shell} demonstrates the existence of the maximum class
$C_H$ without corners.  Section~\ref{s:maximum} establishes the
existence of representation maps for maximum classes.
Section~\ref{s:rep_maps} establishes a bijection between
representation maps and unique sink orientations for ample classes.

\section{Preliminaries}\label{s:prel}

A \emph{concept class} $C$ is a set of subsets (concepts) of a finite
ground set ${\U}$ which is called the \emph{domain} of $C$ and denoted
$\dom(C)$.  We sometimes treat the concepts as characteristic
functions rather than subsets.  The \emph{support} (or \emph{dimension
  set}) $\supp(C)$ of $C$ is the set
$\{x\in\U : x\in c'\setminus c''\text{ for some }c',c''\in C\}$.

Let $C$ be a concept class of $2^{\U}$. The \emph{complement} of $C$
is $C^*:=2^{\U}\setminus C$. The \emph{twisting} of $C$ with respect
to $Y\subseteq {\U}$ is the concept class
$C\Delta Y=\{ c\Delta Y: c\in C\}$.
The \emph{restriction} of a concept $c \in C$ on
$Y\subseteq {\U}$ is the concept $c|Y=c\cap Y$.
The \emph{restriction} of $C$ on
$Y\subseteq {\U}$ is the class $C|Y=\{ c|Y : c\in C\}$ whose domain
is $Y$. We use $C_Y$ as shorthand for $C|({\U}\setminus Y)$; in
particular, we write $C_x$ for $C_{\{ x\}}$ (see
Figure~\ref{fig-ex-red-restr} for an example), and $c_x$ for
$c|(\U\setminus\{x\})$ for $c\in C$ (note that $c_x\in C_x$).  A
concept class $B\subseteq 2^{\U}$ is a \emph{cube} if there exists
$Y\subseteq {\U}$ such that $B|Y = 2^Y$ and~$B_Y$ contains a single
concept (denoted by $\tg(B)$).  Note that $\supp(B)=Y$ and therefore
we say that $B$ is a $Y$-\emph{cube}; $\lvert Y\rvert$ is called the
\emph{dimension} $\dim(B)$ of $B$. Two cubes $B,B'$ with the same
support are called \emph{parallel cubes}. A cube $B$ is \emph{maximal}
if there is no cube $B'$ such that $B \subsetneq B'$.
	
Let $Q_n$ denote the $n$-dimensional cube where $n = \lvert\U\rvert$;
$c,c' \in Q_n$ are called adjacent if the symmetric difference
$c\Delta c'$ is of size $1$. The \emph{1-inclusion graph} of $C$ is
the subgraph $G(C)$ of $Q_n$ induced by the vertex-set $C$ when the
concepts of $C$ are identified with the corresponding vertices of
$Q_n$.  Any cube $B\subseteq C$ is called \emph{a cube of $C$}.  The
\emph{cube complex} of $C$ is the set
$Q(C)=\{B: B\text{ is a cube of }C\}$. The cubes of $C$ are called the
\emph{faces} of $Q(C)$ and the maximal cubes of $C$ are called the
facets of $Q(C)$.  The \emph{dimension} $\dim(Q(C))$ of $Q(C)$ is the
largest dimension $\max_{B\in Q(C)} \dim(B)$ of a cube of $Q(C)$.  A
concept $c\in C$ is called a \emph{corner} of $C$ if $c$ belongs to a
unique maximal cube of $C$.

\begin{figure}
  \includegraphics[scale=0.75,page=9]{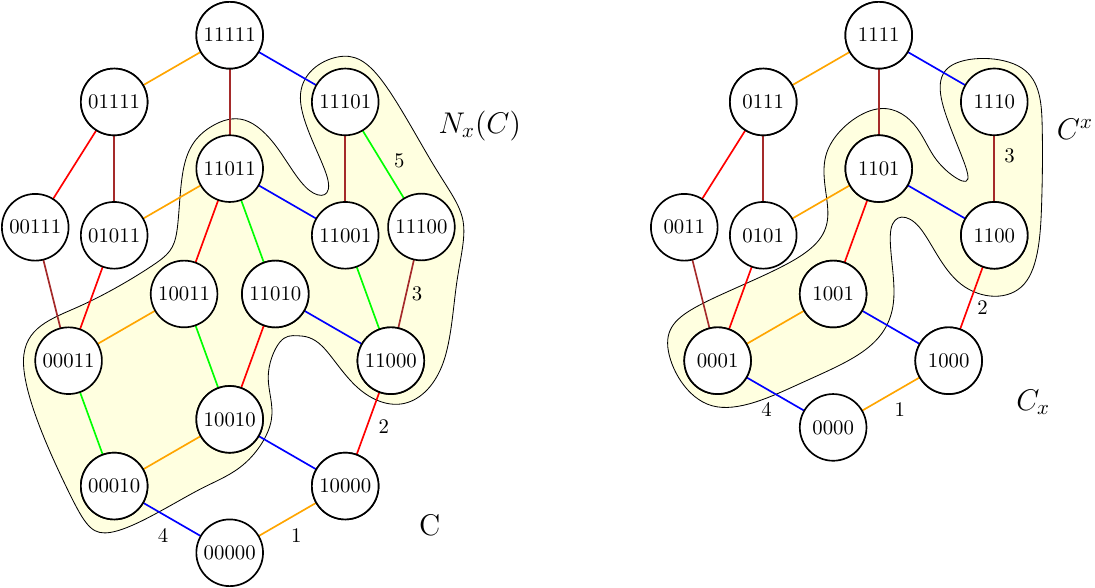}%
  \caption{A $2$-dimensional maximum class
    $C \subseteq 2^{\{1,2,3,4,5\}}$ on the left and the restriction
    $C_x$ for $x = 5$ on the right. The reduction $C^x$ corresponds to
    the restriction of the carrier $N_x(C)$.}%
  \label{fig-ex-red-restr}
\end{figure}

The \emph{reduction} $C^Y$ of a concept class $C$ to $Y\subseteq {\U}$
is a concept class on ${\U}\setminus Y$ which has one concept for each
$Y$-cube of $C$:
$C^Y:=\{ \tg(B): B\in Q(C) \text{ and } \supp(B)=Y\}$.  When $x\in \U$
we denote $C^{\{x\}}$ by $C^x$ and call it \emph{the $x$-hyperplane of
  $C$} (see Figure~\ref{fig-ex-red-restr} for an example). Note that a
concept $c$ belongs to $C^x$ if and only if $c$ and $c \cup \{x\}$
both belong to $C$. The union of all cubes of~$C$ having $x$ in their
support is called the \emph{carrier} of $C^x$ and is denoted by
$N_x(C)$.  If $c \in N_x(C)$, we also denote $c|(\U\setminus\{x\})$ by
$c^x$ (note that $c^x\in C^x$).

The \emph{tail} $\tail_x(C)$ of a concept class $C$ on dimension $x$
consists of all concepts that do not have in $G(C)$ an incident edge
labeled with $x$.  They correspond to the concepts of
$C_x\setminus C^x$, i.e., to the concepts of $C_x$ that have a unique
extension in $C$.  The class $C$ can be partitioned as
$N_x(C) \cupdot \tail_x(C) = 0C^x\cupdot 1C^x\cupdot \tail_x(C)$,
where $\cupdot$ denotes the disjoint union and $bC^x$ consists of all
concepts in $C^x$ extended with bit $b$ in dimension $x$.

Given two classes $C \subseteq 2^{\U}$ and $C' \subseteq 2^{\U'}$
where $\U$ and $\U'$ are disjoint, the \emph{Cartesian product}
$C \times C' \subseteq 2^{\U \cupdot \U'}$ is the concept class
$\{c \cupdot c' : c \in C \text{ and } c' \in C'\}$.

A concept class $C$ is \emph{connected} if the graph $G(C)$ is
connected.  If $C$ is connected, denote by $d_{G(C)}(c,c')$ the
graph-distance between $c$ and $c'$ in $G(C)$ and call it the
\emph{intrinsic distance} between $c$ and $c'$.  The distance
$d(c,c'):=d_{Q_n}(c,c')$ between two vertices $c,c'$ of $Q_n$
coincides with the Hamming distance $\lvert c\Delta c'\rvert$ between
the 0-1-vectors corresponding to $c$ and $c'$.  Let
$B(c,c')=\{ t\subseteq {\U}: d(c,t)+d(t,c')=d(c,c')\}$ be the
\emph{interval} between $c$ and $c'$ in $Q_n$ ; equivalently,
$B(c,c')$ is the smallest cube of $Q_n$ containing $c$ and $c'$.  A
connected concept class $C$ is called \emph{isometric} if
$d(c,c')=d_{G(C)}(c,c')$ for any $c,c'\in C$ and \emph{locally
  isometric} if $d(c,c')=d_{G(C)}(c,c')$ for any $c,c' \in C$ such
that $d(c,c')\leq 2$.  Any path of $C^Y$ connecting two concepts
$\tg(B)$ and $\tg(B')$ of $C^Y$ can be lifted to a path of parallel
$Y$-cubes connecting $B$ and $B'$ in $C$; such a path of cubes is
called a \emph{gallery}.

A \emph{simplicial complex} $X$ on a set $\U$ is a family of subsets
of $X$, called \emph{simplices} or \emph{faces} of $X$, such that if
$\sigma \in X$ and $\sigma' \subseteq \sigma$, then $\sigma' \in
X$. The \emph{facets} of $X$ are the maximal (by inclusion) faces of
$X$. The \emph{dimension} $d$ of $X$ is the size of its largest face.
 A simplicial complex $X$ is a
\emph{pure} simplicial complex of dimension $d$ if all its facets have
size $d$.

A set $Y\subseteq {\U}$ is \emph{shattered} by a concept class
$C\subseteq 2^{\U}$ if $C|Y=2^Y$. Furthermore, $Y$ is \emph{strongly
  shattered} by $C$ if $C$ contains a $Y$-cube.  Denote by $\oX(C)$
and $\uX(C)$ the simplicial complexes consisting respectively of all
shattered and of all strongly shattered sets of $C$.  Clearly,
$\uX(C)\subseteq \oX(C)$ and both $\oX(C)$ and $\uX(C)$ are closed by
taking subsets, i.e., $\oX(C)$ and $\uX(C)$ are simplicial complexes.
The classical \emph{VC-dimension}~\cite{VaCh} $\vcdim(C)$ of a concept
class $C$ is the size of the largest set shattered by $C$, i.e., the
dimension of the simplicial complex $\oX(C)$.\footnote{Note that the
  usual definition of the dimension of a simplicial complex $X$ is the
  size of the largest face of $X$ minus 1. We adopted this convention
  to have an equality between the VC-dimension of a class $C$ and the
  dimension of $\oX(C)$.} The fundamental \emph{sandwich lemma}
(rediscovered independently by Pajor~\cite{Pa}, Bollob\'{a}s and
Radcliffe~\cite{BoRa}, Dress~\cite{Dr}, and Anstee et
al.~\cite{AnRoSa}) asserts that
$\lvert\uX(C)\rvert \le \lvert C\rvert \le \lvert \oX(C)\rvert$.  If
$d=\vcdim(C)$ and $n=\lvert {\U}\rvert$, then $\oX(C)$ cannot contain
more than $\Phi_d(n):=\sum_{i=0}^d \binom{n}{i}$ simplices, yielding
the well-known \emph{Sauer-Shelah-Perles
  lemma}~\cite{Sauer,Shelah,VaCh} that
$\lvert C \rvert \le \Phi_d(n)$.

A \emph{labeled sample} is a set $s=\{(x_1,y_1),\ldots,(x_m,y_m)\}$,
where $x_i\in \U$ and $y_i\in\{0,1\}$.  An \emph{unlabeled sample} is
a set $\{x_1,\ldots,x_m\}$, where $x_i\in \U$.  A subsample $s'$ of a
sample $s$ (labeled or unlabeled) is a subset of $s$.  Given a labeled
sample $s=\{(x_1,y_1),\ldots,(x_m,y_m)\}$, the unlabeled sample
$\{x_1,\ldots,x_m\}$ is the domain of $s$ and is denoted by $\dom(s)$.
A labeled sample $s$ is \emph{realizable} by a concept
$c:\U\to\{0,1\}$ (seen as a map) if $c(x_i)=y_i$ for every $i$, and
$s$ is realizable by a concept class $C$ if it is realizable by some
$c\in C$. For a concept class $C$, let $\RS(C)$ be the set of all
labeled samples realizable by $C$.

A \emph{sample compression scheme} for a concept class $C$ is best
viewed as a protocol between a \emph{compressor} $\alpha$ and a
\emph{reconstructor} $\beta$ (which both depend of $C$).  The
compressor gets a labeled sample $s$ realizable by $C$ from which it
picks a small subsample~$s'$.  The compressor sends $s'$ to the
reconstructor.  Based on $s'$, the reconstructor outputs a concept
$c \in C$ that needs to be consistent with the entire input
sample~$s$.  A sample compression scheme has size $k$ if for every
realizable input sample~$s$ the size of the compressed subsample $s'$
is at most $k$.  An unlabeled sample compression scheme is a sample
compression scheme in which the compressed subsample $s'$ is
unlabeled. So, the compressor removes the labels before sending the
subsample to the reconstructor. An \emph{unlabeled sample compression
  scheme} of size $k$ for a concept class $C \subseteq 2^{\U}$ is thus
defined by a (compressor) function
$\alpha: \RS(C) \to \binom{\U}{\leq k}$ and a (reconstructor) function
$\beta: \Ima(\alpha) := \alpha(\RS(C)) \to C$ such that for any
realizable sample $s$ of $C$, the following conditions hold:
$\alpha(s) \subseteq \dom(s)$ and $\beta(\alpha(s))|\dom(s)= s$.

In the literature, one usually allows the reconstructor $\beta$ to
take values in $2^U$, i.e., the reconstructor can return a subset that
is not a concept of $C$. The unlabeled sample compression schemes we
consider in this paper, i.e., satisfying the property that
$\Ima(\beta) \subseteq C$, are usually called \emph{proper} unlabeled
sample compression schemes.

\section{Ample and Maximum Classes}
\label{s:ample}

In this section, we briefly review the main characterizations and the
basic examples of ample classes (maximum classes being one of them).

\subsection{Characterizations}
A concept class $C$ is called \emph{ample} if
$\lvert C\rvert=\lvert \oX(C)\rvert$. Ample classes are closed by
taking restrictions, reductions, intersections with cubes, twistings,
complements, and Cartesian products.

The following theorem reviews the main combinatorial characterizations
of ample classes:

\begin{theorem}[\!\!\cite{BaChDrKo,BoRa,La}] \label{thm:ample1}
  The following conditions are equivalent for a class $C$:
  \begin{enumerate}[(1)]
  \item $C$ is ample;
  \item $C^*$ is ample;
  \item $\uX(C)=\oX(C)$;
  \item $\lvert \uX(C)\rvert=\lvert C\rvert$;
  \item $\lvert \oX(C)\rvert=\lvert C\rvert$;
  \item $C\cap B$ is ample for any cube $B$;
  \item $(C^Y)_Z=(C_Z)^Y$ for all partitions~${\U}=Y\cupdot Z$;
  \item for all partitions ${\U}=Y\cupdot Z$, either $Y\in \uX(C)$ or
    $Z\in \uX(C^*)$.\footnote{This is the original definition of
      lopsidedness by Lawrence~\cite{La}.}
  \end{enumerate}
\end{theorem}

Condition (3) leads to a simple definition of ampleness: $C$ is ample
if whenever $Y\subseteq {\U}$ is shattered by $C$, then there is a
$Y$-cube in $C$. Thus, if $C$ is ample we will write $\uoX(C)$ instead
of $\uX(C)=\oX(C)$.  It follows that for ample classes, the
VC-dimension of a concept class $C$, the dimension of the simplicial
complex $X(C)$, and the dimension of the cube complex of $C$ are the
same. In the following, we talk about a $d$-dimensional class $C$ when
these three dimensions are equal to $d$.

We continue with metric and recursive characterizations of ample
classes:

\begin{theorem}[\!\!\cite{BaChDrKo}]\label{thm:ample2}
The following are equivalent for a concept class $C$:
\begin{enumerate}[(1)]
\item $C$ is ample;
\item $C^Y$ is connected for all $Y\subseteq {\U}$;
\item $C^Y$ is isometric for all $Y\subseteq {\U}$;
\item $C$ is isometric, and both $C_x$ and $C^x$ are ample for all
  $x\in {\U}$;
\item $C$ is connected and all hyperplanes $C^x$  are ample.
\end{enumerate}
\end{theorem}

\begin{corollary}\label{cor:maxcub}
  Two maximal cubes of an ample class $C$ have different supports.
\end{corollary}

\begin{proof}
  Indeed, if $B$ and $B'$ are two $d$-cubes with the same support, by
  Theorem~\ref{thm:ample2}(2) $B$ and $B'$ can be connected in $C$ by
  a gallery, and thus $B$ is contained in a $d+1$-cube. Therefore, $B$
  and $B'$ cannot be maximal.
\end{proof}

\begin{lemma}\label{lem:cubecontraction}
  Given an ample class $C$ and $x \in U$, for any cube $B$ of $C_x$,
  there exists a cube $B'$ of $C$ such that $\supp(B') = \supp(B)$ and
  $B'_x = B$.
\end{lemma}

\begin{proof}
  Consider the cube $B^*$ of $2^U$ such that
  $\supp(B^*) = \supp(B) \cup \{x\}$ and $B^*_x = B$. By
  Theorem~\ref{thm:ample1}(6), $C \cap B^*$ is ample. Since $\supp(B)$
  is shattered by $C \cap B^*$, there exists a cube $B'$ in
  $C \cap B^*$ such that $\supp(B') = \supp(B)$. Since
  $B'_x = B^*_x = B$, we are done.
\end{proof}

A concept class $C\subseteq 2^{\U}$ of VC-dimension $d$ is
\emph{maximum} if
$\lvert C\rvert=\Phi_d(n)=\sum_{i=0}^d \binom{n}{i}$, i.e., if
$ C=\bigcup_{i=0}^d \binom{[n]}{i}$ (where $n = |\U|$).  The Sandwich
Lemma and Theorem~\ref{thm:ample1}(5) imply that maximum classes are
ample. Analogously to ample classes, maximum classes are hereditary by
taking restrictions, reductions, twistings, and complements. Basic
examples of maximum classes are concept classes derived from
arrangements of hyperplanes in general position, balls in $\R^n$, and
unions of $n$ intervals on the line~\cite{GaWe,Fl,FlWa,Jo}.  The
following theorem summarizes some characterizations of maximum
classes:

\begin{theorem}[\!\!\cite{GaWe,Fl,FlWa,Welzl}]\label{thm:maximum-char}
  The following conditions are equivalent for a concept class $C$:
  \begin{enumerate}[(1)]
  \item $C$ is maximum;
  \item $C_Y$ is maximum for all $Y\subseteq {\U}$;
  \item $C_x$ and $C^x$ are maximum for all $x\in {\U}$
  \item $C^*$ is maximum.
  \end{enumerate}
\end{theorem}

Following Kuzmin and Warmuth~\cite{KuWa}, we define a
\emph{representation map} for an ample class $C$ as a bijection
$r:C\rightarrow \uoX(C)$ satisfying the \emph{non-clashing condition}:
$c|(r(c)\cup r(c'))\ne c'|(r(c)\cup r(c'))$, for all
$c,c'\in C, c\ne c'$.  It was shown in~\cite{KuWa} that the existence
of a representation map for a maximum class $C$ implies an unlabeled
sample compression scheme of size~$\vcdim(C)$ for $C$. In
Section~\ref{s:rep_maps}, we show that this also holds for ample
classes. Moreover we show that for ample classes, they are equivalent
to $\Delta$-representation maps defined as follows.  A
\emph{$\Delta$-representation map} for an ample class $C$ is a
bijection $r:C\rightarrow \uoX(C)$ satisfying the
\emph{$\Delta$-non-clashing condition}:
$c|(r(c)\Delta r(c'))\ne c'|(r(c)\Delta r(c'))$, for all
$c,c'\in C, c\ne c'$.

\subsection{Examples}
We continue with the main examples of ample classes.

\subsubsection{Simplicial Complexes}
The set of characteristic functions of simplices of a simplicial
complex $S$ can be viewed as a concept class $C(S)$: $C(S)$ is a
bouquet of cubes with a common origin $\varnothing$, one cube for each
simplex of $S$. Therefore, $\uX(C(S))=S$ and since
$\lvert S\rvert=\lvert C(S)\rvert$, $C(S)$ is an ample class having
$S$ as its simplicial complex~\cite{BaChDrKo}.

\subsubsection{Realizable Ample Classes.}
Let $K\subseteq \R^n$ be a convex set. Let
$C(K):= \{\sign(v) : v\in K, v_i\neq 0, \forall i\leq n\}$, where
$\sign(v)\in\{\pm 1\}^n$ is the sign pattern of $v$.
Lawrence~\cite{La} showed that $C(K)$ is ample, and called ample
classes representable in this manner \emph{realizable}. Lawrence
presented a non-realizable ample class of $Q_9$ arising from a
non-stretchable arrangement of pseudolines. It is shown
in~\cite{BaChDrKo_geo} that any ample class becomes realizable if
instead of a convex set $K$ one considers a Menger $\ell_1$-convex set
$K$ of $\R^n$.

\subsubsection{Median Classes.}
A class $C$ is called \emph{median} if for every three concepts
$c_1,c_2,c_3$ of $C$ their \emph{median}
$m(c_1,c_2,c_3):=(c_1\cap c_2)\cup (c_1\cap c_3)\cup (c_2\cap c_3)$
also belongs to $C$.  Median classes are ample by~\cite[Proposition
2]{BaChDrKo}. Median classes are closed by taking reductions,
restrictions, intersections with cubes, and products but not under
complementation.

Due to their relationships with other discrete structures, median
classes are one of the most important examples of ample classes.
Median classes are equivalent to finite median graphs (a well-studied
class in metric graph theory, see~\cite{BaCh_survey}), to CAT(0) cube
complexes, i.e., cube complexes of global nonpositive curvature
(central objects in geometric group theory,
see~\cite{Gromov,Sa_survey}), and to the domains of event structures
(a basic model in concurrency theory~\cite{NiPlWi,Winskel}).

\subsubsection{Convex Geometries and Conditional Antimatroids.}
Let $C$ be a concept class such that (i) $\varnothing\in C$ and (ii)
$c,c'\in C$ implies that $c\cap c'\in C$.  A point $x\in c\in C$ is
called \emph{extremal} if $c\setminus \{ x\}\in C$.  The set of
extremal points of $c$ is denoted by $\ex(c)$.  A concept $c\in C$ is
\emph{generated} by $s\subseteq c$ if $c$ is the smallest concept of
$C$ containing $s$.  A concept class $C$ satisfying (i) and (ii) with
the additional property that every concept $c$ of $C$ is generated by
its extremal points is called a \emph{conditional
  antimatroid}~\cite[Section 3]{BaChDrKo}. If ${\U}\in C$, then we
obtain the well-known structure of a \emph{convex geometry} (called
also an \emph{antimatroid})~\cite{EdJa} (See
Figure~\ref{fig-ample-convex} for an example). It was shown
in~\cite[Proposition 1]{BaChDrKo} that if $C$ is a conditional
antimatroid, then $\uX(C)=\oX(C)$, since $\uX(C)$ coincides with the
sets of extremal points and $\oX(C)$ coincides with the set of all
minimal generating sets of sets from $C$. Hence, any conditional
antimatroid is ample. Besides convex geometries, median classes are
also conditional antimatroids.  Another example of conditional antimatroids
is given by the set $C$ of all strict partial orders on a set
$M$. Each partial order is an asymmetric, transitive subset of
${\U}=\{ (u,v): u,v\in M, u\ne v\}$. Then it is shown
in~\cite{BaChDrKo} that for any $c\in C$, $\ex(c)$ is the set of
covering pairs of $c$ (i.e., the pairs $(u,v)$ such that $u < v$ and
there is no $w$ with $u < w < v$) and that
\begin{equation*}
  \oX(C)=\uX(C)=\{H\subseteq {\U}: H \text{ is the Hasse diagram of a
    partial order on } M\}.
\end{equation*}
Convex geometries comprise many interesting and important examples
from geometry, ordered sets, and graphs, see the foundational
paper~\cite{EdJa}.  For example, by the Krein-Milman theorem, any polytope
of $\R^n$ is the convex hull of its extremal points.  A
\emph{realizable convex geometry} is a convex geometry $C$ such that
its point set ${\U}$ can be realized as a finite set of $\R^n$ and
$c\in C$ if and only if $c$ is the intersection of a convex set of
$\R^n$ with ${\U}$. Acyclic oriented geometries (acyclic oriented
matroids with no two point circuits) are examples of convex
geometries, generalizing the realizable ones.

\begin{figure}
  \includegraphics[scale=0.85,page=8]{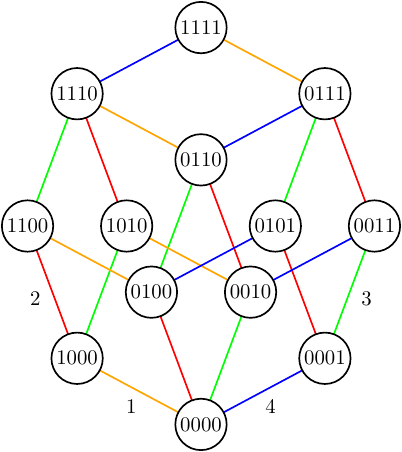}%
  \caption{An ample class which is also a convex geometry}%
  \label{fig-ample-convex}
\end{figure}

We continue with two particular examples of conditional antimatroids.

\begin{example}
  Closer to usual examples from machine learning, let ${\U}$ be a
  finite set of points in $\R^n$, no two points sharing the same
  coordinate, and let the concept class $C_{\Pi}$ consist of all
  intersections of axis-parallel boxes of $\R^n$ with ${\U}$. Then
  $C_{\Pi}$ is a convex geometry: for each $c\in C_{\Pi}$, $\ex(c)$
  consists of all points of $c$ minimizing or maximizing one of the
  $n$ coordinates.  Clearly, for any $p\in \ex(c)$, there exists a box
  $\Pi$ such that $\Pi\cap {\U}=c\setminus \{ p\}$.
\end{example}

\begin{example}
  A \emph{partial linear space} is a pair $(P,L)$ consisting of a
  finite set $P$ whose elements are called \emph{points} and a family
  $L$ of subsets of $P$, whose elements are called \emph{lines}, such
  that any line contains at least two points and any two points belong
  to at most one line. The projective plane (any pair of points belong
  to a common line and any two lines intersect in exactly one point)
  is a standard example, but partial linear spaces comprise many more
  examples.  The concept class $L \subseteq 2^P$ has VC-dimension at
  most 2 because any two points belong to at most one line.  Now, for
  each line $\ell\in L$ fix an arbitrary total order $\pi_{\ell}$ of
  its points. Let $L^*$ consist of all subsets of points that belong
  to a common line $\ell$ and define an interval of $\pi_{\ell}$. Then
  $L^*$ is still a concept class of VC-dimension 2. Moreover, $L^*$ is
  a conditional antimatroid: if $c\in L^*$ and $c$ is an interval of
  the line $\ell$, then $\ex(c)$ consists of the two end-points of $c$
  on $\ell$.
\end{example}

\subsubsection{Ample Classes from Graph Orientations.}
Kozma and Moran~\cite{KoMo} used the sandwich lemma to derive several
properties of graph orientations. They also presented two examples of
ample classes related to distances and flows in networks (see
also~\cite[p.157]{La} for another example of a similar nature). Let
$G=(V,E)$ be an undirected simple graph and let $o^*$ be a fixed
reference orientation of $E$. To an arbitrary orientation $o$ of $E$
associate a concept $c_o\subseteq E$ consisting of all edges which are
oriented in the same way by $o$ and by $o^*$. It is proven
in~\cite[Theorem 26]{KoMo} that if each edge of $G$ has a non-negative
capacity, a source $s$ and a sink $t$ are fixed, then for any
flow-value $w\in \R^+$, the set $C^{\text{flow}}_w$ of all
orientations of $G$ for which there exists an $(s,t)$-flow of value at
least $w$ is an ample class. An analogous result was obtained if
instead of the flow between $s$ and $t$ one considers the distance
between those two nodes.

\section{Corner Peelings and Partial Shellings}\label{s:shell}

In this section, we prove that corner peelings of ample classes are
equivalent to isometric orderings of $C$ as well as to partial
shellings of the cross-polytope.  This equivalence, combined with a
result by Hall~\cite{Hall} yields a maximum class with VC dimension 3
without corners (Theorem~\ref{thm:no_corner_bis} below). We show that
this counterexample also refutes the analysis of the \emph{Tail
  Matching Algorithm} of Kuzmin and Warmuth~\cite{KuWa} for
constructing unlabeled sample compression schemes for maximum
classes. On the positive side, we prove the existence of corner
peelings for conditional antimatroids and 2-dimensional ample
classes. Finally we show that the cube complexes of all ample classes
are collapsible.

\subsection{Corners, Isometric Orderings, and Partial Shellings}

We first establish some properties satisfied by the corners of an
ample class. For~$t\notin C$, let $F[t]$ be the smallest cube of $Q_n$
containing $t$ and all neighbors of $t$ in $Q_n$ that are in $C$. Note
that the dimension of $F[t]$ is the number of neighbors of $t$
in~$G(C)$.

\begin{lemma} \label{ample_bis} Let $C$ be ample. Then:
  \begin{enumerate}[(i)]
  \item If $t\notin C$ then $F[t]\subseteq C\cup\{t\}$.
  \item If $t\notin C$ and $C':=C\cup \{ t\}$ is isometric then $C'$
    is ample and $t$ is a corner of $C'$.
  \item $c$ is a corner of~$C$ if and only if $C\setminus\{ c\}$ is ample.
  \end{enumerate}
\end{lemma}

\begin{proof}
  {Item (i):} Suppose $F[t]\setminus C\neq\{ t\}$. Pick $s\neq t$ that
  is closest to $t$ in $F[t]\setminus C$ (with respect to the Hamming
  distance of $Q_n$). Then $t$ and $s$ are not adjacent (by the
  definition of $F[t]$).  By the choice of $s$,
  $B(s,t)\setminus \{ s,t\}\subseteq C$, i.e.,
  $B(s,t)\cap C^*=\{t,s\}$. However, by Theorem~\ref{thm:ample1},
  $C^*$ is ample and thus isometric by Theorem~\ref{thm:ample2},
  contradicting that $B(s,t)\cap
  C^*=\{t,s\}$. % contrary to the ampleness of $C^*$.

  \smallskip\noindent{Item (ii):} To prove that $C'$ is ample, we use
  Theorem~\ref{thm:ample2}(2).  First note that by
  item~(i),~${F}[t]\subseteq C'$.  Let $F'\neq F''$ be parallel cubes
  of $C'$.  If $t\notin F'\cup F''$, then a gallery connecting $F'$
  and $F''$ in $C$ is a gallery in $C'$.  So, assume $t\in F'$. If
  $F'$ is a proper face of $F[t]$, then $F'$ is parallel to a face $F$
  of $F[t]$ not containing $t$.  Since $F'$ and $F$ are connected in
  $F[t]$ by a gallery and $F$ and $F''$ are connected in $C$ by a
  gallery, we obtain a gallery between $F'$ and $F''$ in $C'$.
  Finally, let $F'=F[t]$. In this case, we assert that $F''$ does not
  exist. Otherwise, we define a parallelism map $\pi$ between the
  concepts of $F'$ and the concepts of $F''$ as follows: for any
  $c' \in F'$, $\pi(c')$ is the unique concept $c'' \in F''$ such that
  $c'|\supp(F') = c''|\supp(F'')$ (recall that
  $\supp(F') = \supp(F'')$).  Note that for any $r\in F'$:
  $d(t,\pi(t))=d(r,\pi(r))=d(F',F'')$.  Since $C'$ is isometric, $t$
  and $\pi(t)$ can be connected in $C'$ by a path $P$ of length
  $d(t,\pi(t))$.  Let~$s$ be the neighbor of $t$ in $P$. Since
  $s\in C$ it follows that $s\in F[t]=F'$.  So, $s$ is a concept in
  $F'$ that is closer to $\pi(t)$ than $t$; this contradicts that
  $d(t,\pi(t))=d(F',F'')$.

  \smallskip\noindent {Item (iii):} If $c\in C$ is a corner then there
  is a unique maximal cube~$F\subseteq C$ containing it.  Combined
  with Corollary~\ref{cor:maxcub}, this implies that
  $\uX(C\setminus\{c\}) = \uX(C)\setminus\{\supp(F)\}$.  Next, since
  $\lvert C\rvert=\lvert \uX(C)\rvert$, we get that
  $\lvert C\setminus\{c\}\rvert = \lvert \uX(C\setminus\{c\})\rvert$,
  and by Theorem~\ref{thm:ample1} $C\setminus\{c\}$ is
  ample. Conversely, if $C \setminus \{c\}$ is ample, applying Item
  (ii) to $C \setminus \{c\}$ and $c$, since $C$ is ample (and thus
  isometric), we get that $c$ is a corner of $C$.
\end{proof}

Let $C_<:=(c_1, \ldots, c_m)$ be an ordering of the concepts in $C$.
For any $1\le i\le m$, let $C_i:=\{ c_1,\ldots,c_i\}$ denote the
\emph{$i$'th level set}.  The ordering $C_<$ is called:
\begin{itemize}
\item an \emph{ample} ordering if every level set $C_i$ is ample;
\item a \emph{corner peeling} if every $c_i$ is a \emph{corner} of
  $C_i$;
\item an \emph{isometric} ordering if every level set $C_i$ is
  isometric;
\item a \emph{locally isometric} ordering if every level set $C_i$ is
  locally isometric.
\end{itemize}

\begin{proposition}\label{dismantling}
  The following conditions are equivalent for an ordering
  $C_< = (c_1, \ldots, c_m)$ of an isometric class $C$:
  \begin{enumerate}[(1)]
  \item $C_<$ is ample;
  \item $C_<$ is a corner peeling;
  \item $C_<$ is isometric;
  \item $C_<$ is locally isometric.
  \end{enumerate}
\end{proposition}

\begin{proof}
  Clearly, (3)$\Rightarrow$(4). Conversely, suppose $C_<$ is locally
  isometric but one of its levels is not isometric.  Hence, there
  exists $i<j$ such that any shortest $(c_i,c_j)$-path in $C$ contains
  some $c_k$ with $k>j$.  Additionally, assume that $c_i,c_j$
  minimizes the distance $d(c_i,c_j)$ among all such pairs.  Since
  $C_j$ is locally isometric, necessarily $d(c_i,c_j)\ge 3$.  Let $c_r$
  be the first vertex among $\{ c_{j+1} ,\ldots, c_m\}$ lying in
  $B(c_i,c_j)\cap C$.  If $d(c_i,c_r)\ge 3$ or $d(c_r,c_j)\ge 3$ (say
  the first), then one can replace $c_i,c_j$ by $c_i,c_r$, which
  contradicts the choice of $c_i,c_j$.  Thus,
  $d(c_i,c_r), d(c_r,c_j)\leq 2$, and at least one of them equals $2$
  (say $d(c_i,c_r)=2$).  Now, weak isometricity implies that $c_i$ and
  $c_r$ have a common neighbor $c_\ell$ with $\ell<\max\{ i,r\}=r$.
  If~$\ell<j$ then $c_\ell,c_j$ contradicts the minimality of
  $c_i,c_j$, and if $j<\ell<r$ then $c_\ell$ contradicts the
  minimality of~$c_r$. This shows (4)$\Rightarrow$(3).

  If $c_i$ is a corner of $C_i$, then any two neighbors of $c_i$ in
  $C_i$ have a second common neighbor in $C_i$, and therefore
  $d_{G(C_{i-1})}$ is the restriction of $d_{G(C_i)}$ on $C_{i-1}$.
  Since $C_m = C$ is isometric, this proves (2)$\Rightarrow$(3).

To show (1)$\Rightarrow$(2), let $C_<$ be an ample ordering of $C$. We
assert that each $c_i$ is a corner of~$C_i$.  Indeed, since $C_{i-1}$
is ample and $c_i\notin C_{i-1}$, by Item (i) in Lemma~\ref{ample_bis}
the cube $F[c_i]$, defined with respect to $C_{i-1}$, is included in
$C_i$.  Thus, $c_i$ belongs to a unique maximal cube $F[c_i]$ of
$C_i$, i.e., $c_i$ is a corner of $C_i$.  To prove
(3)$\Rightarrow$(1), let $C_<$ be an isometric ordering. The ampleness
of each $C_i$ follows by induction from Item (ii) of
Lemma~\ref{ample_bis}.
\end{proof}

% A \emph{simplicial complex} $X$ on a set $\U$ is a family of subsets
% of $X$, called \emph{simplices} or \emph{faces} of $X$, such that if
% $\sigma \in X$ and $\sigma' \subseteq \sigma$, then $\sigma' \in
% X$. The \emph{facets} of $X$ are the maximal (by inclusion) faces of
% $X$. The \emph{dimension} $d$ of $X$ is the size of its largest face.
% A simplicial complex $X$ is a \emph{pure} simplicial complex of
% dimension $d$ if all its facets have size $d$.

An isometric concept class $C$ is \emph{dismantlable} if it admits an
ordering satisfying any of the equivalent conditions (1)-(4) in
Proposition~\ref{dismantling}.  Isometric orderings of $Q_n$ are
closely related to shellings of its dual, the \emph{cross-polytope}
$O_n$ (which we define next). Define
$\pm\U:=\{ \pm x_1,\ldots, \pm x_n\}$; so, $\lvert \pm \U\rvert = 2n$,
and we call $-x_i,+{x}_i$ \emph{antipodal}.  The $n$-dimensional
\emph{cross-polytope} is the pure simplicial complex of dimension $n$
whose facets are all $\sigma\subseteq \pm \U$ that contain exactly one
element in each antipodal pair.  Thus, $O_n$ has $2^n$ facets and each
facet $\sigma$ of $O_n$ corresponds to a vertex $c$ of $Q_n$
($+x_i \in \sigma$ if and only if $x_i \in c$). Observe that
$x_i\in c'\Delta c''$ if and only if
$\{+x_i,-x_i\}\subseteq \sigma'\Delta\sigma''$ where $\sigma'$
correspond to $c'$ and $\sigma''$ corresponds to $c''$.

Let $X$ be a pure simplicial complex of dimension $d$.  Two facets
$\sigma,\sigma'$ are adjacent
if~$\lvert \sigma\Delta \sigma'\rvert=2$.  A \emph{shelling} of $X$ is
an ordering $\sigma_1,\ldots,\sigma_p$ of all of its facets such that
$2^{\sigma_j}\bigcap (\bigcup_{i<j} 2^{\sigma_i})$ is a pure
simplicial complex of dimension $d-1$ for every
$2\leq j\leq p$~\cite[Lecture 8]{Zi}.  A \emph{partial shelling} is an
ordering of some facets that satisfies the above condition.  Note that
$\sigma_1,\ldots,\sigma_m$ is a partial shelling if and only if for
every~$i<j \leq m$ there exists $k < j$ such that
$\sigma_i\cap\sigma_j\subseteq \sigma_k\cap\sigma_j$, and
$|\sigma_k\cap\sigma_j| = d-1$, i.e., $\sigma_k\cap\sigma_j$ is a
facet of both $\sigma_j$ and $\sigma_k$. A partial shelling
$\sigma_1,\ldots,\sigma_m$ of $X$ is \emph{1-step extendable} if there
exists a facet $\tau$ of $X$ such that
$\sigma_1,\ldots,\sigma_m, \tau$ is a partial shelling of $X$. A
partial shelling $\sigma_1,\ldots,\sigma_m$ of $X$ is
\emph{extendable} if it can be extended to a shelling of $X$.  A pure
simplicial complex $X$ is \emph{extendably shellable} if every partial
shelling is extendable.  We next establish a relationship between
partial shellings and isometric orderings.

\begin{proposition}\label{ipshelling2}
  Every partial shelling of the cross-polytope $O_n$ defines an
  isometric ordering of the corresponding vertices of the cube $Q_n$.
  Conversely, if~$C$ is an isometric class of~$Q_n$, then any
  isometric ordering of $C$ defines a partial shelling of $O_n$.
\end{proposition}

\begin{proof}
  Let $\sigma_1,\ldots, \sigma_m$ be a partial shelling of $O_n$ and
  $c_1,\ldots, c_m$ be the ordering of the corresponding vertices of
  $Q_n$.  We need to prove that each level set
  $C_j=\{c_1,\ldots, c_j\}$ is isometric.  It suffices to show that
  for every $i<j$ there is $k<j$ such that $d(c_k,c_j)=1$ and
  $c_k\in B(c_i,c_j)$.  Equivalently, for every $i <j$, there is
  $k <j$ such that $\lvert\sigma_k\Delta \sigma_j\rvert=2$ and
  $\sigma_i\cap \sigma_j\subseteq \sigma_k\subseteq \sigma_i\cup
  \sigma_j$: since $\sigma_1,\ldots,\sigma_m$ is a partial shelling,
  there is a facet $\sigma_k$ with $k<j$ such that
  $\lvert \sigma_k\cap \sigma_j\rvert = n-1$ and
  $\sigma_i\cap \sigma_j\subseteq \sigma_k\cap \sigma_j$.  We claim
  that $\sigma_k$ is the desired facet.  It remains to show that (i)
  $\lvert\sigma_j\Delta \sigma_k \rvert= 2$ and (ii)
  $\sigma_k\subseteq \sigma_i\cup\sigma_j$. Item (i) follows since
  $\lvert\sigma_j\rvert=\lvert\sigma_k\rvert=n$, and
  $\lvert \sigma_k \cap \sigma_j\rvert=n-1$.  For Item (ii),
  let~$\sigma_j\setminus \sigma_k=\{ +x\}$ and
  $\sigma_k\setminus \sigma_j=\{ -x\}$.  We need to show that
  $-x\in\sigma_i$, or equivalently that~$+x\notin \sigma_i$.  The
  latter follows since $+x\in \sigma_j\setminus \sigma_k$ and
  $\sigma_j\cap \sigma_i\subseteq \sigma_k$.

  Conversely, let $c_1,\ldots, c_m$ be an isometric ordering
  and~$\sigma_1,\ldots,\sigma_m$ be the ordering of the corresponding
  facets of $O_n$.  We assert that this is a partial shelling. Let
  $i<j$.  It suffices to exhibit $k < j$ such that
  $\lvert \sigma_k\cap \sigma_j\rvert = n-1$ and
  $\sigma_i\cap\sigma_j\subseteq \sigma_k\cap\sigma_j$.  Since $C_j$
  is isometric, $c_j$ has a neighbor $c_k\in B(c_i,c_j)\cap C_j$.
  Since $d(c_j,c_k)=1$ it follows that
  $\lvert \sigma_k\cap\sigma_j\rvert = n-1$.  Since
  $c_k\in B(c_i,c_j)$ it follows that
  $\sigma_i\cap \sigma_j\subset \sigma_k\subset \sigma_i\cup \sigma_j$
  and thus $\sigma_i\cap\sigma_j\subseteq \sigma_k\cap\sigma_j$.
\end{proof}

\begin{corollary}\label{isometric-shelling}
  For any partial shelling $\sigma_1, \ldots, \sigma_m$, let
  $(c_1, \ldots, c_m)$ be the ordering of the corresponding vertices
  of $Q_n$ and let $C = \{c_1, \ldots, c_m\}$. Then the following
  hold.
  \begin{enumerate}
  \item $C$ and $C^* = Q_n\setminus C$ are ample;
  \item the partial shelling $\sigma_1, \ldots, \sigma_m$ is 1-step
    extendable if and only if $C^*$ has a corner.
  \item the partial shelling $\sigma_1, \ldots, \sigma_m$ is
    extendable if and only if $C^*$ is dismantlable.
  \end{enumerate}
  Consequently, if all ample classes are dismantlable, then $O_n$ is
  extendably shellable.
\end{corollary}

\begin{proof}
  % Let $\sigma_1,\ldots, \sigma_m$ be a partial shelling of $O_n$ and
  % let $C=(c_1 ,\ldots, c_m)$ be the ordering of the corresponding
  % vertices of $Q_n$.
  By Proposition~\ref{ipshelling2}, the level sets
  $\{c_1, \ldots, c_i\}$, $1 \leq i \leq m$ are isometric, thus $C$ is
  ample by Proposition~\ref{dismantling}.  The complement $C^*$ is
  also ample, establishing (1).

  If $C^*$ has a corner, then by Lemma~\ref{ample_bis}(iii), $C^*$
  contains a concept $t$ such that $C^*\setminus \{ t\}$ is ample.
  Consequently, its complement $C':=C\cup \{ t\}$ is ample. Let $\tau$
  be the facet of $O_n$ corresponding to $t$. By
  Proposition~\ref{dismantling} $(c_1 ,\ldots, c_m,t)$ is an isometric
  ordering of $C'$. by Proposition~\ref{ipshelling2},
  $\sigma_1,\ldots, \sigma_m,\tau$ is a partial shelling of $O_n$ and
  thus $\sigma_1,\ldots, \sigma_m$ is 1-step extendable.  Conversely,
  assume that there exists $\tau$ such that
  $\sigma_1,\ldots, \sigma_m,\tau$ is a partial shelling of $O_n$ and
  let $t$ be the vertex of $Q_n$ corresponding to $\tau$. By
  Proposition~\ref{ipshelling2}, $(c_1, \ldots, c_m, t)$ is an
  isometric ordering of $C' = C \cup \{t\}$. By
  Proposition~\ref{dismantling}, $C'$ is ample and thus its complement
  $C^* \setminus \{t\}$ is also ample, establishing that $t$ is a
  corner of $C^*$ by Lemma~\ref{ample_bis}(iii). This establishes (2).

  Assertion (3) is a direct consequence of Assertion (2).
  Consequently, if all ample classes are dismantlable, then any
  partial shelling of $O_n$ is extendable and thus $O_n$ is extendably
  shellable.
\end{proof}

It was asked in~\cite{Zi} if any cross-polytope $O_n$ is extendably
shellable.  In the PhD thesis of H. Tracy Hall from 2004, a nice
counterexample to this question is given~\cite{Hall}. Hall's
counterexample arises from the 299 regions of an arrangement of 12
pseudo-hyperplanes.  These regions are encoded as facets of the
12-dimensional cross-polytope $O_{12}$ and it is shown in~\cite{Hall}
that the subcomplex of $O_{12}$ consisting of all other facets admits
a shelling which cannot be extended by adding any of the 299
simplices. By Corollary~\ref{isometric-shelling}(ii), the
ample concept class $C_H$ defined by those 299 simplices does not have
any corner (see Fig.~\ref{fig-hall} for a picture of
$C_H$).\footnote{For the interested reader, a file containing the 299
  concepts of $C_H$ represented as elements of $\{0,1\}^{12}$ is
  available at \url{https://arxiv.org/src/1812.02099/anc/CH.txt}} A
counting shows that $C_H$ is a maximum class of VC-dimension 3. This
completes the proof of our first main result:

\begin{theorem}\label{thm:no_corner_bis}
  There exists a maximum class $C_H$ of VC-dimension 3 without any
  corner.
\end{theorem}

\begin{figure}[ht]
  \includegraphics[scale=0.08]{cenew.pdf}%
  \caption{The maximum class $C_H\subset 2^{12}$ without corners of
      VC-dimension 3 with ${\binom{12}{\le 3}}=299$ concepts.  A
      different edge color is used for each of the 12 dimensions.
      The orientation of the edges is a unique sink orientation
        (defined in Section~\ref{s:rep_maps}) and defines a
        representation map $r_1$ for the class $C_H$: for each concept
        $c$, $r_1(c)$ is the set of the labels of the outgoing
        edges. This representation map was found using the
        \textsf{MiniSat} solver. The appendix contains more
	discussions about this and another representation map of $C_H$. 
	Best viewed in color.}
  \label{fig-hall}
\end{figure}

\begin{remark}
  Hall's concept class $C_H$ also provides a counterexample
  to~\cite[Conjecture 4.2]{MoWa} asserting that \emph{for any two
    ample classes $C_1\subset C_2$ with
    $\lvert C_2\setminus C_1\rvert\ge 2$ there exists an ample class
    $C$ such that $C_1\subset C\subset C_2$.} As noticed
  in~\cite{MoWa}, this conjecture is stronger than the corner peeling
  conjecture disproved by Theorem~\ref{thm:no_corner_bis}.
\end{remark}

\begin{remark}
  Notice also that since ample classes are closed by Cartesian
  products, and any corner in a Cartesian product comes from corners
  in each factor, one can construct other examples of ample classes
  without corners by taking the Cartesian product of $C_H$ by any
  ample class.
\end{remark}

\subsection{$C_H$ refutes the Analysis by Kuzmin and
  Warmuth~\cite{KuWa}}\label{KuWaerror}

The algorithm of~\cite{KuWa} uses the notion of forbidden labels for
maximum classes, introduced in the PhD thesis of Floyd~\cite{Fl} and
used in~\cite{FlWa,KuWa}; we closely follow~\cite{Fl}. Let $C$ be a
maximum class of VC-dimension $d$ on the set ${\U}$. For any
$Y\subseteq {\U}$ with $\lvert Y\rvert=d+1$, the restriction $C|Y$ is
a maximum class of dimension $d$. Thus $C|Y$ contains
$\Phi_d(d+1)=2^{d+1}-1$ concepts. There are $2^{d+1}$ possible
concepts on $Y$.  We call the characteristic function of the unique
concept that is not a member of $C|Y$ a \emph{forbidden label of size
  $d+1$} for $Y$.  Each forbidden label \emph{forbids} all concepts
that contain it from belonging to $C$. Let $c$ be a concept which
contains the forbidden label for $Y$. Since $C$ is a maximum class,
adding $c$ to $C$ would shatter the set $Y$ that is of cardinality
$d+1$.

The algorithm of~\cite{KuWa}, called the \emph{Tail Matching
  Algorithm}, recursively constructs a representation map $\tilde{r}$
for $C^x$, expands $\tilde{r}$ to a map $r$ on the carrier
$N_x(C)=0C^x\cupdot 1C^x$ of $C^x$, and extends $r$ to $\tail_x(C)$
using a special subroutine.  This subroutine and the correctness proof
of the algorithm (Theorem 11) heavily uses that a specially
defined bipartite graph (which we will call $\Gamma$) has a unique
perfect matching. This bipartite graph has the concepts of
$\tail_x(C)$ on one side and the forbidden labels of size $d$ for
$C^x$ on another side. Both sides have the same size.  There is an
edge between a concept and a forbidden label if and only if the
forbidden label is contained in the respective concept. The graph
$\Gamma$ is defined in~\cite[Theorem 10]{KuWa}, which also asserts
that \emph{$\Gamma$ has a unique perfect matching}. In the following,
we show that this assertion is false.

We will show that uniqueness of the matching implies that there is a
corner in the tail (which is contradicted by Hall's concept class
$C_H$). We use the following lemma about perfect bipartite matchings:

\begin{lemma}[\!\!\cite{Cech}]\label{matching}
  Let $G$ be a bipartite graph with bipartition $X,Y$ and unique
  perfect matching $M$. Then, there are vertices $x \in X$ and
  $y \in Y$ with degree one.
\end{lemma}

By Lemma~\ref{matching}, the uniqueness of the matching claimed
in~\cite[Theorem 10]{KuWa} implies that there is a forbidden labeling
that is contained in exactly one concept $c$ of the tail. We claim
that $c$ must be a corner: $c$ is the only concept in the tail
realizing this forbidden labeling and removing this concept from $C$
reduces the number of shattered sets by at least one.  After the
removal, the number of concepts is $\lvert C\rvert-1$. By the Sandwich
Lemma, the number of shattered sets is always at least as big as the
number of concepts.  So removing $c$ reduces the number of shattered
sets by exactly one set $S$ and the resulting class is ample.  For
ample classes, the number of concepts equals the number of supports of
cubes of the 1-inclusion graph and for every shattered set there is a
cube with this shattered set as its support. Thus $c$ lies in a cube
$B$ with support $S$.  There is no other cube with support $S$ because
after removing $c$ there is no cube left with support $S$.  Thus $B$
is the unique cube with support $S$ and must be maximal.  We conclude
that $c$ is a corner of $C$. Since $C_H$ is maximum and does not
contain corners, this leads to a contradiction. This establishes that
Theorem~10 from~\cite{KuWa} is false.

\subsection{Two Families of Dismantlable Ample
  Classes}\label{s:dismantlable}

We continue with two families of dismantlable ample classes:
conditional antimatroids and 2-dimensional ample classes.

\begin{proposition} \label{dismantlable_antimatroids}
  Conditional antimatroids are dismantlable.
\end{proposition}

\begin{proof}
  Let $c_1 , \ldots, c_m$ be an ordering of the concepts of $C$, where
  $i<j$ if and only if $\lvert c_i\rvert \le \lvert c_j\rvert$
  (breaking ties arbitrarily). Clearly, $c_1=\varnothing$ and the
  order $c_1, \ldots, c_m$ is monotone with respect to distances from
  $c_1$.  In particular, any order defined by a Breadth First Search
  (BFS) from $c_1$ satisfies this condition.  We prove that for each
  $i$, $c_i$ is a corner of the level set
  $C_i=\{ c_1 , \ldots, c_{i-1},c_i\}$.  The neighbors of $c_i$ in
  $C_i$ are subsets of $c_i$ containing $\lvert c_i\rvert -1$
  elements. From the definition of extremal points of $c_i$ that we
  denote by $\ex(c_i)$, it immediately follows that
  $c_i\setminus \{ x\}\in C$ if and only if $x\in \ex(c_i)$. For any
  $s\subseteq \ex(c_i)$, since
  $c_i\setminus s=\bigcap_{x\in s} (c_i\setminus \{ x\})$ and $C$ is
  closed under intersections, we conclude that $c_i\setminus s\in
  C$. Therefore, the whole Boolean cube between $c_i$ and
  $c_i\setminus \ex(c_i)$ is included in $C$, showing that $c_i$ is a
  corner of $C_i$.
\end{proof}

The fact that 2-dimensional maximum classes have corners was proved
in~\cite[Theorem~34]{RuRu} and it was generalized in~\cite{MeRo} to
2-dimensional ample classes.  We provide here a different proof of
this result, originating from 1997-1998 and based on a local
characterization of convex sets of general ample classes, which may be
of independent interest.

Given two classes $C'\subseteq C$, $C'$ is \emph{convex} in $C$ if
$B(c,c')\cap C\subseteq C'$ for any $c,c'\in C'$ and $C'$ is
\emph{locally convex} in $C$ if $B(c,c')\cap C\subseteq C'$ for any
$c,c'\in C'$ with $d(c,c') = 2$.

\begin{lemma} \label{local-convexity}
  A connected subclass $C'$ of an ample class $C$ is convex in $C$ if
  and only if $C'$ is locally convex.
\end{lemma}

\begin{proof}
  One direction of the statement is trivial. For the other direction,
  assume that $C'$ is locally convex.  For $c,c'\in C'$, recall that
  $d_{G(C')}(c,c')$ denotes the distance between $c$ and $c'$ in
  $G(C')$. Recall also that since $C$ is ample, $C$ is isometric and
  thus $d_{G(C)}(c,c') = d(c,c')$ for any $c,c' \in C$.  We prove that
  for any $c,c'\in C'$, $B(c,c')\cap C\subseteq C'$ by induction on
  $k=d_{G(C')}(c,c')$, the case $k=2$ being covered by the initial
  assumption. Pick any $t\in B(c,c')\cap C$ and let $L'$ be a shortest
  $(c,c')$-path of $C$ passing via $t$. By Theorem~\ref{thm:ample2}, $L'$ is
  contained in $B(c,c')$.  Let $L''$ be a shortest
  $(c,c')$-path in $C'$. Let $c''$ be the neighbor of $c$ in $L''$.
  Since $d_{G(C')}(c'',c')<k$, by the induction assumption
  $B(c'',c')\cap C\subseteq C'$. Since $G(C)$ is bipartite, either
  $c \in C \cap B(c'',c')$ or $c'' \in C \cap B(c,c')$.  In the first
  case, $t \in C\cap B(c,c') \subseteq C \cap B(c'',c') \subseteq C'$
  and we are done.  So suppose that
  $d_{G(C')}(c,c')=d(c,c')$. This implies that $L''$ is a
  shortest $(c,c')$-path in $C$. Moreover, we can assume that
  $t \notin B(c'',c')$ and consequently, $c'' \notin L'$ and the edge
  $cc''$ is parallel to an edge $uv$ of the path $L'$.  Consider a
  shortest gallery $e_0:=cc'',e_1,\ldots,e_k:=uv$ connecting the edges
  $cc''$ and $uv$ in $C$.  It is constituted of two shortest paths
  $P'=(u_0:=c,u_1,\ldots,u_k:=u)$ and
  $P''=(v_0:=c'',v_1,\ldots,v_k=v)$. Then $P''$ together with the
  subpath of $L'$ comprised between $v$ and $c'$ constitute a shortest
  path between $c''$ and $c'$, thus it belongs to $C'$. Therefore, if
  $t$ is comprised in $L'$ between $v$ and $c'$, then we are
  done. Thus suppose that $t$ belongs to the subpath of $L'$ between
  $c$ and $u$. Since $u_1$ is adjacent to $c,v_1\in C'$, by local
  convexity of $C'$ we obtain that $u_1\in C'$. Applying this argument
  several times, we deduce that the whole path $P'$ belongs to
  $C'$. In particular, $u\in C'$. Since $v$ is between $u$ and $c'$,
  $u\ne c'$, thus $d_{G(C')}(c,u)<k$. By induction hypothesis,
  $B(c,u)\cap C\subseteq C'$.  Since $t\in B(c,u)$, we are done.
\end{proof}

\begin{proposition}\label{dismantlable_2dim}
  2-dimensional ample classes are dismantlable.
\end{proposition}

\begin{proof}
  As for conditional antimatroids, the corner peeling for
  $2$-dimensional classes is based on an algorithmic order of
  concepts.  Consider the following total order $c_1, \ldots, c_m$ of
  the concepts of a $2$-dimensional ample class $C$ constructed
  recursively as follows: start with an arbitrary concept and denote
  it $c_1$ and at step $i$, having numbered the concepts
  $C_{i-1}=\{ c_1, \ldots, c_{i-1}\}$, select as $c_i$ a concept in
  $C\setminus C_{i-1}$ which is adjacent to a maximum number of
  concepts of $C_{i-1}$. We assert that $C_i:=C_{i-1}\cup \{ c_i\}$ is
  ample and that $c_i$ is a corner of $C_i$. Since $C_{i-1}$ is ample
  and $C$ is 2-dimensional, by Lemma~\ref{ample_bis}, $c_i$ has at
  most two neighbors in $C_{i-1}$. First suppose that $c_i$ has two
  neighbors $u$ and $v$ in $C_{i-1}$. By Lemma~\ref{ample_bis},
  $c_i, u$ and $v$ are included in a 2-cube $F[c_i]$ with a fourth
  concept $t$ belonging to $C_{i-1}$. Since $F[c_i]$ is a maximal cube
  of $C$, by Corollary~\ref{cor:maxcub}, $F[c_i]$ is not parallel to
  any other cube of $C$. On the other hand, the edges $c_iu$ and
  $c_iv$ are parallel to the edges $vt$ and $ut$, respectively. This
  shows that $C_i^Y$ is connected for any $Y\subseteq {\U}$, and
  therefore by Theorem~\ref{thm:ample2}(2), $C_i$ is ample and $c_i$
  is a corner. Now suppose that $c_i$ has exactly one neighbor $c$ in
  $C_{i-1}$. From the choice of $c_i$ at step $i$, any concept of
  $C\setminus C_{i-1}$ has at most one neighbor in $C_{i-1}$, i.e.,
  $C_{i-1}$ is locally convex. By Lemma~\ref{local-convexity},
  $C_{i-1}$ is convex in $C$. But then $C_{i}=C_{i-1}\cup \{ c_i\}$
  is isometric since $c_i$ has degree $1$ in $G(C_i)$.
  % must be isometric: otherwise, $c_i$ will be metrically between $c$
  % and another concept of $C_{i-1}$, contradicting the convexity of
  % $C_{i-1}$.
  By Lemma~\ref{ample_bis}, $C_i$ is thus ample and $c_i$
  is a corner of $C_i$.
\end{proof}

\subsection{Collapsibility}\label{s:collapse}

A \emph{free face} of a cube complex $Q(C)$ is a face $Q$ of $Q(C)$
strictly contained in only one other face $Q'$ of $Q(C)$. An
\emph{elementary collapse} is the deletion of a free face $Q$ (thus
also of $Q'$) from $Q(C)$. A cube complex $Q(C)$ is \emph{collapsible}
to a vertex $v_0$ if $C$ can be reduced to $v_0$ by a sequence of
elementary collapses.  Namely, there exists an ordered sequence
$\Lambda = ((Q_1,Q_1'), \ldots, (Q_n,Q_n'))$ where each $Q_i$ is a
free face of $Q(C) \setminus \{Q_1,Q_1', \ldots, Q_{i-1},Q_{i-1}'\}$
contained in $Q_i'$ and
$Q(C) \setminus \{Q_1,Q_1', \ldots, Q_{n},Q_{n}'\}$ is
$\{v_0\}$. Observe that each face of $Q(C)$ distinct from $v_0$
appears exactly once in the sequence.

Collapsibility is a stronger version of
contractibility. The sequences of elementary collapses of a
collapsible cube complex $Q(C)$ can be viewed as discrete Morse
functions~\cite{Fo} without critical cells, i.e., acyclic perfect
matchings of the face poset of $Q(C)$.  From the definition it follows
that if $C$ has a corner peeling, then the cube complex $Q(C)$ is
collapsible: the sequence of elementary collapses follows the corner
peeling order (in general, detecting if a finite complex is
collapsible is NP-complete~\cite{Ta}).  Theorem~\ref{thm:ample2}(5)
implies that the cube complexes of ample classes are contractible (see
also~\cite{BaChKn} for a more general result). In fact, the cube
complexes of ample classes are collapsible (this extends the
collapsibility of finite CAT(0) cube complexes established
in~\cite{AdBe}):

\begin{proposition}\label{collapsible1}
  If $C\subseteq 2^{\U}$ is an ample class, then the cube complex
  $Q(C)$ is collapsible.
\end{proposition}

\begin{proof}
  We proceed by induction on the size of $C$. If $C$ contains only one
  concept, then the statement is trivially true. Otherwise, let
  $x \in \U$ and suppose by induction hypothesis that $Q(C_x)$ is
  collapsible to $v_0$. Let $\Lambda:=((Q_1,Q'_1),\ldots, (Q_n,Q'_n))$
  be the corresponding collapsing sequence of $Q(C_x)$, i.e., a
  partition of the faces of $Q(C_x)\setminus \{v_0\}$ into pairs
  $(Q_i,Q'_i)$ such that $Q_i$ is a free face in the current
  subcomplex $Q(C_x) \setminus \{Q_1,Q_1', \ldots, Q_{i-1},Q_{i-1}'\}$
  of $C_x$ and $Q'_i$ is the unique face properly containing $Q_i$.

  Each cube $Q^*$ of $C$ is mapped to a cube $Q$ of $C_x$ with
  $\supp(Q) = \supp(Q^*) \setminus \{x\}$. Observe that
  $\dim(Q) = \dim(Q^*) -1 $ if $x \in \supp(Q^*)$ and
  $\dim(Q) = \dim(Q^*)$ otherwise.  In the first case, $Q$ is
  contained in $C^x$ and $Q^*$ is entirely contained in the carrier
  $N_x(C)$.
  % In this case, denote by $P^*$ and $R^*$ the two
  % opposite facets of $Q^*$ such that
  % $\supp(P^*) = \supp(R^*) = \supp(Q) = \supp(Q^*)\setminus \{x\}$
  % (observe that $P^*, R^*$ and $Q^*$ are all three mapped to the same
  % cube $Q$ that is contained in $C^x$).
  Conversely, by Lemma~\ref{lem:cubecontraction}, any cube $Q$ of
  $C_x$ is the image of at least one cube $Q^*$ of $C$ (with
  $\supp(Q) = \supp(Q^*) \setminus \{x\}$).  If there exists $Q^*$
  such that $\supp(Q^*) = \supp(Q) \cup \{x\}$, then there exist
  exactly three cubes of $C$ that are mapped to $Q$: $Q^*$ and the two
  opposite facets $P^*$ and $R^*$ of $Q^* $ such that
  $\supp(P^*) = \supp(R^*) = \supp(Q) = \supp(Q^*)\setminus \{x\}$.
  Otherwise, there exists a unique cube $Q^*$ of $C$ that is mapped to
  $Q$.

  % For any cube $Q$ of $C_x$, consider the cube $Q' \subseteq 2^U$ such
  % that $x \in \supp(Q')$ and $Q'_x = Q$ . By
  % Theorem~\ref{thm:ample1}(6), $C \cap Q'$ is ample. Since
  % $\supp(Q) = \supp(Q')\setminus \{x\}$ is shattered by $C \cap Q'$,
  % there exists a cube $Q^*$ in $C \cap Q'$ such that
  % $\supp(Q^*) = \supp(Q)$. Consequently, $Q^*$ is mapped to $Q$ in
  % $C_x$ and thus any cube of $C_x$ is the image of a cube of $C$.

  % Consequently, we
  % conclude that each cube $Q$ of $C_x$ is the image of either a single
  % cube $Q^*$ of $C$ with $\dim(Q^*) = \dim(Q)$, or of three cubes
  % $P^*, Q^*, R^*$ with $\dim(P^*) = \dim(R^*)= \dim(Q)$ and
  % $\dim(Q^*) = \dim(Q)+1$. The latter case occurs if and only if $Q$
  % is contained in $C^x$.

  We derive a collapsing sequence $\Lambda^*$ for $Q(C)$ by replacing
  each elementary collapse of $\Lambda$ by one, two, or three
  elementary collapses in $Q(C)$ and when $v_0 \in C^x$, we add a last
  elementary collapse.  Consider a pair $(Q_i,Q'_i)\in \Lambda$. If
  neither $Q_i$ nor $Q'_i$ are contained in $C^x$, let $Q^*_i$ and
  ${Q^*_i}'$ be the unique preimages of $Q_i$ and $Q'_i$ in $Q(C)$ and
  insert the pair $(Q^*_i,{Q^*_i}')$ in $\Lambda^*$.

  If both $Q_i$ and ${Q_i}'$ are contained in $C^x$, let
  $P_i^*, R_i^*, Q_i^*$ be the cubes of $C$ mapped to $Q$ and
  ${P_i^*}', {R_i^*}', {Q_i^*}'$ the ones mapped to $Q'$ such that
  $\dim(P_i^*) = \dim(R_i^*) = \dim(Q_i^*) - 1 = \dim (Q_i)$,
  $\dim({P_i^*}') = \dim({R_i^*}') = \dim({Q_i^*}') - 1 = \dim
  (Q_i')$, $P_i^*$ is contained in ${P_i^*}'$, $R_i^*$ is contained in
  ${R_i^*}'$, and $Q_i^*$ is contained in ${Q_i^*}'$. We insert the three
  pairs $(Q^*_i,{Q^*_i}'), (P^*_i,{P^*_i}'),(R^*_i,{R^*_i}')$ in
  $\Lambda^*$ in this order.

  Suppose now that $Q_i$ is included in $C^x$ and $Q'_i$ is not
  included in $C^x$.  Let ${Q_i^*}'$ be the unique cube of $C$ mapped
  to $Q_i'$ and let $P_i^*, R_i^*, Q_i^*$ be the cubes of $C$ mapped
  to $Q$ such that
  $\dim(P_i^*) = \dim(R_i^*) = \dim(Q_i^*) - 1 = \dim (Q_i)$. Assume
  without loss of generality that $P_i^*$ is a facet of both $Q_i^*$
  and ${Q_i^*}'$. We insert the two pairs
  $(R_i^*,Q_i^*), (P_i^*,{Q_i^*}')$ in $\Lambda^*$ in this order.

  Finally, we consider the vertex $v_0$. If $v_0 \notin C^x$, then the
  preimage of $v_0$ contains only one vertex that we denote by
  $v_0^*$. Suppose now that $v_0 \in C^x$. Let $v_0^*, u_0^* \in C$
  such that $v_0^* = v_0$ and $u_0^* = v_0 \cup \{x\}$. Then the
  preimage of $v_0$ in $Q$ is of size $3$: it contains the vertices
  $v_0^*, u_0^*$ and the edge $\{u_0^*,v_0^*\}$. In this case, we add
  to $\Lambda^*$ the elementary collapse $(u_0^*,\{u_0^*,v_0^*\})$.

  Each cube $Q^*$ of $Q(C)$ is in the preimage of a single cube of
  $Q(C_x)$. Moreover, each cube of $Q(C_x)\setminus \{v_0\}$ appears
  in exactly one collapsing pair of the sequence
  $\Lambda$. Consequently, each cube of $Q(C)$ whose image is not
  $v_0$ appears exactly once in the sequence $\Lambda^*$. Notice also
  that when the preimage of $v_0$ is of size $3$, both $u_0^*$ and
  $\{u_0^*,v_0^*\}$ appear in the last pair of
  $\Lambda^*$. Consequently, each cube of $Q(C)\setminus \{v_0^*\}$
  appears in exactly one pair of the sequence $\Lambda^*$.  Since each
  pair of $\Lambda^*$ (except potentially $(u_0^*, \{u_0^*,v_0^*\}$)
  is derived from a collapsing pair of $\Lambda$ and since $\Lambda$
  is a sequence of elementary collapses of $Q(C_x)$, we can deduce
  that $\Lambda^*$ is a sequence of elementary collapses of the cube
  complex $Q(C)$ to $v_0^*$.
  % % If $Q^*_i$ and ${Q^*_i}'$ have the same
  % % dimensions as $Q_i$ and $Q'_i$, then we insert the pair
  % % $(Q^*_i,{Q^*_i}')$ at the end of $\Lambda^*$. If both $Q^*_i$ and
  % % ${Q^*_i}'$ have larger dimensions than $Q_i$ and $Q'_i$, then we
  % % insert at the end of $\Lambda^*$ the pairs
  % % $(Q^*_i,{Q^*_i}'), (P^*_i,{P^*_i}'),(R^*_i,{R^*_i}')$.
  % Finally,
  % suppose that ${Q^*_i}'$ has the same dimension as $Q'_i$ and $Q^*_i$
  % has dimension larger than $Q_i$. Suppose that $P^*_i$ is the facet
  % of $Q^*_i$ contained in ${Q^*_i}'$. Then we insert at the end of
  % $\Lambda^*$ the pairs $(R^*_i,Q^*_i),(P^*_i,{Q^*_i}')$. One can
  % easily check that each cube of $Q(C)$ belongs to a single pair of
  % $\Lambda^*$ and, since each pair of $\Lambda^*$ is derived from a
  % collapsing pair of $\Lambda$, we deduce that $\Lambda^*$ is a
  % sequence of elementary collapses of $Q(C)$.
\end{proof}

  % \begin{figure}
  %   \includegraphics[scale=0.75,page=10]{figs-lopsided.pdf}%
  %   \caption{An ample class $C$ and its restriction $C_x$ for $x=3$.}%
  %   \label{fig-ex-collapsible}
  % \end{figure}

\begin{example}\label{ex-collapsibility}
  We illustrate the construction in the proof of
  Proposition~\ref{collapsible1} with the ample class $C$ from
  Figure~\ref{fig-ex-collapsible} and its restriction $C_x$.  The
  names of the concepts are depicted in the figure. The faces will be
  denoted by the list of their vertices.

  Consider the following two collapsing sequences of $C_x$:
  \begin{itemize}
  \item
    $\Lambda_1 = (\{a\},\{a,b\}), (\{b\},\{b,c\})$,  and
  \item $\Lambda_2 = (\{a\},\{a,b\}), (\{c\},\{b,c\})$.
  \end{itemize}
  Note that $\Lambda_1$ collapses $Q(C_x)$ to $c$ and $\Lambda_1$
  collapses $Q(C_x)$ to $b$.

  The corresponding collapsing sequences for $C$ obtained by the
  construction described in the proof of
  Proposition~\ref{collapsible1} are the following:
  \begin{itemize}
  \item  $\Lambda_1^* = (\{a,a'\},\{a,b,b',a\}), (\{a'\},\{a',b'\}),
    (\{a\},\{a,b\}), (\{b'\},\{b,b'\}), (\{b\},\{b,c\})$ collapses
    $Q(C)$ to $c$, and

  \item  $\Lambda_2^* = (\{a,a'\},\{a,b,b',a\}), (\{a'\},\{a',b'\}),
    (\{a\},\{a,b\}), (\{c\},\{b,c\}), (\{b'\},\{b,b'\})$ collapses
    $Q(C)$ to $b$.
  \end{itemize}

  In both cases, the subsequence
  $(\{a,a'\},\{a,b,b',a\}), (\{a'\},\{a',b'\}), (\{a\},\{a,b\})$
  corresponds to the elementary collapse $(\{a\},\{a,b\})$. In
  $\Lambda_1^*$, the subsequence $ (\{b'\},\{b,b'\}), (\{b\},\{b,c\})$
  corresponds to the elementary collapse $(\{b\},\{b,c\})$ in
  $\Lambda_1$. In $\Lambda^*_2$, the elementary collapse
  $ (\{c\},\{b,c\})$ corresponds to the elementary collapse
  $(\{c\},\{b,c\})$ in $\Lambda_2$ and the elementary collapse
  $(\{b'\},\{b,b'\})$ is the elementary collapse associated to the
  last vertex $b$ of $\Lambda_2$.
  \end{example}

  \begin{figure}
    \includegraphics[scale=0.75,page=10]{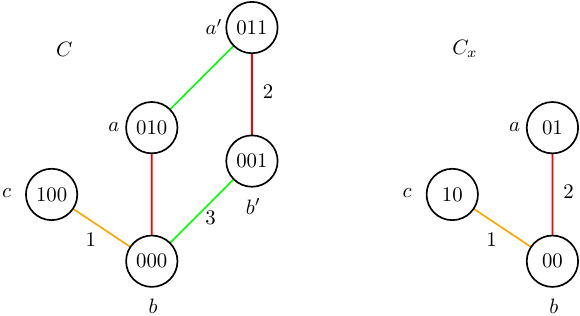}%
    \caption{The ample class $C$ used in
      Example~\ref{ex-collapsibility} and its restriction $C_x$ for
      $x=3$.}%
    \label{fig-ex-collapsible}
  \end{figure}

\section{Representation Maps for Maximum Classes}\label{s:maximum}

In this section, we prove that maximum classes admit representation
maps, and therefore optimal unlabeled sample compression schemes.

\begin{theorem}\label{th-maximum}
  Any maximum class $C\subseteq 2^{\U}$ of VC-dimension $d$ admits a
  representation map, and consequently, an unlabeled sample
  compression scheme of size $d$.
\end{theorem}

The crux of the proof of Theorem~\ref{th-maximum} is the following
proposition.  Let $C$ be a $d$-dimensional maximum class and let
$D\subseteq C$ be a $(d-1)$-dimensional maximum subclass.  A
\emph{missed simplex} for the pair $(C,D)$ is a simplex
$\sigma\in \uoX(C)\setminus \uoX(D)$.  Note that since $C$ and $D$ are
maximum, any missed simplex has size $d$.  An \emph{incomplete
  cube}~$Q$ for $(C,D)$ is a cube of $C$ such that $\supp(Q)$ is a
missed simplex.  For any incomplete cube $Q$ with $\sigma = \supp(Q)$,
$C|\sigma$ is a $d$-cube and $D|\sigma$ is a maximum class of
dimension $d-1$. Observe that any incomplete cube for $(C,D)$ is a
maximal cube in $C$ and by Corollary~\ref{cor:maxcub}, there is a
bijection between the missed simplices for $(C,D)$ and the incomplete
cubes for $(C,D)$.

Consider a missed simplex $\sigma$ for $(C,D)$ and the incomplete cube
$Q$ for $(C,D)$ such that $\supp(Q) = \sigma$. Since
$\lvert \sigma\rvert = d$, we have
$\lvert C|\sigma\rvert = \binom{d}{\leq d} = \binom{d}{\leq d-1} +1 =
\lvert D|\sigma\rvert + 1$. Since $Q|\sigma = C|\sigma$, there exists
a unique concept $c \in Q$ such that~$c|\sigma\notin D|\sigma$.  We
call $c$ the \emph{source} of $Q$ and we consider the
\emph{source-map} $s$ from the set of incomplete cubes for $(C,D)$ to
$C\setminus D$ where $s(Q)$ is the source of $Q$. In fact, we show in the
following proposition that the source-map is a bijection between the
incomplete cubes for $(C,D)$ and the concepts of $C\setminus D$.

% We denote $c$ by $s(Q)$, and call $c$ the \emph{source} of $Q$. In
% fact, the source-map is a bijection between missed simplices for
% $(C,D)$ and concepts of $C\setminus D$:

\begin{proposition}\label{lem-unique-source}
  Each $c\in C \setminus D$ is the source of a unique incomplete cube
  for $(C,D)$. Moreover, if $r': D \to \uoX(D)$ is a representation
  map for $D$ and $r:C\to \uoX(C)$ extends $r'$ by setting
  $r(c) = \supp(s^{-1}(c))$ for each $c\in C\setminus D$, then $r$ is
  a representation map for $C$.
\end{proposition}

\begin{proof}[Proof of Theorem~\ref{th-maximum}.]
    Following the general recursive construction idea of~\cite{KuWa}, we derive a
  representation map for $C$ by induction on $\lvert \U\rvert$.  If
  $\lvert U \rvert = 0$, then $C$ is a $0$-dimensional maximum class
  containing a unique concept $c$. In this case, we set
  $r(c)=\emptyset$ and $r$ is clearly a representation map for $C$.
  For the induction step, (see
  Figure~\ref{fig-representation-map-max}), pick $x\in \U$ and
  consider the maximum classes $C_x$ and $C^x\subset C_x$ with domain
  $\U\setminus\{x\}$.  By induction, $C^x$ has a representation map
  $r^x$.  Use Proposition~\ref{lem-unique-source} to extend $r^x$ to a
  representation map $r_x$ of $C_x$. Define a map $r$ on $C$ as
  follows:
  \begin{equation*}
    r(c) = \begin{cases}
      r_x(c_x) & \text{if } c_x \notin C^x \text{ or } x\notin c,\\
      r_x(c_x) \cup\{x\} & \text{if } c_x \in  C^x \text{ and } x \in c.
    \end{cases}
  \end{equation*}
  It is easy to verify that $r$ is non-clashing: indeed, if
  $c'\neq c''\in C$ satisfy $c'_x\neq c''_x$ then
  $c'_x\vert r_x(c'_x)\cup r_x(c''_x) \neq c''_x\vert r_x(c'_x)\cup
  r_x(c''_x)$.  Since
  $r_x(c'_x)\subseteq r(c'), r_x(c''_x)\subseteq r(c'')$, it follows
  that also $c',c''$ disagree on $r(c')\cup r(c'')$.  Else,
  $c'_x = c''_x\in C^x$ and $c'(x)\neq c''(x)$.  In this case,
  $x\in r(c')\cup r(c'')$ and therefore $c',c''$ disagree on
  $r(c')\cup r(c'')$.

  It remains to show that $r$ is a bijection between $C$ and
  $\uoX(C) = {\binom{\U}{\leq d}}$.  It is easy to verify that $r$ is
  injective. So, it remains to show that $\lvert r(c)\rvert \leq d$,
  for every $c\in C$. This is clear when $c_x\notin C^x$ or
  $x\notin c$.  If $c_x\in C^x$ and $x\in c$, then
  $r(c) = r^x(c_x)\cup\{x\}$ and $\lvert r^x(c_x)\rvert \leq d-1$
  (since $C^x$ is $(d-1)$-dimensional).  Hence,
  $\lvert r(c)\rvert \leq d$ as required, concluding the proof.
\end{proof}

\begin{figure}[ht]
  \includegraphics[page=7,scale=0.7]{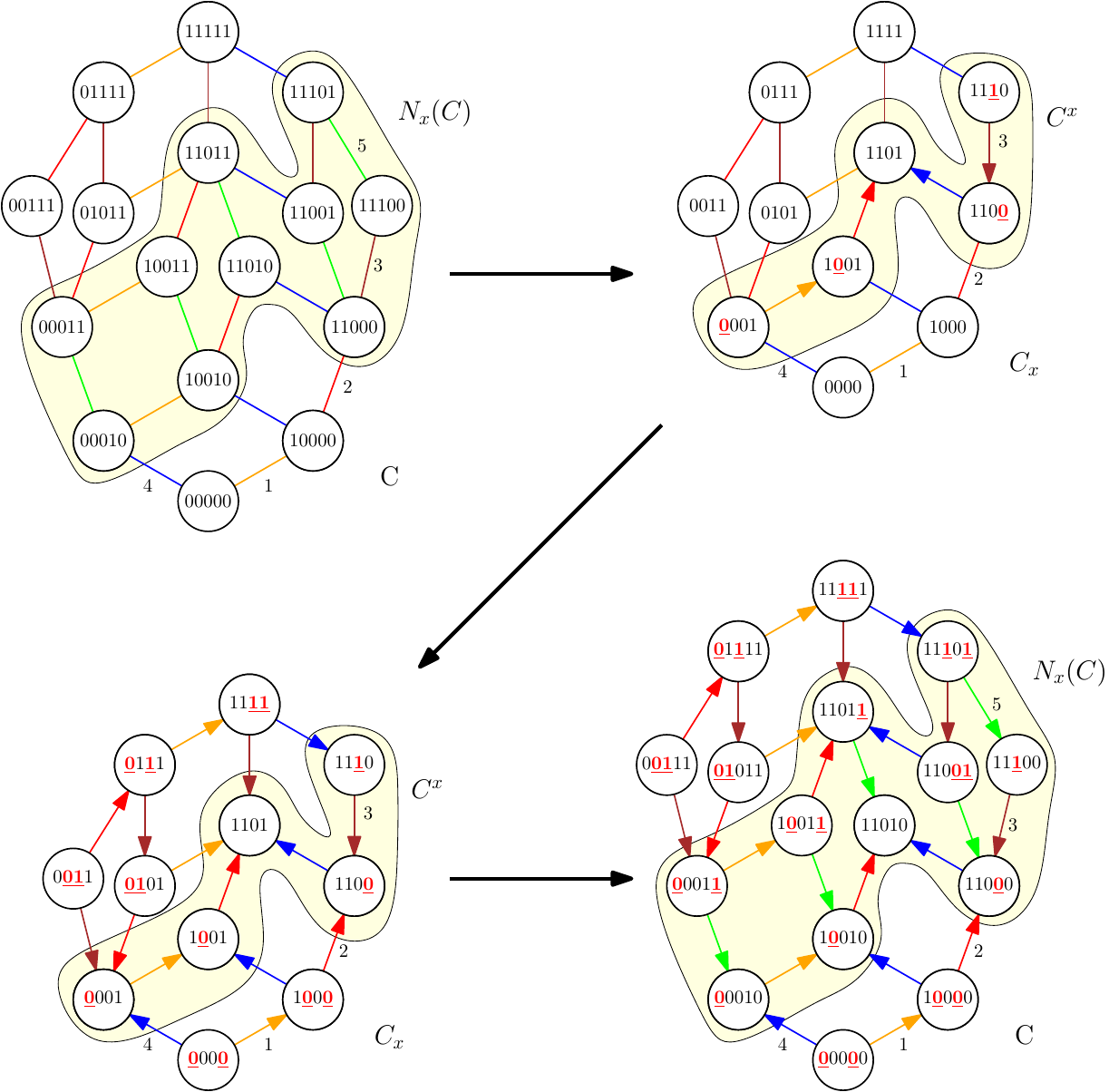}%
  \caption{Illustrating the proof of Theorem~\ref{th-maximum} (when
    $x=5$): to construct a representation map for $C$, we inductively
    construct a representation map $r^x$ for $C^x$, extend it to a
    representation map $r_x$ for $C_x$ using
    Proposition~\ref{lem-unique-source} with $D = C^x$, and finally
    extend it to a representation map $r$ for $C$.  The representation
    maps $r^x,r_x,$ and $r$ are defined by the coordinates of the
    underlined bits or by the labels of the outgoing edges (see
    Theorem~\ref{t:local-to-global}).}%
  \label{fig-representation-map-max}
\end{figure}

\begin{proof}[Proof of Proposition~\ref{lem-unique-source}.]
  To prove the proposition, we first prove that incomplete cubes and
  their sources are preserved by restrictions and reductions
  (Lemma~\ref{lem-source-restriction}). To show that each concept of
  $C \setminus D$ is the source of one incomplete cube, we consider a
  minimal counterexample and we establish that in this counterexample
  each concept is the source of at most $2$ incomplete
  cubes. Moreover, if a concept is the source of $0$ (respectively,
  $1, 2$) incomplete cubes, then any of its neighbors in $G(C)$ is the
  source of $2$ (respectively, $1, 0$) incomplete cubes
  (Lemma~\ref{claim-source-sink}). Using this and the notions of
  galleries and the associated trees defined below, we establish the
  first assertion of the proposition. Then using
  Lemma~\ref{lem-source-restriction} and the first assertion, we
  establish the second part of the proposition. We also give a
  geometric characterization of sources
  (Lemma~\ref{source_incomplete}) that is used later to estimate the
  complexity of computing the representation map (see
  Remark~\ref{rem-complexity}).

  Call a maximal cube of $C$ a \emph{chamber} and a facet of a chamber
  a \emph{panel} (a $\sigma'$-\emph{panel} if its support is
  $\sigma'$).  Any $\sigma'$-panel in $C$ satisfies
  $\lvert \sigma'\rvert = d-1$ and $\sigma' \in \uoX(D)$. Recall that
  a gallery between two parallel cubes $Q',Q''$ (say, two
  $\sigma'$-cubes) is any simple path of $\sigma'$-cubes
  $(Q_0:=Q',Q_1,\ldots, Q_k:=Q'')$, where $Q_i\cup Q_{i+1}$ is a
  $d$-cube.  By Theorem~\ref{thm:ample2}(3), any two parallel cubes of
  $C$ are connected by a gallery in $C$.  Since $D$ is a maximum
  class, any panel of $C$ is parallel to a panel that is a maximal
  cube of $D$.  Also for any maximal simplex $\sigma' \in \uoX(D)$,
  the class $C^{\sigma'}$ is a maximum class of dimension~1 and
  $D^{\sigma'}$ is a maximum class of dimension 0 (single concept).
  Thus $C^{\sigma'}$ is a tree (e.g.~\cite[Lemma 7]{GaWe}) which
  contains the unique concept $c\in D^{\sigma'}$.  We call $c$
  \emph{the root} of $C^{\sigma'}$ and we denote by $P(\sigma')$ the
  unique~$\sigma'$-panel $P$ of $D$ such that $P^{\sigma'}=c$.

  The next result provides a geometric characterization of sources:

  \begin{lemma}\label{source_incomplete}
    Let $\sigma$ be a missed simplex of the pair $(C,D)$.  A concept
    $c\in Q$ is the source of the unique $\sigma$-cube $Q$ if and only
    if for any $x\in \sigma$, if $\sigma':=\sigma \setminus \{x\}$ and
    $P',P''$ are the two $\sigma'$-panels of $Q$ with $c\in P''$, then
    $(P')^{\sigma'}$ is on the path between $(P'')^{\sigma'}$ and the root
    $(P(\sigma'))^{\sigma'}$ of the tree $C^{\sigma'}$.
  \end{lemma}

  \begin{proof}
    Observe that there exists a unique concept $c \in Q$ such that for
    any $x\in \sigma$, if $\sigma':=\sigma \setminus \{x\}$ and
    $P',P''$ are the two $\sigma'$-panels of $Q$ with $c\in P''$, then
    $(P')^{\sigma'}$ is on the path between $(P'')^{\sigma'}$ and the
    root $(P(\sigma'))^{\sigma'}$ of the tree $C^{\sigma'}$. Since $Q$
    contains a unique source, it is enough to show that this unique
    concept $c$ is the source.

    % Consider $c \in Q$, $x \in \sigma$ and $P', P''$, the two
    % $\sigma'$-panels of $Q$ with $c \in P''$. Suppose that that the
    % unique gallery $L$ between $P'$ and the root $P(\sigma')$ passes
    % via $P''$, i.e.,
    % $L = (P_0 = P(\sigma'), P_1, \ldots, P_{m-1} = P'', P_m = P')$.
    % Since $Q$ is a maximal cube, $x$ is not in the domain of the
    % chamber $P_i\cup P_{i+1}$ for $i < m-1$. This implies that there
    % exists $c_0 \in P(\sigma') \subseteq D$ such that
    % $c_0|\sigma = c|\sigma$, and consequently, $c|\sigma$ is not the
    % missed sample for $\sigma$ and $c$ cannot be the source of $Q$.

    Assume by way of contradiction that this is not the case, i.e.,
    that $c$ is the source of $Q$ and that there exists $x \in \sigma$
    and two $\sigma'$-panels $P', P''$ with
    $\sigma' = \sigma\setminus\{x\}$ and $c \in P''$ such that the
    unique gallery $L$ between $P'$ and the root $P(\sigma')$ passes
    via $P''$, i.e.,
    $L = (P_0 = P(\sigma'), P_1, \ldots, P_{m-1} = P'', P_m = P')$.
    Since $Q$ is a maximal cube, by Corollary~\ref{cor:maxcub}, $x$ is
    not in the domain of the chamber $P_i\cup P_{i+1}$ for $i <
    m-1$. This implies that there exists
    $c_0 \in P(\sigma') \subseteq D$ such that
    $c_0|\sigma = c|\sigma$, and consequently, $c|\sigma$ is not the
    missed sample for $\sigma$.
  \end{proof}

  In the next lemma, we show that incomplete cubes and their sources
  are preserved by restrictions and reductions.

  \begin{lemma}\label{lem-source-restriction}
    Let $Q$ be an incomplete cube for $(C,D)$ with source $s$ and
    support $\sigma$, and let $x,y \in \U$ such that $x \notin \sigma$
    and $y\in \sigma$. Then, the following holds:
    \begin{enumerate}[(i)]
    \item $Q_x$ is an incomplete cube for $(C_x,D_x)$ whose source is
      $s_x$.
    \item $Q^y$ is an incomplete cube for $(C^y,D^y)$ whose source is
      $s^y$.
    \end{enumerate}
  \end{lemma}

  \begin{proof}
    Item~(i): $C_x$ and $D_x$ are maximum classes on
    $\U \setminus \{x\}$ of VC-dimensions $d$ and $d-1$, and
    $\supp(Q_x)=\sigma$. Therefore, $Q_x$ is an incomplete cube for
    $(C_x,D_x)$.  By definition, $s$ is the unique concept $c\in Q$
    such that $c|\sigma \notin D| \sigma$. Since $x \notin \sigma$,
    $D| \sigma = D_x | \sigma$ and $s_x$ is the unique concept $c$ of
    $Q_x$ so that $c | \sigma \notin D_x| \sigma$, i.e., $s_x$ is the
    source of $Q_x$.

    Item~(ii): $C^y$ and $D^y$ are maximum classes on
    $\U \setminus \{y\}$ of VC-dimensions $d-1$ and $d - 2$.  Since
    $y \in \supp(Q)$, $\dim(Q^y) = d-1$ and $Q^y$ is an incomplete
    cube for $(C^y,D^y)$.  Let $\sigma'=\sigma\setminus\{y\}$. It
    remains to show that $s^y \vert \sigma' \notin D^y \vert \sigma'$.
    Indeed, otherwise both extensions of $s^y$ in $\sigma$,
    namely~$s,s\Delta\{y\}$, are in $D\vert \sigma$ which contradicts
    that $s=s(Q)$.
  \end{proof}

  Next we prove that each concept of $C \setminus D$ is the source of
  a unique incomplete cube. Since there is a bijection between
  incomplete cubes and missed simplices, the number of incomplete
  cubes is $|X(C) \setminus X(D)| = |C\setminus D|$. Therefore, it is
  sufficient to show that each concept of $C \setminus D$ is the
  source of at most one incomplete cube.  Assume the contrary and let
  $(C,D)$ be a counterexample minimizing the size of $\U$. First, if a
  concept $c \in C\setminus D$ is the source of two incomplete cubes
  $Q_1,Q_2$, then $\dom(C) = \supp(Q_1) \cupdot\supp(Q_2)$.  Indeed,
  let $\sigma_1 = \supp(Q_1)$ and $\sigma_2 = \supp(Q_2)$.  By
  Lemma~\ref{lem-source-restriction}(i) and minimality of $(C,D)$,
  $\dom(C) = \sigma_1 \cup \sigma_2$.  Indeed, if there exists
  $x \notin \sigma_1 \cup \sigma_2$, $c_x$ is the source of the
  incomplete cubes $(Q_1)_x$ and $(Q_2)_x$ for $(C_x, D_x)$, contrary
  to minimality of $(C,D)$.  By Lemma~\ref{lem-source-restriction}(ii)
  and minimality of $(C,D)$, $\sigma_1 \cap \sigma_2 =
  \varnothing$. Indeed, if there exists $x \in\sigma_1 \cap \sigma_2$,
  $c^x$ is the source of the incomplete cubes $Q_1^x$ and $Q_2^x$ for
  $(C^x, D^x)$, contrary to minimality of $(C,D)$.

  Next we assert that any $c\in C\setminus D$ is the source of at most
  2 incomplete cubes. Indeed, let $c$ be the source of incomplete
  cubes $Q_1, Q_2, Q_3$. Then
  $\dom(C) = \supp(Q_1) \cupdot \supp(Q_2)$, i.e.,
  $\supp(Q_2) = \dom(C) \setminus \supp(Q_1)$. For similar reasons,
  $\supp(Q_3) = \dom(C) \setminus \supp(Q_1) = \supp(Q_2)$.  Thus, by
  Corollary~\ref{cor:maxcub}, $Q_2=Q_3$.

  \begin{lemma}\label{claim-source-sink}
    Let $c',c''\in C\setminus D$ be neighbors and let
    $c'\Delta c'' =\{x\}$.  Then, $c'$ is the source of 2 incomplete
    cubes if and only if $c''$ is the source of 0 incomplete cubes.
    Consequently, every connected component in $G(C\setminus D)$
    either contains only concepts $c$
    with~$\lvert s^{-1}(c)\rvert \in\{0,2\}$, or only concepts $c$
    with~$\lvert s^{-1}(c)\rvert = 1$.
  \end{lemma}

  \begin{proof}
    By minimality of $(C,D)$, $(c')^x=(c'')^x$ is the source of a
    unique incomplete cube for~$(C^x,D^x)$ and $c'_x=c''_x$ is the
    source of a unique incomplete cube for $(C_x,D_x)$. Let $Q_1$ be
    the incomplete cube for $(C,D)$ such that $(c')_x$ is the source
    of $(Q_1)_x$. Let $Q_2$ be the incomplete cube for $(C,D)$ such that
    $(c')^x$ is the source of $(Q_2)^x$.  Since $(Q_1)_x$ has a unique
    source $c'_x = c''_x$, by Lemma~\ref{lem-source-restriction}(i),
    $(s(Q_1))_x = c'_x$ and consequently, $s(Q_1) \in
    \{c',c''\}$. Similarly, since $Q_2^x$ has a unique source
    $(c')^x = (c'')^x$, by Lemma~\ref{lem-source-restriction}(ii),
    $(s(Q_2))^x = (c')^x$ and thus $s(Q_2) \in \{c',c''\}$.
    Consequently, $c'$ is the source of 2 incomplete cubes ($Q_1$ and
    $Q_2$) if and only if $c''$ is the source of 0 incomplete cubes.
  \end{proof}

  Pick $c\in C\setminus D$ that is the source of two incomplete cubes
  for $(C,D)$ and an incomplete cube $Q$ such that $c = s(Q)$. Let
  $\sigma=\supp(Q)$, $x\in\sigma$, and $\sigma'=\sigma\setminus\{x\}$.
  The concept $c$ belongs to a unique $\sigma'$-panel $P$.  Let
  $L=(P_0 = P(\sigma'), P_1 ,\ldots, P_{m-1}, P_m = P)$ be the unique
  gallery between the root $P(\sigma')$ of the tree $C^{\sigma'}$ and
  $P$.  For $i=1,\ldots, m$, denote the chamber $P_{i-1}\cup P_i$ by
  $Q_i$.  Since $Q_i\cap D$ is ample and $Q_i$ is not contained in
  $D$, it follows that the complement $Q_i\setminus D$ is a nonempty
  ample class.  Hence $Q_i\setminus D$ induces a nonempty connected
  subgraph of $G(C\setminus D)$.  Therefore, it follows that~$c$ and
  each concept $c'\in Q_i\setminus D$ are connected by a path in
  $G(C\setminus D)$, and by Lemma~\ref{claim-source-sink} it follows
  that
  % Since $P_i\cap D$ is ample for $i \geq 0$, $Q_i\cap D$ is
  % ample for $i > 0$, and $P_i$ is not contained in $D$ for $i>0$, it
  % follows that the complements $P_i\setminus D$ and $Q_i\setminus D$
  % are nonempty ample classes.  Hence $P_i\setminus D$ and
  % $Q_i\setminus D$ induce nonempty connected subgraphs of
  % $G(C\setminus D)$ when $i >0$.  Therefore, it follows that~$c$ and
  % each concept $c'\in Q_i\setminus D$ are connected in
  % $G(C\setminus D)$ by a path for $i>0$, and by
  % Lemma~\ref{claim-source-sink} it follows that
  \begin{equation}\label{eq:0or2}
    \text{For each $i$, each $c'\in Q_i\setminus D$ is the source of
      either 0 or 2 incomplete cubes.}
  \end{equation}
	
  Consider the chamber $Q_1=P_0\cup P_1$ and its source $s=s(Q_1)$.
  By the definition of the source, necessarily $s \in P_1$ and
  $s\notin D$.  Therefore, Property~(\ref{eq:0or2}) implies that there
  must exist another cube $Q'$ such that $s=s(Q')$.  Let $s'$ be the
  neighbor of $s$ in $P_0=P(\sigma')$; note that $s'\in D$.  Since
  $\supp(Q_1) \cap \supp(Q') = \varnothing$, it follows that
  $s | \supp(Q') = s' |\supp(Q') \in D|\supp(Q')$, contradicting that
  $s=s(Q')$.  This establishes the first assertion of
  Proposition~\ref{lem-unique-source}.

  We prove now that the map $r$ defined in
  Proposition~\ref{lem-unique-source} is a representation map for $C$.
  It is easy to verify that it is a bijection between $C$ and
  $\uoX(C)$, so it remain to establish the non-clashing property:
  $c|(r(c) \cup r(c')) \neq c'|(r(c) \cup r(c'))$ for all distinct
  pairs $c,c'\in C$.  This holds when $c,c' \in D$ because $r'$ is a
  representation map.  Next, if $c \in C\setminus D$ and $c' \in D$,
  this holds because $c|r(c) \notin D|r(c)$ by the properties of $s$.

  Thus, it remains to show that every distinct $c,c' \in C\setminus D$
  satisfy the non-clashing condition.  Assume towards contradiction
  that this does not hold and consider a counterexample with minimal
  domain size~$\lvert\U\rvert$. Consequently, there exist distinct
  $c, c' \in C\setminus D$ such that $c(z) = c'(z)$ for any
  $z \in \supp(Q)\cup\supp(Q')$, where $Q = s^{-1}(c)$ and
  $Q' = s^{-1}(c')$. Since $c \neq c'$, there exists
  $x \in U \setminus (\supp(Q)\cup\supp(Q'))$ such that
  $c(x) \neq c(x')$.  If there exists
  $y \in U \setminus (\supp(Q)\cup\supp(Q'))$ distinct from $x$, then
  from Lemma~\ref{lem-source-restriction}(1), $Q_y$ and $Q'_y$ are
  incomplete cubes for $(C_y,D_y)$ whose respective sources are $c_y$
  and $c'_y$. Since $r(c_y) = r(c)$, $r(c'_y) = r(c')$, and
  $c_y \neq c'_y$ (since they differ on $x$), $c_y$ and $c'_y$ clash
  in $C_y$, contradicting the minimality of the counterexample
  $(C,D)$. Consequently, we can assume that
  $U = \supp(Q) \cup \supp(Q') \cup \{x\}$ and that $c$ and $c'$
  differ only on $x$. In this case, by
  Lemma~\ref{lem-source-restriction}(1), $c_x = c'_x$ is the source of
  both incomplete cubes $Q_x$ and $Q'_x$. Since $Q$ and $Q'$ are
  distinct incomplete cubes, we have
  $\supp(Q_x) = \supp(Q) \neq \supp(Q') = \supp(Q'_x)$ and thus,
  $c_x = c'_x$ is the source of two different incomplete cubes of
  $C_x$, contradicting the first statement of the proposition.
  % By
  % minimality, $\supp(Q')\cup\supp(Q) = \U$ (or else $(C_x,D_x)$, for
  % some $x\notin \supp(Q')\cup\supp(Q)$ would be a smaller
  % counterexample).  Therefore, since $c,c'$ are distinct, there must
  % be $x\in\U=\supp(Q')\cup\supp(Q)$ such that~$c(x)\neq c'(x)$, which
  % is a contradiction.
  This ends the proof of
  Proposition~\ref{lem-unique-source}.
\end{proof}	

\begin{remark}
  Proposition~\ref{lem-unique-source} relies on a canonical bijection
  between missed simplices and incomplete cubes (that allows to define
  the source-map). This is due to the fact that any missed simplex for
  $(C,D)$ is a maximal simplex of $C$ and is thus the support of a
  unique cube in $C$.  This property does not longer hold for general
  ample classes since there are missed simplices that are not maximal
  and therefore there are the supports of several cubes in $C$. A
  first step to generalize Proposition~\ref{lem-unique-source} to
  ample classes could be to find a way to define a source-map between
  the missed simplices in $X(C)\setminus X(D)$ and the concepts in
  $C\setminus D$.
\end{remark}

\begin{remark}\label{rem-complexity}
  To compute the representation map for a maximum class $C$ of
  dimension $d$, we make $d$ recursive calls. For each call, the
  costliest operation is to compute the source-map that can be done in
  $O(|C|^3)$ time as follows. First we naively compute in $O(|C|^3)$
  the cubes of dimension $d$ in $C$ and the cubes of dimension $d-1$
  in $C^x$. Then we compute in $O(d|C|)$ time the trees $C^{\sigma'}$
  for any maximal simplex $\sigma' \in X(C^x)$ (the roots of those
  trees are the cubes of $C^x$ of dimension $d-1$) and we can then
  compute the source of each cube by
  Lemma~\ref{source_incomplete}. Consequently, one can compute a
  representation map for a maximum class $C$ of dimension $d$ in
  $O(d|C|^3)$.
\end{remark}

\begin{remark}\label{rem-rm-CH}
  In the appendix, we give two representation maps for Hall's concept
  class $C_H$ presented in Figure~\ref{fig-hall}. One of this
  representation map is obtained by the method described in the proof
  of Theorem~\ref{th-maximum}. The other one is obtained by
  transforming the problem into a SAT formula and using a SAT solver.
\end{remark}

\section{Representation Maps for Ample Classes}\label{s:rep_maps}

In this section, we provide combinatorial and geometric
characterizations of representation maps of ample classes. We first
show that representation maps lead to optimal unlabeled sample
compression schemes and that they are equivalent to unique sink
orientations (USO) of $G(C)$ (see below for the two conditions
defining USOs).  We also show that corner peelings of ample classes
are equivalent to the existence of acyclic USOs.  In
Section~\ref{s:substr}, we show how from representation maps for an
ample class $C$ to derive representation maps for substructures of $C$
(intersections with cubes, restrictions and reductions). In
Section~\ref{s:pre-rep}, we show that there exist maps satisfying each
one of the two conditions defining USOs (but not both). Finally, using
the geometric characterization of representation maps as USOs, we show
that constructing a representation map for an ample concept class can
be reduced to solving an instance of the Independent System of
Representatives problem~\cite{AhBeZi}.

\subsection{Unlabeled Sample Compression Schemes and Representation Maps}

In the next theorem, we prove that, analogously to maximum classes,
representation maps for ample classes lead to unlabeled sample
compression schemes of size $\vcdim(C)$. This also shows that the
representation maps for ample classes are equivalent to
$\Delta$-representation maps.

\begin{theorem}\label{t:clash_bis}\label{t:clashalt} Let
  $C\subseteq 2^{\U}$ be an ample class and let $r:C\to \uoX(C)$ be a
  bijection.  The following conditions are equivalent:
  \begin{enumerate}[{(R}1)]
  \item \label{i:or} \textsf{$\cup$-non-clashing}:
    For all distinct concepts $c',c''\in C$,
    $c'|r(c')\cup r(c'')\neq c''|r(c')\cup r(c'')$.
  \item \label{i:rec} \textsf{Reconstruction}:
    For every realizable sample $s$ of $C$, there is a unique $c\in C$
    that is consistent with $s$ and $r(c)\subseteq \dom(s)$.
  \item \label{i:inj} \textsf{Cube injective}: For every cube $B$ of
    $2^{\U}$, the map $c\mapsto r(c)\cap\supp(B)$ is an injection
    from  $C\cap B$ to $\uoX(C\cap B)$.
  \item \label{i:xor} \textsf{$\Delta$-non-clashing}: For all distinct
    concepts $c',c''\in C$,
    $c'|r(c')\Delta r(c'')\neq c''|r(c')\Delta r(c'')$.
  \end{enumerate}
  Moreover, any $\Delta$-non-clashing map $r: C \to \uoX(C)$ is
  bijective and is therefore a representation map.  Furthermore, if
  $r$ is a representation map for $C$, then there exists an unlabeled
  sample compression scheme for $C$ of size $\vcdim(C)$.
\end{theorem}

\begin{proof}
  Fix $Y\subseteq{\U}$ and partition $C$ into equivalence classes
  where two concepts $c,c'$ are equivalent if $c|Y=c'|Y$.  Thus, each
  equivalence class corresponds to a sample of $C$ with domain $Y$,
  i.e., a concept in $C|Y$. We first show that the number of such
  equivalence classes equals the number of concepts whose
  representation set is contained in $Y$:
  \begin{align*}
    \lvert C|Y\rvert
    &= \lvert{\oX(C|Y)}\rvert
    & \text{(Since $C|Y$ is ample)}\\
    &= \lvert{\oX(C)}\cap 2^Y\rvert & \\
    &=\lvert\{c:r(c)\subseteq Y\}\rvert
    & \text{(Since $r:C\rightarrow{\oX(C)=\uX(C)}$ is a bijection)}
  \end{align*}
  Condition (R2) asserts that in each equivalence class there is
  exactly one concept $c$ such that $r(c)\subseteq Y$.

  \medskip\noindent {(R1) $\Rightarrow$ (R2):}
  Assume $\neg (R2)$ and consider a sample $s$ for which the property
  does not hold. This implies that there exists an equivalence class
  with either zero or (at least) two equivalent concepts $c$ for which
  $r(c) \subseteq Y$ with $Y = \dom(s)$.  Note that since the number
  of equivalence classes equals the number of concepts whose
  representation set is contained in $Y$, if some equivalence class
  has no concept $c$ for which $r(c) \subseteq Y$, then there must be
  another equivalence class with two distinct concepts $c',c''\in C$
  for which $r(c'),r(c'')\subseteq Y$.  Therefore, in both cases,
  there exist two equivalent concepts $c, c'\in C$ such that
  $r(c),r(c')\subseteq Y$. Since $c |Y = c'|Y$, we have
  $c|r(c)\cup r(c') = c'|r(c)\cup r(c')$, contradicting $(R1)$.

  \medskip\noindent{(R2) $\Rightarrow$ (R1):}
  Assume $\neg (R1)$, i.e., for two distinct concepts $c',c''\in C$,
  we have $c'|r(c')\cup r(c'') = c''|r(c')\cup r(c'')$.  Now for the
  sample $s=c'|r(c')\cup r(c'')$, we have $\dom(s)=r(c')\cup
  r(c'')$. Furthermore, $c'|\dom(s) =c''|\dom(s)$ and
  $r(c'),r(c'')\subseteq \dom(s)$. This implies $\neg (R2)$.

  \medskip\noindent{(R1)$\&$(R2) $\Rightarrow$ (R3):} Since $C\cap B$
  is ample, it suffices to show that for every $Y\in \uoX(C\cap B)$
  there is some $c\in C\cap B$ with $r(c) \cap \supp(B) = Y$.  This is
  established by the following claim.

  \begin{claim}\label{claim_rm}
    Conditions (R1) and (R2) together imply that for any
    $Y\in \uoX(C\cap B)$, there exists a unique concept
    $c_Y\in C\cap B$ such that $r(c_Y)\cap \supp(B)=Y$.
  \end{claim}

  \begin{proof}
    For any concept $c \in C \cap B$, let
    $r_B(c):=r(c)\cap \supp(B)$.  Let
    $Z = {\U} \setminus \supp(B)$.  Note that all concepts in $B$
    agree on domain $Z$: indeed, $B|Z$ is the single sample $\tg(B)$.
    We prove the claim by induction on $\lvert Y\rvert$.

    % \smallskip\noindent\emph{Base case: $Y=\varnothing$.}
    Suppose first that $Y = \varnothing$. Since
    $\uoX(C\cap B)\ne \varnothing$, $C\cap B\ne \varnothing$ and
    $\tg(B)$ is a sample of $C$.  By condition (R2), there is a unique
    concept $c\in C$ such that (i) $c|Z=\tg(B)$ (i.e., $c\in B$), and
    (ii) $r(c)\subseteq Z$ (i.e., $r_B(c) = \varnothing$).  Thus
    choosing $c_\varnothing=c$ settles this case.

    % \smallskip\noindent\emph{Induction step: $Y\neq\varnothing$.}
    Assume now that $Y \neq \varnothing$ and that
    % By induction hypothesis,
    for every $V\subsetneq Y$, there is a
    unique $c_V\in C\cap B$ with $r_B(c_V)=r(c_V)\cap \supp(B)=V$.

    We assert that for any two distinct concepts $c, c' \in C\cap B$
    such that $r_B(c) \cup r_B(c') \subseteq Y$, we have
    $c| Y \neq c'| Y$.  Indeed, suppose by way of contradiction that
    $c|Y=c'|Y$ and note that
    $r_B(c) \cup r_B(c') =(r(c)\cup r(c'))\cap \supp(B) \subseteq Y$.
    If $c|Y=c'|Y$, then
    \begin{equation*}
      c|(r(c)\cup r(c'))\cap \supp(B)= c'|(r(c)\cup r(c'))\cap \supp(B).
    \end{equation*}
    Since $c |Z = c' | Z =
    \tg(B)$, this implies that $c|r(c)\cup r(c')= c'|r(c)\cup r(c')$,
    contradicting condition (R1).

    Consequently, all the samples $c_V|Y$ are pairwise distinct. There
    are $2^{\lvert Y\rvert} - 1$ such samples and each sample $c_V|Y$
    corresponds to a proper subset $V$ of $Y$. By the previous
    assertion, there exist at most $2^{\lvert Y\rvert}$ concepts $c$
    such that $r_B(c) \subseteq Y$, and thus there exists at most one
    concept $c_Y \in C \cap B$ such that $r_B(c_Y) = Y$.

    Thus, it remains to establish the existence of $c_Y\in C\cap B$
    such that $r_B(c_Y)= r(c_Y) \cap \supp(B)= Y$.  Since there are
    $2^{\lvert Y\rvert}-1$ $c_V$'s for $V\subsetneq Y$, it follows
    that there is a unique sample $s'$ with $\dom(s')=Y$ that is not
    realized by any of the $c_V$'s.  Consider the sample $s$ with
    domain $Y\cup Z$ defined by
    \begin{equation*}
      s(x) = \begin{cases}
        s'(x)  & \text{if } x\in Y,\\
        \tg(B)(x) & \text{if } x\in Z.
      \end{cases}
    \end{equation*}
    Since $Y$ is shattered by $C\cap B$, it follows that $s$ is
    realized by $C$.  By condition (R2) there is a unique concept
    $c\in C$ that agrees with $s$ such that $r(c)\subseteq Y\cup Z$
    (i.e.~$r(c)\cap \supp(B) \subseteq Y$).  We claim that $c$ is the
    desired concept $c_Y$. First notice that $c\in C\cap B$, because
    $c$ agrees with $\tg(B)$ on $Z$. Since $s$ is not realized by any
    $c_V$, $V\subsetneq Y$, and since $c_V$ is the unique concept of
    $B$ such that $r_B(c_V) = V$ (by induction hypothesis),
    necessarily we have that $c \neq c_V$ for any $V\subsetneq
    Y$. Consequently, $r_B(c) = Y$, concluding the proof of the claim.
  \end{proof}

  \medskip\noindent{(R3) $\Rightarrow$ (R4):} For any distinct
  concepts $c',c''\in C$, consider the minimal cube $B:=B(c',c'')$
  which contains both $c',c''$.  This means that $c'(x)\neq c''(x)$
  for every $x\in\supp(B)$, and that $c'(x)= c''(x)$ for every
  $x\notin\supp(B)$.  Condition (R3) guarantees that the map
  $r(c)\mapsto r(c)\cap \supp(B)$ is an injection from $C\cap B$ to
  $\uoX(C\cap B)$. Therefore
  $r(c')\cap \supp(B) \ne r(c'')\cap \supp(B)$.  It follows that there
  must be some $x\in\supp(B)$ such that
  $x\in (r(c')\cap \supp(B)) \Delta (r(c'')\cap \supp(B))$.  Since
  $x\in \supp(B)$, $c'(x)\ne c''(x)$ and therefore
  $c'|{r(c')\Delta r(c'')} \neq c''|{r(c')\Delta r(c'')}$ and
  condition (R4) holds for $c'$ and $c''$.

  \medskip\noindent{(R4) $\Rightarrow$ (R1):} This
  is immediate because if two concepts clash on their symmetric
  difference, then they also clash on their union.

  \medskip Moreover, observe that for any map $r: C \to \uoX(C)$, if
  $r(c) = r(c')$ for $c\neq c'$, then $r(c)\Delta r(c') = \varnothing$
  and $r$ is not $\Delta$-non-clashing. Consequently, any
  $\Delta$-non-clashing map $r: C \to \uoX(C)$ is injective and thus
  bijective since $\lvert C\rvert = \lvert \uoX(C)\rvert$.

  \medskip We now show that if $r: C \to \uoX(C)$ is a representation
  map for $C$ then there exists an unlabeled sample compression scheme
  for $C$. Indeed, by (R2), for each realizable sample $s \in \RS(C)$,
  let $\gamma(s)$ be the unique concept $c \in C$ such that
  $r(c) \subseteq \dom(s)$ and $c|\dom(s) = s$.  Then consider the
  compressor $\alpha: \RS(C) \to \uoX(C)$ such that for any
  $s \in \RS(C)$, $\alpha(s) = r(\gamma(s))$ and the reconstructor
  $\beta: \uoX(C) \to C$ such that for any $Z \in \uoX(C)$,
  $\beta(Z) = r^{-1}(Z)$. Observe that by the definition of
  $\gamma(s)$, $\alpha(s) \subseteq \dom(s)$ and
  $\beta(\alpha(s)) = \gamma(s)$ coincides with $s$ on
  $\dom(s)$. Consequently, $\alpha$ and $\beta$ defines an unlabeled
  sample compression scheme for $C$ of size
  $\dim(\uoX(C)) = \vcdim(C)$. This concludes the proof of the
  theorem.
\end{proof}

By applying (R3) to the $1$-dimensional cubes of $B$, we get the
following corollary.

\begin{corollary}\label{cor-xedges}
  For any representation map $r$ of an ample class $C$, for any
  $c \in C$ and $x \in r(c)$, we have $c \Delta \{x\} \in C$.
\end{corollary}

\subsection{Representation Maps as Unique Sink Orientations}\label{s-USOs}

We call a map $r: C \to 2^{\U}$ \emph{edge-non-clashing} if for any
$x$-edge $cc'$, $x \in r(c)\Delta r(c')$. Note that $r$ defines an
orientation~$o_r$ of the edges of $G(C)$: an $x$-edge $cc'$ is
oriented from $c$ to $c'$ if and only if $x\in r(c)\setminus
r(c')$. Conversely, given an orientation $o$ of the edges of $G(C)$,
the \emph{out-map} $r_o$ of $o$ associates to each $c\in C$ the
coordinate set of the edges outgoing from $c$. Note that the out-map
of any orientation is edge-non-clashing.  Note also that any
$\Delta$-representation map $r: C\rightarrow \uoX(C)$ is
edge-non-clashing and thus defines an orientation $o_r$ of
$G(C)$. Moreover, by Corollary~\ref{cor-xedges}, the out-map of $o_r$
coincides with $r$.

% We now show that $o_r$ satisfies extra local conditions on the stars
% $\St(c)$ of all concepts $c \in C$ (the \emph{star} $\St(c)$ is the
% set of all faces of the cubes containing $c$).

An orientation $o$ of the edges of $G(C)$ (or the corresponding
out-map $r_o$) is a \emph{unique sink orientation (USO)} if it
satisfies the following two conditions.
\begin{enumerate}[{(C}1)]
\item for any $c\in C$, the cube of $2^{\U}$ which has support
  $r_o(c)$ and contains $c$ is a cube of $C$, i.e., all outgoing
  neighbors of $c$ belong to a cube of $C$;
\item For any cube $B$ of $C$, there exists a unique $c \in C\cap B$
  such that $r_o(c) \cap \supp(B) = \varnothing$, i.e., $c$ is a sink
  in $G(C \cap B)$. %Consequently, $o_r$ is a USO on each cube of $C$.
\end{enumerate}

% The two statements for (C1) are indeed equivalent because if $c$ is
% contained in an $r_o(c)$-cube, then for any $x \in r(c)$,
% $c \Delta \{x\}$ is a concept of $C$.
Observe also that a map $r: C \to 2^{\U}$ satisfying Condition (C2) is
necessarily an edge-non-clashing map.  If $C$ is a cube, then
Condition (C1) trivially holds and Condition (C2) corresponds to the
usual definition of USOs on cubes~\cite{SzWe}.
In the following, we use the characterization of USOs for cubes given
in~\cite{SzWe}.
%% and its immediate consequences.

\begin{lemma}[\!\!\cite{SzWe}]\label{lem-SzWe}
  For a cube $B$ and a map $r : B \to 2^{\supp(B)}$, the following are
  equivalent:
  \begin{enumerate}[(1)]
  \item $r$ is the out-map of a unique sink orientation of $B$;
  \item $r$  is $\Delta$-non-clashing;
  \item for any subcube $B'$ of $B$, the map
    $c \mapsto r(c) \cap \supp(B')$ is a bijection between $B'$ and
    $2^{\supp(B')}$;
  \item $r$ is the out-map of a \emph{unique source orientation} of
    $B$, i.e., for any subcube $B'$ of $B$, there exists a unique
    $c \in B'$ such that $r(c) \cap \supp(B') = \supp(B)$.
  \end{enumerate}
  % if and
  % only if $r$ is $\Delta$-non-clashing.
  % Moreover, the out-map of a unique sink orientation of a cube $B$
  % is
  % a bijection between $B$ and $2^{\supp(B)}$.
  % an orientation $o$ of the edges of $B$ and the
  % corresponding out-map $r_o: B \to 2^{\supp(B)}$, the following
  % conditions are equivalent:
  % \begin{enumerate}[(i)]
  % \item $o$ is a unique sink orientation of $B$;
  % \item $r_o$ is a bijective;
  % \item $r_o$ is $\Delta$-non-clashing.
  % \end{enumerate}
\end{lemma}

The equivalences (1)$\Leftrightarrow$(2) and (1)$\Leftrightarrow$(4)
are respectively~\cite[Lemma~2.3]{SzWe}
and~\cite[Lemma~2.1]{SzWe}. The implications (2)$\Rightarrow$(3) and
(3)$\Rightarrow$(1) are trivial.

\begin{remark}
  In view of Lemma~\ref{lem-SzWe}, Condition (C2) looks similar to
  Condition (R3) of Theorem~\ref{t:clash_bis}. Note however that (C2)
  is about the cubes of $2^{\U}$ contained in $C$ while (R3) is about
  all cubes of $2^{\U}$. Similarly, Condition (C2) implies that $r_o$
  is $\Delta$-non-clashing on each cube of $2^{\U}$ contained in $C$
  while condition (R4) of Theorem~\ref{t:clash_bis} requires that
  $r_o$ is $\Delta$-non-clashing on $C$.
\end{remark}

\begin{remark}\label{rem-USO-subcubes}
  If $C$ is an ample class and $o$ is a USO of $G(C)$, then for any
  cube $B$ of $C$, the restriction of $o$ to the edges of $G(B)$
  trivially satisfies (C1) and (C2) and is thus a USO of
  $G(B)$. Consequently, its out-map
  $r_B: c \mapsto r_o(c) \cap \supp(B)$ from $B$ to $2^{\supp(B)}$
  satisfies the conditions of Lemma~\ref{lem-SzWe}.
\end{remark}
% Note that this implies that for any subcube $B'$ of $B$, the out-map
% of a USO induces a bijection between $B'$ and $2^{\supp(B')}$. In
% particular, this implies that each subcube $B'$ has a unique source.

We now show that representation maps for ample classes give rise to
USOs.

% We now show that if $r$ is a representation map for an ample class
% $C$, then $o_r$ is a unique sink orientation.

\begin{corollary}\label{c:cube}
  If $r: C\rightarrow \uoX(C)$ is a representation map for an ample
  class $C\subseteq 2^{\U}$, then $o_r$ is a unique sink orientation.
  % satisfy the following
  % two conditions:
  % \begin{enumerate}[{(C}1)]
  % \item for any $c\in C$, the cube of $2^{\U}$ which has support
  %   $r(c)$ and contains $c$ is a cube of $C$, i.e., all outgoing
  %   neighbors of $c$ belong to a cube of $C$;
  % \item For any cube $B$ of $C$, there exists a unique
  %   $c \in C\cap B$
  %   such that $r(c) \cap \supp(B) = \varnothing$, i.e., $c$ is a
  %   sink
  %   in $G(C \cap B)$. Consequently, $o_r$ is a USO on each cube of
  %   $C$.
  % \end{enumerate}
\end{corollary}

\begin{proof}
  Pick some concept $c \in C$ and consider the  unique cube $B$ of
  $2^{\U}$ that contains $c$ and has support $\supp(B) = r(c)$.  Since
  $r$ is a representation map, by Condition (R3) of
  Theorem~\ref{t:clashalt}, the map $c' \mapsto r(c') \cap \supp(B)$ is
  a bijection between the ample set $C\cap B$ and $X(C \cap
  B)$. Consequently, $r(c) = \supp(B) \in X(C \cap B)$, which is
  possible only if $C \cap B = B$.
  % Therefore, there exists $Y\in \uoX(C\cap B)$ such that
  % $Y=r(c)\cap \supp(B)=\supp(B)$ (since $r(c)=\supp(B)$).  Hence
  % $\supp(B)\in \uoX(C\cap B)$, which is possible only if $C\cap
  % B=B$.
  This proves (C1).  Condition (C2) follows from Condition (R4)
  of Theorem~\ref{t:clashalt} applied to the cubes of $C$ and
  Lemma~\ref{lem-SzWe}.
\end{proof}

We continue with a characterization of representation maps of ample
classes as out-maps of USOs, extending a similar result of Szab\'{o}
and Welzl~\cite{SzWe} for cubes.  This characterization is
``local-to-global'', since (C1) and (C2) are conditions on the cubes
around each concept $c \in C$.

% An orientation $o$ of the edges of $G(C)$ is a \emph{unique sink
% orientation (USO)} if $o$ satisfies (C1) and (C2).  The
% \emph{out-map} $r_o$ of an orientation $o$ associates to each
% $c\in C$ the coordinate set of the edges outgoing from $c$. We
% continue with a characterization of representation maps {of ample
% classes as out-maps of USOs, extending a similar result of Szab\'{o}
% and Welzl~\cite{SzWe} for cubes.}  {This characterization is
% ``local-to-global'', since (C1) and (C2) are conditions on the stars
% $\St(c)$ of all concepts $c \in C$.}

\begin{theorem}\label{t:local-to-global}
  For an ample class $C$ and a %n edge-non-clashing
  map $r:C\to 2^{\U}$,
  the following are equivalent:
  \begin{enumerate}[(1)]
  \item $r$ is a representation map;
  \item $r$ is the out-map of a USO;
  \item $r(c)\in \uoX(C)$ for any $c\in C$ and $r$ satisfies (C2).
  \end{enumerate}
\end{theorem}

Before proving the theorem, starting from an edge-non-clashing map for
an ample class $C$, we show how to derive maps for restrictions $C_x$,
reductions $C^x$ and intersections $C \cap B$ with cubes of $2^{\U}$.
Consider an edge-non-clashing map $r: C\rightarrow \uoX(C)$ for an
ample class $C$.  Given a cube $B$ of $2^{\U}$, define
$r_B : C\cap B \to \uoX(C\cap B)$ by setting
$r_B(c):= r(c)\cap \supp(B)$ for any $c \in C\cap B$. Note that $r_B$
is the out-map of the orientation $o_r$ restricted to the edges of
$G(C\cap B)$.

Given $x \in \U = \dom(C)$, define
$r^x: C^x \mapsto 2^{\U\setminus\{x\}}$ by
\begin{equation*}
  r^x(c) = \begin{cases}
    r(c)\setminus\{x\}  & \text{ if } x\in r(c),\\
    r(c \cup \{x\})\setminus\{x\} & \text{ otherwise.}
  \end{cases}
\end{equation*}
Hence, $r^x(c\setminus \{x\}) = r(c) \setminus \{x\}$ for each
$c \in C$ with $x \in r(c)$. Consequently, for an $x$-edge of $C$
between $c$ and $c \cup \{x\}$, $r^x(c)$ is the label of the origin of
this edge minus $x$; we call $r^x$ the \emph{$x$-out-map} of $r$.

Given $x \in \U = \dom(C)$, define
$r_x: C_x \mapsto 2^{\U\setminus\{x\}}$ by
\begin{equation*}
  r_x(c) = \begin{cases}
    r(c)  & \text{ if } x\notin r(c),\\
    r(c \cup \{x\}) & \text{ otherwise.}
  \end{cases}
\end{equation*}
Hence, $r_x(c\setminus \{x\}) = r(c)$ for each $c \in C$ with
$x \notin r(c)$. Consequently, for an $x$-edge of $C$ between $c$ and
$c \cup \{x\}$, $r_x(c)$ is the label of the destination of this edge;
we call $r_x$ the \emph{$x$-in-map} of $r$.

In the next lemma, we show that if the map $r$ is the out-map of a
USO, then $r_B$ and $r^x$ are also outmaps of USOs. A similar result
holds for $r_x$ but it will be proved in Section~\ref{s:substr}. The
first assertion of the lemma generalizes
Remark~\ref{rem-USO-subcubes}.
\begin{lemma}\label{lem-substructure}
  Consider a map $r: C\rightarrow 2^{\U}$ that is the out-map of a USO
  of an ample class $C$.  For any cube $B \subseteq 2^{\U}$ and any
  $x \in \U = \dom(C)$, the following hold:
  \begin{enumerate}[(1)]
  \item $r_B$ is the out-map of a USO of $C \cap B$;
  \item $r^x$ is the out-map of a USO of $C^x$;
%  \item $r_x$ is the out-map of a USO of $C_x$.
  \end{enumerate}
\end{lemma}

\begin{proof}
  Item~(1):
  Consider a cube $B \subseteq 2^{\U}$.  Since (C1) holds for $C$ and
  $r$, for any concept $c \in C\cap B$, the cube $B(c)$ that has
  support $r(c)$ and contains $c$ is a cube of $C$. Since the
  intersection of two cubes is a cube, the cube $B'(c) = B(c) \cap B$
  has support $r(c) \cap \supp(B) = r'(c)$ and contains
  $c$. Consequently, (C1) holds for $C'$ and $r'$. Since any cube of
  $C' = C\cap B$ is also a cube of $C$ and since (C2) holds for $C$
  and $r$, (C2) also holds for $C'$ and $r'$.

  \medskip\noindent Item~(2): Consider $c \in C^x$ and let
  $c' \in \{c, c\cup \{x\}\}$ such that
  $r^x(c) = r(c') \setminus \{x\}$.  By (C1), there exists an
  $r(c')$-cube $B'$ containing $c'$ in $C$. Thus, there exists an
  $(r(c')\setminus \{x\})$-cube containing $c$ in $C^x$. Consequently,
  $r^x$ satisfies (C1).

  Suppose that there exists a cube $B$ in $C^x$ violating (C2), i.e.,
  there exist $c_1, c_2 \in C^x \cap B$ such that
  $r^x(c_1)\cap \supp(B) = r^x(c_2)\cap \supp(B)$. By definition of
  $r^x$, there exist
  $c'_1 \in \{c_1, c_1 \cup \{x\}\}, c'_2 \in \{c_2, c_2 \cup \{x\}\}$
  such that $r(c'_1) = r^x(c_1) \cup \{x\}$ and
  $r(c'_2) = r^x(c_2) \cup \{x\}$. Since $B$ is a cube of $C^x$
  containing $c_1$ and $c_2$, there exists a cube $B'$ of $C$
  containing $c'_1$ and $c'_2$ such that
  $\supp(B') = \supp(B)\cup \{x\}$.  Consequently,
  $r(c'_1) \cap \supp(B') = (r^x(c_1) \cup \{x\}) \cap (\supp(B) \cup
  \{x\}) = (r^x(c_1)\cap \supp(B)) \cup \{x\}$ and similarly,
  $r(c'_2) \cap \supp(B') = (r^x(c_2)\cap \supp(B)) \cup
  \{x\}$. Therefore, $r(c'_1) \cap \supp(B') = r(c'_2) \cap \supp(B')$
  and $r$ is not injective on the cube $B'$ of $C$, contradicting
  Lemma~\ref{lem-SzWe} (see Remark~\ref{rem-USO-subcubes}).

  % \medskip\noindent Item~(3): Consider $c \in C_x$ and let
  % $c' \in \{c, c\cup \{x\}\}$ such that $r_x(c) = r(c')$.  By (C1),
  % there exists an $r(c')$-cube $B'$ containing $c'$ in $C$. Thus,
  % there exists an $r(c')$-cube containing $c$ in $C_x$. Consequently,
  % $r_x$ satisfies (C1).

  % Suppose that there exists a cube $B$ in $C_x$ violating (C2), i.e.,
  % there exist $c_1, c_2 \in C_x \cap B$ such that
  % $r_x(c_1)\cap \supp(B) = r_x(c_2)\cap \supp(B)$. Let
  % $\sigma = r_x(c_1)\cap \supp(B)$. By definition of $r_x$, there
  % exist
  % $c'_1 \in \{c_1, c_1 \cup \{x\}\}, c'_2 \in \{c_2, c_2 \cup \{x\}\}$
  % such that $r(c'_1) = r_x(c_1)$ and $r(c'_2) = r_x(c_2) \cup
  % \{x\}$.

  % Note that if $c'_1$ and $c'_2$ belong to a $\sigma$-cube of
  % $C$, then $r$ does not satisfy (C2), a contradiction.

  % Since $B$ is a cube of $C_x$
  % containing $c_1$ and $c_2$, there exists a cube $B'$ of $2^U$
  % containing $c'_1$ and $c'_2$ such that
  % $\supp(B') = \supp(B)\cup \{x\}$. Observe that if $x \notin c'_1
  % \Delta c'_2$,

  % Consequently,
  % $r(c'_1) \cap \supp(B') = (r^x(c_1) \cup \{x\}) \cap (\supp(B) \cup
  % \{x\}) = (r^x(c_1)\cap \supp(B)) \cup \{x\}$ and similarly,
  % $r(c'_2) \cap \supp(B') = (r^x(c_2)\cap \supp(B)) \cup
  % \{x\}$. Therefore, $r(c'_1) \cap \supp(B') = r(c'_2) \cap \supp(B')$
  % and $r$ is not injective on the cube $B'$ of $C$, contradicting (R3)
  % for $r$.
\end{proof}

\begin{proof}[Proof of Theorem~\ref{t:local-to-global}]
  The implication {(1)}$\Rightarrow${(2)} is
  established in Corollary~\ref{c:cube}.  Now, we prove
  {(2)}$\Rightarrow${(1)}.  Clearly, property
  (C1) implies that $r(c)\in \uoX(C)$ for any $c\in C$, whence $r$ is
  a map from $C$ to $\uoX(C)$.

  % \begin{claim}\label{claim-intersection-cubes}
  %   Consider an ample class $C$ and a map $r: C \to 2^{\U}$ satisfying
  %   (C1) and (C2). For any cube $B$ of $2^{\U}$, $C' = C\cap B$ and
  %   the restriction $r'$ of $r$ to $C'$ satisfy (C1) and (C2).
  % \end{claim}

  % \begin{proof}
  %   Since (C1) holds for $C$ and $r$, for any concept $c \in C\cap B$,
  %   the cube $B(c)$ that has support $r(c)$ and contains $c$ is a cube
  %   of $C$. Since the intersection of two cubes is a cube, the cube
  %   $B'(c) = B(c) \cap B$ has support $r(c) \cap \supp(B) = r'(c)$ and
  %   contains $c$. Consequently, (C1) holds for $C'$ and $r'$. Since
  %   any cube of $C' = C\cap B$ is also a cube of $C$ and since (C2)
  %   holds for $C$ and $r$, (C2) also holds for $C'$ and $r'$.
  % \end{proof}

  Let $C$ be an ample class of smallest size admitting a
  non-representation map $r: C\to \uoX(C)$ satisfying (C1) and
  (C2). Hence there exist $u_0,v_0\in C$ such that
  $u_0 \vert (r(u_0)\Delta r(v_0))=v_0\vert(r(u_0) \Delta r(v_0))$,
  i.e., $(u_0\Delta v_0)\cap (r(u_0)\Delta r(v_0))=\varnothing$;
  $(u_0,v_0)$ is called a \emph{clashing pair}.

  \begin{claim} \label{claim_cp1}
    If  $(u_0,v_0)$ is a clashing pair,
    then $C=C\cap B(u_0,v_0)$ and  $r(u_0)=r(v_0)=\varnothing$.
  \end{claim}

  \begin{proof}
    Since $C\cap B(u_0,v_0)$ is ample and
    $(u_0\Delta v_0)\cap (r(u_0)\Delta r(v_0))=\varnothing$,
    $(u_0,v_0)$ is a clashing pair for $C\cap B(u_0,v_0)$ and the
    restriction $r_B$ of $r$ to $\supp(B(u_0,v_0))$. By
    Lemma~\ref{lem-substructure}, $r_B$ is the out-map of a USO of
    $C\cap B(u_0,v_0)$. Consequently, by minimality of $C$,
    $C=C\cap B(u_0,v_0)$ and thus
    $\dom(C) = \supp(B(u_0,v_0)) = u_0 \Delta v_0$. Moreover, if
    $r(u_0) \neq r(v_0)$, then there is $x \in r(u_0)\Delta r(v_0)$
    and $x \in \supp(B(u_0,v_0)) = u_0\Delta v_0$, contradicting that
    $(u_0,v_0)$ is a clashing pair.

    Suppose $r(u_0)\ne \varnothing$ and pick $x \in r(u_0)=r(v_0)$.
    % Consider the carrier $N_x(C)$ of $C^x$. We show that
    % $r(u_0)\subseteq \supp(N_x(C))$.  Trivially $x \in N_x(C)$. Let
    % $y\in r(u_0)\setminus \{x\}$. By (C1), $u_0$ belongs to an
    % $\{ x,y\}$-square of $C$, whence $y\in
    % \supp(N_x(C))$.
    % Consequently, $(u_0,v_0)$ is a clashing pair for
    % $N_x(C)$ and the restriction of $r$ to $N_x(C)$. $N_x(C)$ is ample
    % as the product of $C^x$ by an $x$-edge. By minimality of $C$,
    % $C=N_x(C)$.
    % Define $r^x: C^x \mapsto 2^{\U\setminus\{x\}}$ by
    % \begin{equation*}
    %   r^x(c) = \begin{cases}
    %     r(c)\setminus\{x\}  & \text{ if } x\in r(c),\\
    %     r(c \cup \{x\})\setminus\{x\} & \text{ otherwise.}
    %   \end{cases}
    % \end{equation*}
    % Consequently, for an $x$-edge of $C$ between $c$ and
    % $c \cup \{x\}$, $r^x(c)$ is the label of the origin of this edge
    % minus $x$; we call $r^x$ the \emph{$x$-out-map} of $r$.
    % Observe
    % that $r^x$ coincides with $r^{\{x\}}$ and
    Consider the $x$-out-map $r^x: C^x \to 2^{\U \setminus \{x\}}$.
    By Lemma~\ref{lem-substructure}, $r^x$ is the out-map of a USO for
    $C^x$.
    % We show
    % that $C^x$ and $r^x$ satisfy (C1).  For any $c \in C^x$, let
    % $c' \in \{c,c \cup \{x\}\}$ such that
    % $ r(c') = r^x(c) \cup \{x\}$. Since (C1) holds for $C$ and $r$,
    % $c'$ belongs to a cube $B$ of $C$ with support
    % $r^x(c) \cup \{x\}$. Since $B^x$ is a cube of $C^x$ with support
    % $r^x(c)$ containing $c$, (C1) holds for $C^x$ and $r^x$.  To
    % establish condition (C2), suppose that there exists a cube $B'$
    % of
    % $C^x$ and $u',v' \in B'$ such that
    % $r^x(u')\cap \supp(B')=r^x(v')\cap \supp(B')$.  The cube
    % $B:=B'\times \{x\}$ is included in $C$ since $B'$ is a cube of
    % $C^x$. Then among the four pairs
    % $(u',v'), (u',v'\cup\{x\}), (u'\cup\{x\},v'), (u'\cup\{x\},
    % v'\cup\{x\})$ of $B$ one can select a pair $(u,v)$ such that
    % $r(u)= r^x(u') \cup \{x\} = r^x(v') \cup \{x\} = r(v)$, a
    % contradiction with condition (C2) for $C$ and $r$. This
    % established that $C^x$ and $r^x$ satisfy (C2).
    Suppose without loss of generality that $x\in v_0\setminus
    u_0$. Let $u'_0=u_0$ and $v'_0=v_0 \setminus \{x\}$. Then
    $r^x(u'_0)=r(u_0)\setminus \{x\}=r(v_0)\setminus \{x\}=r^x(v'_0)$,
    and consequently $(u'_0,v'_0)$ is a clashing pair for $r^x$ on
    $C^x$. Since $C^x$ is ample, smaller than $C$, and since $r^x$ is
    the out-map of a USO of $C^x$, this contradicts the minimality of
    $C$.
\end{proof}

\begin{claim}\label{claim_cp2}
  $C$ is a cube minus a vertex.
\end{claim}

\begin{proof}
  By (C2) and Claim~\ref{claim_cp1}, $C$ is not a cube. If $C$ is not
  a cube minus a vertex, since the complement $C^*=2^{\U}\setminus C$
  is also ample (thus $G(C^*)$ is connected), $G(C^*)$ contains an
  $x$-edge $ww'$ with $x\notin w$ and $x\in w'$.  Consider the
  $x$-in-map $r_x: C_x \mapsto 2^{\U\setminus\{x\}}$.

  For any $u' \in C_x$, let $u \in \{u',u'\cup\{x\}\}$ such that
  $r_x(u') = r(u)$. Since $C$ satisfies (C1), $u$ belongs to an
  $r(u)$-cube in $C$ and consequently, $u'$ belongs to an
  $r_x(u')$-cube in $C_x$. Thus $C_x$ and $r_x$ satisfy (C1).

  Suppose that $C_x$ and $r_x$ violate (C2).  Then there exists a cube
  $B'$ of $C_x$ and $u',v'\in B'$ such that
  $(u'\Delta v') \cap (r_x(u')\Delta r_x(v'))= \varnothing$. Without
  loss of generality, we can assume that $B' = B(u',v')$. Let
  $u \in \{u', u'\cup\{x\}\}$ such that $r(u)=r_x(u')$ and let
  $v \in \{v', v'\cup\{x\}\}$ such that $r(v)=r_x(v')$.  The
  restriction $C'_x$ of the ample class $C':=C\cap B(u,v)$ is the cube
  $B'$. By Lemma~\ref{lem-substructure}, $r_{B(u,v)}$ is the out-map
  of an USO of $C'$.  Since $w , w' \notin C$ and $ww'$ is an
  $x$-edge, $w \notin C'_x$. Thus there exists $y\in \supp(C)$ such
  that $C'$ and the edge $ww'$ of $C^*$ belong to different
  $y$-half-spaces $C^-=\{ c\in C: y\notin c\}$ and
  $C^+=\{ c\in C: y\in c\}$ of the cube $2^{\U}$.  Since
  $y\in \supp(C)$, the half-space containing $ww'$ also contains a
  concept of $C$.  Hence, $C'$ is a proper ample subset of $C$. Since
  $u\subseteq u'\cup \{x\}$, $v\subseteq v'\cup \{x\}$,
  $x\notin r(u)=r_x(u')$, $x \notin r(v)=r_x(v')$, we deduce that
  $u\cap (r(u)\Delta r(v))=u'\cap (r_x(u')\Delta r_x(v'))$ and
  $v' \cap (r_x(u') \Delta r_x(v'))= v\cap (r(u)\Delta r(v))$. Since
  $(u'\Delta v') \cap (r_x(u')\Delta r_x(v'))= \varnothing$, $(u,v)$
  is a clashing pair for the restriction of $r$ on $C'$, contrary to
  the minimality of $C$. Consequently, $C_x$ and $r_x$ satisfy (C2)
  and $r_x$ is the outmap of an USO of $C_x$. By minimality of $C$,
  $r_x$ is a representation map for $C_x$.

  Consider a clashing pair $(u_0,v_0)$ for $C$ and $r$, and let
  $u'_0 = u_0 \setminus\{x\}$ and $v'_0 = v_0 \setminus\{x\}$. By
  Claim~\ref{claim_cp1},
  $r_x(u'_0) = r(u_0) = r(v_0) = r_x(v_0') = \varnothing$. Since $r_x$
  is a representation map for $C_x$, necessarily $u_0' =
  v_0'$. Consequently, $u_0\Delta v_0 = \{x\}$, i.e., $u_0v_0$ is an
  $x$-edge of $G(C)$. This is impossible since $C$ satisfies
  (C2). Therefore, $C$ is necessarily a cube minus a vertex.
\end{proof}

Now, we complete the proof of the implication {(2)}$\Rightarrow${(1)}.
By Claim~\ref{claim_cp1}, $r(u_0)=r(v_0)=\varnothing$. By condition
(C1) and Claim~\ref{claim_cp2}, $r(c) \neq \U$ for any $c\in C$.  Thus
there exists a set
$s \in \uoX(C) = 2^{\U}\setminus \{\U,\varnothing\}$ such that
$s \neq r(c)$ for any $c\in C$.  Every $s$-cube $B$ of $C$ contains a
source $p(B)$ for $o_{r_B}$ (i.e., $s \subseteq r(p(B))$). For each
$s$-cube $B$ of $C$, let $t(B) = r(p(B))\setminus s$.  Notice that
$\varnothing \subsetneq t(B) \subsetneq \U\setminus s$ since
$s \subsetneq r(p(B)) \subsetneq \U$. Consequently, the number of
distinct sets $t(B)$ (when $B$ runs over the $s$-cubes of $C$) is at
most $2^{\lvert \U\rvert -\lvert s\rvert}-2$. On the other hand, since
$C$ is a cube minus one vertex by Claim~\ref{claim_cp2}, there are
$2^{\lvert\U\rvert - \lvert s\rvert}-1$ $s$-cubes in
$C$. Consequently, there exist two $s$-cubes $B, B'$ such that
$t(B)=t(B')$. Thus
$\varnothing \subsetneq s \subsetneq r(p(B)) = r(p(B'))$ and
$(p(B),p(B'))$ is a clashing pair for $C$ and $r$, contradicting
Claim~\ref{claim_cp1}.

The implication {(2)}$\Rightarrow${(3)} is trivial. To prove
{(3)}$\Rightarrow${(2)}, we show by induction on $\lvert \U\rvert$
that a map $r: C\rightarrow \uoX(C)$ satisfying (C2) also satisfies
(C1).  For any $x\in \U$, consider the $x$-out-map $r^x$. Recall that
if $cc'$ is an $x$-edge directed from $c$ to $c'$, then $x\in r(c)$
and $r^x$ maps $c^x=(c')^x\in C^x$ to
$r(c)\setminus \{x\}\in \uoX(C^x)$. Thus $r^x$ maps $C^x$ to
$\uoX(C^x)$.  Moreover, each cube $B^x$ of $C^x$ is contained in a
unique cube $B$ of $C$ such that $\supp(B) = \supp(B^x) \cup \{x\}$.
If there exist $c_1^x, c_2^x \in B^x$ such that
$r^x(c_1^x) = r(c_2^x)$, then there exist $c_1,c_2 \in B$ such that
$r(c_1) = r^x(c_1^x) \cup \{x\} = r^x(c_2^x) \cup \{x\} = r(c_2)$,
contradicting (C2). Consequently, $o_{r^x}$ satisfies (C2). By
induction hypothesis, $o_{r^x}$ satisfies (C1) for any $x \in \U$.

For any concept $c \in C$, pick $x \in r(c)$.  Since $r^x$ satisfies
(C1), $c^x$ belongs to a $\sigma'$-cube in $C^x$ with
$\sigma' =r^x(c^x) = r(c) \setminus\{x\}$.  This implies that $c$
belongs to a $\sigma$-cube in $C$ with
$\sigma = \sigma' \cup \{x\} = r(c)$. Thus $o_r$ satisfies (C1),
concluding the proof of Theorem~\ref{t:local-to-global}.
\end{proof}

A consequence of Theorems~\ref{t:clashalt} and~\ref{t:local-to-global}
is that corner peelings correspond \emph{exactly} to acyclic unique
sink orientations.

\begin{proposition}\label{prop-acyclicUSO}
  An ample class $C$ admits a corner peeling if and only if there
  exists an acyclic orientation $o$ of the edges of $G(C)$ that is a
  unique sink orientation.
\end{proposition}

\begin{proof}
  Suppose that $C_< = (c_1, \ldots, c_m)$ is a corner peeling and
  consider the orientation $o$ of $G(C)$ where an edge $c_ic_j$ is
  oriented from $c_i$ to $c_j$ if and only if $i > j$. Clearly, this
  orientation is acyclic. For any $i$, since $c_i$ is a corner in
  $C_i = \{c_1,\ldots,c_i\}$, the outgoing neighbors of $c_i$ belong
  to a cube of $C_i$, i.e., $o$ satisfies (C1). For any cube $B$ of
  $C$, assume that $c_{i_B}$ is the first concept of $B$ in the
  ordering $C_<$. Observe that $c_{i_B}$ is a sink of $B$ for the
  orientation $o$. Note that for each $i > i_B$,
  $\emptyset \subsetneq C_{i-1} \cap B \subseteq C_i \cap B$. By
  Theorem~\ref{thm:ample1}(6), $C_i \cap B$ is ample and thus
  connected. Consequently, for each $i > i_B$, there exists $c_ic_j$
  in $G(C)$ with $i_B \leq j < i$ such that $c_j \in B$. Consequently,
  since $c_ic_j$ is oriented from $c_i$ to $c_j$, $c_i$ is not a sink
  of $B$. Therefore every cube $B$ has a unique sink for the
  orientation $o$ and $o$ is an acyclic unique sink orientation of
  $G(C)$.

  Suppose now that $G(C)$ admits an acyclic unique sink orientation
  $o$. Consider a concept $c$ that is a source for $o$. By (C1) all
  the neigbors of $c$ belong to a cube of $C$ and consequently, $c$ is
  a corner of $C$. Since $C$ is ample, by Lemma~\ref{ample_bis},
  $C' = C\setminus\{c\}$ is ample. Clearly the restriction $o'$ of the
  orientation $o$ to $C'$ is acyclic. We claim that $o'$ is also a USO
  for $C'$. Observe that any cube $B$ of $C'$ is also a cube of $C$
  and consequently, $o'$ satisfies (C2) since $o$ satisfies (C2).
  Suppose now that there exists $c' \in C'$ such that the outgoing
  neighbors of $c'$ in $C'$ do not belong to a cube of $C'$. Then
  since $o$ is a USO for $C$, necessarily the outgoing neighbors of
  $c'$ in $C$ belong to a cube $B$ of $C$ that contains $c$. But then
  $c$ and $c'$ are both sources of $B$ for the orientation $o$. By
  Lemma~\ref{lem-SzWe}, this implies that $o$ does not satisfy (C2) on
  $B$ and thus on $C$ (see Remark~\ref{rem-USO-subcubes}), a
  contradiction. Therefore, $o'$ satisfies (C1) and (C2) and is a
  unique sink orientation of $C'$. Applying the previous argument
  inductively we obtain a corner peeling of $C$.
\end{proof}

\begin{remark}\label{rem-rm-CH-acyclic}
  A consequence of Proposition~\ref{prop-acyclicUSO} is that for any
  representation map for Hall's concept class $C_H$, the corresponding
  USO of $G(C_H)$ contains a directed cycle and thus at least one
  non-trivial strongly connected component.

  The USOs corresponding to the two representation maps of $C_H$
  presented in the appendix contain one or two non-trivial strongly
  connected components.
\end{remark}

\subsection{Representation Maps for Substructures}\label{s:substr}

In this subsection, from a representation map $r$ for an ample class
$C$, we show how to derive representation maps for restrictions $C_Y$,
reductions $C^Y$ and intersections $C\cap B$ with cubes of $2^{\U}$.

% Let $r: C\rightarrow \uoX(C)$ be a representation map for an ample
% class $C$.  Given a cube $B$ of $2^{\U}$, define
% $r_B : C\cap B \to \uoX(C\cap B)$ by setting
% $r_B(c):= r(c)\cap \supp(B)$ for any $c \in C\cap B$. Note that $r_B$
% is the out-map of the orientation $o_r$ restricted to the edges of
% $G(C\cap B)$.

Given a subset $Y \subseteq \U = \dom(C)$, define
$r^Y: C^Y \to \uoX(C^Y)$ as follow. For any $c \in C^Y$, there exists
a unique $Y$-cube $B$ in $C$ such that $\tg(B) = c$.  By
Lemmas~\ref{lem-SzWe} and~\ref{lem-substructure}(1), there exists a
unique $c^B \in B$ such that $r_B(c^B) = r(c^B)\cap Y = Y$ ($c^B$ is
the source of $B$ for $o_r$). We set $r^Y(c) := r(c^B)\setminus Y$. If
$Y = \{x\}$, $r^Y$ coincides with the $x$-out-map $r^x$ defined in
Section~\ref{s-USOs}.

Given a subset $Y \subseteq \U = \dom(C)$, define
$r_Y: C_Y \to \uoX(C_Y)$ as follow. For any $c \in C_Y$, there exists
a unique $Y$-cube $B$ in $2^{\U}$ such that $\tg(B) = c$. Since
$c \in C_Y$, $C \cap B \neq \varnothing$.  By Claim~\ref{claim_rm},
there exists a unique $c_B \in C \cap B$ such that
$r_B(c_B) = r(c_B)\cap Y = \varnothing$ ($c_B$ is the unique sink of
$C \cap B$ for $o_r$). We set $r_Y(c) := r(c_B)$. If $Y = \{x\}$,
$r_Y$ coincides with the $x$-in-map $r_x$ defined in
Section~\ref{s-USOs}.

\begin{proposition}\label{p:substructure}
  For a representation map $r$ for an ample class $C$, any cube $B$ of
  $2^{\U}$, and any $Y\subseteq {\U}$, the following hold:
  \begin{enumerate}[(1)]
  \item $r_B$ is a representation map for $C\cap B$;
  \item $r^Y$ is a representation map for $C^Y$;
  \item $r_Y$ is a representation map for $C_Y$.
  \end{enumerate}
\end{proposition}

\begin{proof}
  By Theorem~~\ref{t:local-to-global}, Assertion (1) follows from
  Assertion (1) of Lemma~\ref{lem-substructure} and Assertion (2)
  follows by iteratively applying Assertion (2) of
  Lemma~\ref{lem-substructure} to the elements of $Y$.

  % Item~(1): By Theorem~\ref{t:local-to-global}, the
  % orientation $o_r$ satisfies conditions (C1) and (C2) and obviously
  % its restriction to the edges of $G(C\cap B)$ still satisfies these
  % two conditions. Therefore, $r_B$ is a representation map for
  % $C\cap B$.

  % \medskip\noindent Item~(2): Recall that for any $c \in C^Y$,
  % we have $r^Y(c) = r(c^B)\setminus Y$, where $B$ is the unique
  % $Y$-cube such that $\tg(B) = c$ and $c^B$ is the unique source of
  % $B$ for $o_r$. By (C1), there exists an $r(c^B)$-cube $B'$
  % containing $c^B$ in $C$. Thus, there exists an
  % $(r(c^B)\setminus Y)$-cube containing $c$ in $C^Y$ and thus
  % $r^Y(c) \in \uoX(C^Y)$. Consequently, $r^Y$ is a map from $C^Y$ to
  % $\uoX(C^Y)$ and $r^Y$ satisfies (C1).

  % Suppose that there exists a cube $B$ in $C^Y$ violating (C2), i.e.,
  % there exist $c_1, c_2 \in C^Y \cap B$ such that
  % $r^Y(c_1)\cap \supp(B) = r^Y(c_2)\cap \supp(B)$.  Any cube $B$ of
  % $C^Y$ extends to a unique cube $B'$ of $C$ such that
  % $\supp(B') = \supp(B)\cup Y$. By definition of $r^Y$, there exist
  % $c'_1, c'_2 \in B'$ such that $r(c'_1) = r^Y(c_1) \cup Y$ and
  % $r(c'_2) = r^Y(c_2) \cup Y$. Consequently,
  % $r(c'_1) \cap \supp(B') = (r^Y(c_1) \cup Y) \cap (\supp(B) \cup Y) =
  % (r^Y(c_1)\cap \supp(B)) \cup Y$ and similarly,
  % $r(c'_2) \cap \supp(B') = (r^Y(c_2)\cap \supp(B)) \cup
  % Y$. Therefore, $r(c'_1) \cap \supp(B') = r(c'_2) \cap \supp(B')$ and
  % $r$ is not injective on the cube $B'$ of $C$, contradicting (C2) for
  % $r$.

  By Theorem~\ref{t:local-to-global}, to prove Assertion~(3), it is
  enough to show that $r_Y$ is the out-map of a USO of $C_Y$. Recall
  that for any $c \in C_Y$, we have $r_Y(c) = r(c_B)$, where $B$ is
  the unique $Y$-cube of $2^{\U}$ such that $\tg(B) = c$ and $c_B$ is
  the unique sink of $C \cap B$ for $o_r$.  By (C1), there exists an
  $r(c_B)$-cube $B'$ containing $c_B$ in $C$. Since
  $r(c_B)\cap Y = \varnothing$, there exists an $r(c_B)$-cube of $C_Y$
  containing $c$. Consequently, $C_Y$ and $r_Y$ satisfy (C1).

  Suppose there exists a cube $B$ in $C_Y$ violating (C2), i.e., there
  exist $c_1, c_2 \in C_Y \cap B$ such that
  $r_Y(c_1) \cap \supp(B)=r_Y(c_2) \cap \supp(B)$.  Any cube $B$ in
  $C_Y$ extends to a unique cube $B'$ of $2^{\U}$ such that
  $\supp(B') = \supp(B)\cup Y$. By definition of $r_Y$, there exist
  $c'_1, c'_2 \in B'$ such that $r(c'_1) = r_Y(c_1)$ and
  $r(c'_2) = r_Y(c_2)$. Consequently,
  $r(c'_1) \cap \supp(B') = r_Y(c_1) \cap \supp(B)$ since
  $r(c'_1)\cap Y = \varnothing$ and similarly
  $r(c'_2) \cap \supp(B')=r_Y(c_2) \cap \supp(B)$.  Consequently, the
  map $c \mapsto r(c)\cap \supp(B')$ is not injective on
  $C \cap B'$, contradicting Condition (R3) of
  Theorem~\ref{t:clashalt}.
\end{proof}

\subsection{Pre-representation Maps}\label{s:pre-rep}

We now show that we can find maps satisfying each of the conditions
(C1) and (C2). Nevertheless, we were not able to find a map satisfying
(C1) and (C2).  It is surprising that, while each $d$-cube has at
least $d^{\Omega(2^d)}$ USOs~\cite{Ma}, it is so difficult to find a
single USO for ample classes.

\begin{proposition}\label{prop:pre-rep}
  For any ample class $C$ there exists a bijection
  $r': C\rightarrow \uoX(C)$ and an injection
  $r'': C\rightarrow 2^{\U}$ such that $r'$ satisfies the condition
  (C1) and $r''$ satisfies the condition (C2).
\end{proposition}

\begin{proof}
  First we prove the existence of the bijection $r'$. For
  $s\in \uoX(C)$, denote by $N_s(C)$ the union of all $s$-cubes
  included in $C$ and call $N_s(C)$ the \emph{carrier} of $s$. For
  $S\subseteq \uoX(C)$, denote by $N_S(C)$ the union of all carriers
  $N_s(C), s\in S$.  Define a bipartite graph
  $\Gamma(C)=(C {\mathaccent\cdot\cup} \uoX(C),E)$, where there is an
  edge between a concept $c\in C$ and a strongly shattered set
  $s\in \uoX(C)$ if and only if $c$ belongs to the carrier
  $N_s(C)$. We assert that $\Gamma(C)$ admits a perfect matching $M$.
  By the definition of the edges of $\Gamma(C)$, if the edge $cs$ is
  in $M$, then $c$ belongs to $N_s(C)$ and thus the unique $s$-cube
  containing $c$ is included in $C$. Thus $r':C\rightarrow \uoX(C)$
  defined by setting $r'(c)=s$ if and only if $cs\in M$ is a bijection
  satisfying (C1).

  Since $C$ is ample, we have $\lvert \uoX(C)\rvert =\lvert C\rvert$,
  and thus to prove the existence of $M$, we show that the graph
  $\Gamma$ satisfies the conditions of Philip Hall's
  theorem~\cite{LoPl}: if $S$ is an arbitrary subset of simplices of
  $\uoX(C)$, then $\lvert N_S(C)\rvert \ge \lvert S\rvert$.  Indeed,
  since $C$ is ample, for any $s\in S$ the carrier $N_s(C)$ contains
  at least one $s$-cube, thus $S\subseteq
  \underline{X}(N_S(C))$. Consequently,
  $\lvert S\rvert\le \lvert\underline{X}(N_S(C))\rvert \le \lvert
  N_S(C)\rvert$ by the Sandwich Lemma applied to the class $N_S(C)$.

  \medskip We now prove that there exists an injection
  $r'': C \to 2^{\U}$ satisfying (C2). We prove the existence of $r''$
  by induction on the size of $\U$. If $\lvert \U\rvert=0$, $r''$
  trivially exists. Consider now $x \in X$ and note that there exists
  an injection $r''_x: C_x \to 2^{\U\setminus\{x\}}$ satisfying (C2)
  by induction hypothesis. We define $r''$ by:
  \begin{equation*}
    r''(c) = \begin{cases}
      r''_x(c) & \text{if }  x \notin c,\\
      r''_x(c\setminus\{x\}) & \text{if }  x \in c \text{ and }
      c\setminus\{x\} \notin C,\\
      r''_x(c\setminus\{x\})\cup\{x\} & \text{otherwise.}
    \end{cases}
  \end{equation*}
  It means that the orientation of the edges of $G(C)$ is obtained by
  keeping the orientation of the edges of $G(C_x)$ and orienting all
  $x$-edges of $G(C)$ from $c \cup \{x\}$ to $c$. It is easy to verify that
  $r''$ is injective.

  Consider two distinct concepts $c, c' \in C$ belonging to a common
  cube of $C$ and let $B = B(c,c')$ that is the minimal cube of $C$
  containing $c$ and $c'$. If $c \Delta c' \neq \{x\}$, then since
  $r''(c) \setminus \{x\} = r''_x(c\setminus \{x\})$, since
  $r''(c') \setminus \{x\} = r''_x(c'\setminus \{x\})$, and since
  $r''_x$ satisfies (C2) on the cube $B_x$ of $C_x$, there exists
  $y \in (c\setminus \{x\} \cup c'\setminus \{x\})\Delta
  (r''_x(c\setminus \{x\}) \cup r''_x(c'\setminus \{x\})) \subseteq (c
  \cup c')\Delta (r''(c) \cup r''(c'))$. Suppose now that $c' =
  c \cup \{x\}$. In this case, $x \in r''(c')\setminus r''_x(c)$ and
  $x \in (c\Delta c')\cap (r''(c)\Delta r''(c'))$.
\end{proof}

One can try to find representation maps for ample classes by extending
the approach for maximum classes: given ample classes $C$ and $D$ with
$D\subset C$, a representation map $r$ for $C$ is called
\emph{$D$-entering} if all edges $cd$ with $c\in C\setminus D$ and
$d\in D$ are directed by $o_r$ from $c$ to $d$.  The representation
map defined in Proposition~\ref{lem-unique-source} is $D$-entering.
Given $x\in \dom(C)$, suppose that $r_x$ is a $C^x$-entering
representation map for $C_x$.

We can extend the orientation $o_{r_x}$ to an orientation $o$ of
$G(C)$ as follows.  Each $x$-edge $cc'$ of $G(C)$ is directed
arbitrarily, while each other edge $cc'$ is directed as the edge
$c_xc'_x$ is directed by $o_{r_x}$.  Since~$o_{r_x}$ satisfies (C1),
(C2) and $r_x$ is $C^x$-entering, $o$ also satisfies (C1), (C2), thus
the map~$r_o$ is a representation map for $C$. So, ample classes would
admit representation maps, if \emph{for any ample classes
  $D\subseteq C$, any representation map $r'$ of $D$ extends to a
  $D$-entering representation map $r$ of $C$.}

\subsection{Representation Maps as ISRs} \label{s:ISR}

Our next result formulates the construction of representation maps for
ample classes as an instance of the Independent System of
Representatives problem.  A system $(G,(V_i)_{1\leq i\leq n})$
consisting of a graph $G$ and a partition $V_1 ,\ldots, V_n$ of its
vertex set $V(G)$ is called an \emph{ISR-system}.  An independent set
in $G$ of the form $\{v_1, \ldots, v_n\}$, where $v_i \in V_i$ for
each $1 \leq i \leq n$, is called an an \emph{Independent System of
  Representatives}, or \emph{ISR} for short~\cite{AhBeZi}.

Consider an ample class $C\subseteq 2^{\U}$. We build an ISR-system
$(G,(V_c)_{c \in C})$ as follows. For each concept $c$ and each set
$Y \subseteq \uoX(C)$ such that $c$ belongs to a $Y$-cube of $C$,
there is a vertex $(c,Y)$ in $V(G)$. For each concept $c\in C$, set
$V_c:=\{(c,Y) \in V(G)\}$. Finally, $E(G)$ is defined as follows:
there is an edge between two vertices $(c_1,Y_1)$, $(c_2, Y_2)$ in $G$
if $c_1, c_2$ are distinct vertices belonging to a common cube $B$
such that $Y_1 \cap \supp(B) = Y_2 \cap \supp(B)$.

\begin{proposition}\label{ISR}
  There is an ISR for $(G,(V_c)_{c \in C})$ if and only if $C$ admits
  a representation map.
\end{proposition}

\begin{proof}
  Assume first that $r$ is a representation map for $C$. We show that
  $\{(c,r(c))\}_{c \in C}$ is an ISR for $(G,(V_c)_{c \in C})$. By
  (C1), for every $c \in C$, $c$ belongs to an $r(c)$-cube of $C$ and
  thus $(c,r(c)) \in V(G)$. Moreover, if $\{(c,r(c))\}_{c \in C}$ is
  not an independent set of $G$, there exist two concepts $c_1, c_2$
  in a cube $B$ such that
  $r(c_1) \cap \supp(B) = r(c_2) \cap \supp(B)$, contradicting (C2).

  Conversely, if $\{(c,Y_c)\}_{c \in C}$ is an ISR for
  $(G,(V_c)_{c\in C})$, then the map $r: c \to \uoX(C)$ defined by
  $r(c) = Y_c$ is a representation map.  Indeed, (C1) is satisfied by
  the definition of the vertices of $V(G)$, (C2) is satisfied by the
  definition of the edges of $E(G)$ and since $\{(c,Y_c)\}_{c \in C}$
  is an independent set of $G$.
\end{proof}

\subsection*{Acknowledgements}
We are very grateful to the referees of this submission and to some of
the referees of previous versions of this paper for useful remarks
that helped improving the presentation of this work. We are also
grateful to Olivier Bousquet for several insightful discussions.  The
research of J.\ C.\ and V.\ C.\ was supported by the ANR project DISTANCIA
(ANR-17-CE40-0015).
S.\ M.\ is a Robert J.\ Shillman Fellow and was supported by the ISF, grant no.~1225/20, by an Azrieli Faculty Fellowship, and by BSF grant 2018385;
part of this project was carried while S.M. was at the Institute for Advanced Study and was supported by the NSF 
%under agreement No.\ 
grant CCF-1412958.
% The research of S.M.\ is supported by the Simons
% Foundation and the NSF; part of this project was carried while S.M.\
% was at the Institute for Advanced Study and was supported by the
% National Science Foundation under agreement No. CCF-1412958. 
The research of M.\ W.\ was partially supported by the NSF grant IIS-1619271.

\bibliographystyle{plainurl}
\bibliography{refs-lopsided}

\newpage
\appendix
\section{Two representation maps for Hall's concept class $C_H$}

We give two representation maps for Hall's concept class $C_H$
given in Figure~\ref{fig-hall}. The representation map $r_1$ presented
in Table~\ref{table-rm-M} was found by constructing Boolean clauses
for a matching between the concepts and the dimension sets of size
$\le 3$ that satisfies the non-clashing condition.  A satisfying
assignment was then found with the open source \textsf{MiniSat} solver.
% (\url{http://minisat.se/Main.html}).
The representation map $r_2$ presented in Table~\ref{table-rm-J} was
found by following the steps of the proof of Theorem~\ref{th-maximum}
by recursing on dimensions 12, 11, $\ldots$ In both cases, the
representation map $r_i$, $i=1,2$ of $C_H$ is described as follows: we
list all the 299 concepts as bit vectors and for each concept
$c \in C_H$, we underline the bits of $r_i(c)$ (a subset of size
$\le 3$ from $\{1,2,\ldots,12\}$).

The USOs corresponding to these two representation maps are given in
Figures~\ref{fig-rm-ce-M} and ~\ref{fig-rm-ce-J}; they were obtained
by using \textsf{Sage} and \textsf{Graphviz}.  In both cases, the
non-trivial strongly connected components are represented by the blue
boxes. Except from the arcs within these strongly connected
components, the other arcs are downward arcs.

% \subsection{On the number of non-trivial strongly connected
%   components}

\medskip

The USO from Figure~\ref{fig-rm-ce-M} has a unique non-trivial
strongly connected component (NSCC) $K$ while the one from
Figure~\ref{fig-rm-ce-J} has two such components $K'$ and $K''$. Since
those components are small (each of them contains $12$ concepts),
these USOs are not ``far'' from being acyclic and $C_H$ is (not
``far'' from being) a minimal corner-free example.

Given a concept class $C$ and an orientation $o$ of the edges of
$G(C)$, we denote by $\dGo(C)$ the resulting directed graph.
From the following proposition, we deduce that any representation map
for $C_H$ obtained by the algorithm of Theorem~\ref{th-maximum}
contains at least two NSCCs. In particular, this shows that the 
representation map $r_1$ cannot be obtained in such a way.

\begin{proposition}
  For any maximum class $C$ without corners,  any representation map
  $r$ of $C$ computed by Theorem~\ref{th-maximum}, and the USO $o$
  corresponding to $r$, the directed graph $\dG_{o}(C)$ contains at
  least two non-trivial strongly connected components.
\end{proposition}
 
\begin{proof}
  We prove the result by induction on the dimension of $C$. If $C$ is
  of dimension at most $2$, then by Proposition~\ref{dismantlable_2dim} 
  $C$ contains a corner and we are done.

  Now suppose that the dimension of $C$ is at least 3. Recall that the first step 
  is to contract $C$ along one of the coordinates
  $x \in \U$ (to construct $r_2$, we contracted $C_H$ along coordinate
  $12$, see Figure~\ref{fig-rm-ce-CD}) and to compute recursively a
  USO $o^x$ (and the corresponding representation map $r^x$) of
  $C^x$. Then we extend $o^x$ to a USO $o_x$ (and find the
  corresponding representation map $r_x$) of $C_x$. Finally, we obtain
  $o$ from $o_x$ by keeping the orientation of all edges that do not
  correspond to coordinate $x$ and by orienting all $x$-edges from
  $1C^x$ to $0C^x$.

  Observe that any directed cycle of $\dGo(C)$ corresponds either to
  a cycle of $C^x$ or to a directed cycle of a connected component of
  $G(C_x \setminus C^x)$.  If $C^x$ does not contain a corner, then by
  induction hypothesis, $C^x$ contains at least two NSCCs and thus
  there are two NSCCs in $0 C^x$ and thus in $C$ and we are done.

  Suppose now that $C^x$ contains a corner.  Consider a connected
  component $A$ of $G(C_x \setminus C^x)$. If $A$ does not contain any
  directed cycle, then the concept class $A$ is acyclic and it
  contains a source $s$. By (C1) applied to $s$ and $r$, we conclude 
  that $s$ is a corner of $C$. Consequently, each connected component
  of $G(C_x \setminus C^x)$ contains at least one NSCC.

  If $C^x$ is not a separator of $G(C_x)$, i.e., if $N_x(C)$ is not a
  separator of $G(C)$, then without loss of generality, we can assume
  that there is no edge from $0C^x$ to $C\setminus N_x(C)$. Consider
  any corner $c$ of $C^x$ that is contained in a unique maximal cube
  $B^x$ of $C^x$. In $C$ there exists a unique cube $B$ containing $c$
  such that $\supp(B) = \supp(B^x) \cup \{x\}$. Since $c$ has no
  neighbor outside the carrier $N_x(C)$, $B$ is the unique maximal cube containing
  $c$ and thus $c$ is a corner of $C$, a contradiction.

  We can thus assume that $C^x$ is a separator of $G(C_x)$ and thus
  $G(C_x \setminus C^x)$ contains at least two connected
  components, and by the previous assertion, each of them contains at
  least one NSCC and we are done.
\end{proof}

%\subsection{}

\medskip

A strongly connected component (SCC) $S$ of a directed graph $\dG$ is
called a \emph{source-component} if there is no arc from $u$ to $v$
with $v \in S$ and $u \notin S$.

In the directed graph $\dGo$ corresponding to a representation map $r$
of an ample class $C$, if a source-component $S$ is reduced to a
single concept $c$ (i.e., $S$ is trivial), then $c$ is a corner of
$C$. Thus, one can view the source-components of $\dGo$ as a
generalization of corners. However, the definition of
source-components depends on a given representation map $r$ and the
corresponding orientation $o_r$ of $G(C)$, while the corners are
defined in the undirected graph $G(C)$. In the following proposition,
we prove that, similarly to corners, removing a source component $S$
from an ample class $C$ results into an ample class $C\setminus
S$. Moreover, the restriction of the representation map $r$ to
$C\setminus S$ is still a representation map.

\begin{proposition} \label{source-component}
  Let $C$ be an ample class $C$, $r$ be a representation map of $C$, $o=o_r$ the 
  USO defined by $r$, and $\dGo$ be the corresponding directed graph. Then  
  for any source-component $S$ of $\dGo$, $C \setminus S$ is an ample
  class and the restriction of $r$ to $C \setminus S$ is a
  representation map of $C\setminus S$.
\end{proposition}

\begin{proof} Set $C':=C\setminus S$ and denote by $r'$ the restriction of $r$ to $C'$. 
Since $r$ satisfies (C2) on $C$,  $r'$ satisfies (C2) 
on $C'$. Now, we show that $r'$ satisfies (C1) on $C'$. Pick any concept $c\in C'$ and let 
$B$ be the cube of $2^{\U}$ containing $c$ and defined by $r(c)$. By (C1), $B$ is included 
in $C$. If $B$ is also included in $C'$, then we are done. So, suppose that $B\cap S\ne \varnothing$ 
and pick a concept $c'\in B\cap S$ closest to $c$. Let $B'=B(c,c')$ be the cube spanned by $c$ and $c'$. 
Since $B'\subset B$, $B'$ is a cube of $C$. By the definition of $B$, $c$ is a source of $B$ and thus of $B'$. 
On the other hand, from the choice of $c'$ all neighbors of $c'$ in $B'$ belong to $C\setminus S$. Since 
$c'\in S$ and $S$ is a source-component of $\dGo$, all these edges have $c'$ as a source. Consequently, $c'$ 
is a second source of $B'$, contrary to condition (C2) applied to $C$. This shows that $B$ is included in $C'$, 
establishing that $r'$ satisfies (C1) on $C'$. 

Since $r$ is a bijective map from $C$ to $X(C)$, from (C1) applied to $r'$ we conclude that $r'$ is an injective map from 
$C'$ to $\uX(C')$, yielding $|C'|\le |\uX(C')|$. Since $|\uX(C')|\le |C'|$ by the sandwich lemma, we deduce that 
$|C'|=|\uX(C')|$ and thus $C'$ is ample by Theorem \ref{thm:ample1}(iv). Since the out-map of $r'$ satisfies 
(C1) and (C2), $r'$ is a representation map of $C'$. 
\end{proof}

  % Ideas (USO):
  %  (C2) is trivial

  %  (C1) for a vertex v, consider the closest vertex in the cube
  %  spanned by r(v) that is in K and obtain a contradiction (two
  %  sources)

  %  Ideas(ampleness):
  %  |C\setminus K| >= X(S\setminus K) : sandwich lemma
  %  |C\setminus K| <= X(S\setminus K) : use r and (C1)

\begin{corollary}
  $C_H\setminus K$ and $C_H\setminus (K' \cup K'')$ are ample classes admitting
  corner peelings.
\end{corollary}

\begin{proof} By Proposition \ref{source-component}, the concept
  classes $C_H\setminus K$ and $C_H\setminus K'$ are ample classes and
  the restrictions of $r$ to them are representation maps. Applying
  Proposition \ref{source-component} we conclude that
  $C_H\setminus (K' \cup K'')$ is an ample class and the restriction
  of $r$ to $C_H\setminus (K' \cup K'')$ is a representation map.
  Moreover, since all other NSCCs of $\dG(C)$ are trivial, the
  orientations of the edges of graphs $G(C\setminus K)$ and
  $C_H\setminus (K' \cup K'')$ defined by these representation maps
  are acyclic USOs. By Proposition \ref{prop-acyclicUSO},
  $C_H\setminus K$ and $C_H\setminus (K' \cup K'')$ admit corner
  peelings.
\end{proof}

\newcommand{\0}{0}
\newcommand{\1}{1}
\newcommand{\2}{\underline{0}}
\newcommand{\3}{\underline{1}}
\setlength{\tabcolsep}{1.1pt}

\newpage

\begin{table}[ht]
    \caption{The representation $r_1$ map for $C_H$ obtained using \textsf{MiniSat}}
  \centering
{\small
\begin{tabular}{ccc ccc c}
$\2\2\0\2\0\0\0\0\0\0\0\0$ &$\0\2\0\0\0\0\0\0\0\0\2\3$ &$\2\2\0\0\0\0\0\0\0\0\3\0$ &$\2\2\0\0\0\0\0\0\0\0\1\3$ &$\0\2\0\2\0\0\0\3\0\0\0\0$ &$\0\0\0\0\0\2\0\3\0\0\0\3$ &$\0\0\0\0\0\2\0\3\0\0\3\1$\\
$\0\2\0\2\0\0\3\1\0\0\0\0$ &$\0\0\0\0\0\2\3\1\0\0\0\3$ &$\0\0\0\0\0\2\3\1\0\0\3\1$ &$\0\0\0\2\0\2\1\1\0\3\1\1$ &$\0\0\2\0\0\0\1\1\3\2\1\1$ &$\0\0\0\2\0\2\1\1\3\1\1\1$ &$\0\2\0\0\0\3\0\0\0\0\0\2$\\
$\0\2\0\0\0\3\0\0\0\0\2\1$ &$\2\2\0\0\0\3\0\0\0\0\1\1$ &$\0\2\0\0\0\3\0\3\0\0\0\0$ &$\0\0\0\0\2\1\0\3\0\0\0\3$ &$\0\0\0\0\2\1\0\3\0\0\3\1$ &$\0\2\0\0\0\3\3\1\0\0\0\0$ &$\0\0\0\0\2\1\3\1\0\0\0\3$\\
$\0\0\0\0\2\1\3\1\0\0\3\1$ &$\0\0\0\2\2\1\1\1\0\3\1\1$ &$\0\0\0\0\0\3\1\1\3\2\1\1$ &$\0\0\0\2\2\1\1\1\3\1\1\1$ &$\2\0\0\0\3\2\0\0\0\0\1\1$ &$\0\0\0\0\3\2\0\3\0\0\1\1$ &$\0\0\0\0\3\2\3\1\0\0\1\1$\\
$\0\0\0\0\3\2\1\1\3\0\1\1$ &$\0\0\0\0\3\2\1\1\1\3\1\1$ &$\0\2\0\0\3\1\0\0\0\0\0\2$ &$\0\2\0\0\3\1\0\0\0\0\2\1$ &$\2\2\0\0\3\1\0\0\0\0\1\1$ &$\0\2\0\0\3\1\0\3\0\0\0\0$ &$\0\2\0\0\1\1\0\3\0\0\0\3$\\
$\0\2\0\0\1\1\0\3\0\0\3\1$ &$\0\2\0\0\3\1\3\1\0\0\0\0$ &$\0\2\0\0\1\1\3\1\0\0\0\3$ &$\0\2\0\0\1\1\3\1\0\0\3\1$ &$\0\2\0\2\1\1\1\1\0\3\0\0$ &$\0\2\0\0\1\1\1\1\0\3\0\3$ &$\0\2\0\0\1\1\1\1\0\3\3\1$\\
$\0\0\0\0\3\1\1\1\3\2\1\1$ &$\0\2\0\2\1\1\1\1\3\1\0\0$ &$\0\2\0\0\1\1\1\1\3\1\0\3$ &$\0\2\0\0\1\1\1\1\3\1\3\1$ &$\2\2\2\1\0\0\0\0\0\0\0\0$ &$\0\2\2\1\0\0\0\3\0\0\0\0$ &$\0\2\2\1\0\0\3\1\0\0\0\0$\\
$\0\0\0\3\0\2\1\1\0\0\0\3$ &$\0\0\0\3\0\2\1\1\0\0\3\1$ &$\0\0\2\1\0\2\1\1\0\3\1\1$ &$\0\0\2\1\0\2\1\1\3\1\1\1$ &$\0\2\0\3\0\3\1\1\0\0\0\0$ &$\0\0\0\3\2\1\1\1\0\0\0\3$ &$\0\0\0\3\2\1\1\1\0\0\3\1$\\
$\0\0\2\1\2\1\1\1\0\3\1\1$ &$\0\0\2\1\2\1\1\1\3\1\1\1$ &$\0\0\0\3\3\2\1\1\1\1\1\1$ &$\0\2\0\3\3\1\1\1\0\0\0\0$ &$\0\0\0\3\1\1\1\1\0\2\0\3$ &$\0\0\0\3\1\1\1\1\0\2\3\1$ &$\0\2\2\1\1\1\1\1\0\3\0\0$\\
$\0\0\0\3\1\1\1\1\2\1\0\3$ &$\0\0\0\3\1\1\1\1\2\1\3\1$ &$\0\2\2\1\1\1\1\1\3\1\0\0$ &$\0\2\0\3\1\1\1\1\1\1\0\3$ &$\0\2\0\3\1\1\1\1\1\1\3\1$ &$\0\0\3\2\0\0\1\1\0\2\1\1$ &$\0\0\3\2\0\0\1\1\2\1\1\1$\\
$\0\0\1\2\0\0\1\1\3\2\1\1$ &$\0\0\3\2\2\0\1\1\1\1\1\1$ &$\0\0\3\2\1\2\1\1\1\1\1\1$ &$\0\2\3\2\1\1\1\1\1\1\1\1$ &$\2\2\1\1\0\0\0\0\2\0\0\0$ &$\0\2\1\1\0\0\0\0\2\3\0\0$ &$\2\2\1\1\0\0\0\0\1\2\0\0$\\
$\2\2\1\1\0\0\0\2\1\1\0\0$ &$\0\2\1\1\0\0\0\3\0\2\0\0$ &$\0\2\1\1\0\0\0\3\2\1\0\0$ &$\2\2\1\1\0\0\2\1\1\1\0\0$ &$\0\0\1\1\0\0\3\2\0\2\0\0$ &$\0\0\1\1\0\0\3\2\2\1\0\0$ &$\0\0\1\1\0\0\3\0\3\2\0\0$\\
$\2\0\1\1\0\0\3\2\1\1\0\0$ &$\0\2\1\1\0\0\3\1\0\2\0\0$ &$\0\0\3\1\0\2\1\1\0\0\0\3$ &$\0\0\3\1\0\2\1\1\0\0\3\1$ &$\0\2\1\1\0\0\3\1\2\1\0\0$ &$\0\0\1\1\0\2\1\1\0\3\0\3$ &$\0\0\1\1\0\2\1\1\0\3\3\1$\\
$\0\0\1\1\0\0\1\3\3\2\0\0$ &$\0\0\1\1\0\0\1\1\3\2\0\3$ &$\0\0\1\1\0\0\1\1\3\2\3\1$ &$\2\2\1\1\0\0\1\1\1\1\0\0$ &$\0\0\1\1\0\2\1\1\3\1\0\3$ &$\0\0\1\1\0\2\1\1\3\1\3\1$ &$\0\2\3\1\0\3\1\1\0\0\0\0$\\
$\0\0\3\1\2\1\1\1\0\0\0\3$ &$\0\0\3\1\2\1\1\1\0\0\3\1$ &$\0\2\1\1\0\3\1\1\0\3\0\0$ &$\0\0\1\1\2\1\1\1\0\3\0\3$ &$\0\0\1\1\2\1\1\1\0\3\3\1$ &$\0\2\1\1\0\3\1\1\3\1\0\0$ &$\0\0\1\1\2\1\1\1\3\1\0\3$\\
$\0\0\1\1\2\1\1\1\3\1\3\1$ &$\0\0\3\1\3\2\1\1\1\1\1\1$ &$\0\2\3\1\3\1\1\1\0\0\0\0$ &$\0\0\3\1\1\1\1\1\0\2\0\3$ &$\0\0\3\1\1\1\1\1\0\2\3\1$ &$\0\2\1\1\3\1\1\1\0\3\0\0$ &$\0\0\3\1\1\1\1\1\2\1\0\3$\\
$\0\0\3\1\1\1\1\1\2\1\3\1$ &$\0\2\1\1\3\1\1\1\3\1\0\0$ &$\0\2\3\1\1\1\1\1\1\1\0\3$ &$\0\2\3\1\1\1\1\1\1\1\3\1$ &$\2\1\0\2\0\0\0\2\0\0\0\0$ &$\2\1\0\0\0\2\0\0\0\0\0\3$ &$\2\1\0\0\0\0\0\0\0\0\3\2$\\
$\2\1\0\0\0\2\0\0\0\0\3\1$ &$\2\1\0\2\0\0\2\1\0\0\0\0$ &$\2\1\0\2\0\2\1\1\0\0\0\0$ &$\2\1\0\0\0\3\0\2\0\0\0\0$ &$\2\1\0\0\2\1\0\0\0\0\0\3$ &$\2\1\0\0\2\1\0\0\0\0\3\1$ &$\2\1\0\0\0\3\2\1\0\0\0\0$\\
$\2\1\0\2\2\1\1\1\0\0\0\0$ &$\2\1\0\0\3\1\0\2\0\0\0\0$ &$\2\1\0\0\1\1\0\2\0\0\0\3$ &$\2\1\0\0\1\1\0\2\0\0\3\1$ &$\2\1\0\0\3\1\2\1\0\0\0\0$ &$\2\1\0\0\1\1\2\1\0\0\0\3$ &$\2\1\0\0\1\1\2\1\0\0\3\1$\\
$\2\1\0\2\1\1\1\1\0\2\0\0$ &$\2\1\0\0\1\1\1\1\0\2\0\3$ &$\2\1\0\0\1\1\1\1\0\2\3\1$ &$\2\1\0\2\1\1\1\1\2\1\0\0$ &$\2\1\0\0\1\1\1\1\2\1\0\3$ &$\2\1\0\0\1\1\1\1\2\1\3\1$ &$\2\1\0\2\1\1\1\1\1\1\0\2$\\
$\2\1\0\2\1\1\1\1\1\1\2\1$ &$\2\1\2\2\1\1\1\1\1\1\1\1$ &$\2\1\2\1\0\0\0\2\0\0\0\0$ &$\2\1\2\1\0\0\2\1\0\0\0\0$ &$\2\1\2\1\0\2\1\1\0\0\0\0$ &$\2\1\2\1\2\1\1\1\0\0\0\0$ &$\2\1\2\1\1\1\1\1\0\2\0\0$\\
$\2\1\2\1\1\1\1\1\2\1\0\0$ &$\2\1\2\1\1\1\1\1\1\1\0\2$ &$\2\1\2\1\1\1\1\1\1\1\2\1$ &$\2\1\2\1\1\1\1\1\1\1\1\1$ &$\2\1\1\2\1\1\1\1\1\1\1\1$ &$\2\1\1\1\0\0\0\2\0\2\0\0$ &$\2\1\1\1\0\0\0\2\2\1\0\0$\\
$\2\1\1\1\0\0\0\0\3\2\0\0$ &$\2\1\1\1\0\0\0\2\1\1\0\0$ &$\2\1\1\1\0\0\2\1\0\2\0\0$ &$\2\1\1\1\0\0\2\1\2\1\0\0$ &$\2\1\1\1\0\0\2\1\1\1\0\0$ &$\2\1\1\1\0\2\1\1\0\2\0\0$ &$\2\1\1\1\0\2\1\1\2\1\0\0$\\
$\2\1\1\1\0\2\1\1\1\1\0\0$ &$\2\1\1\1\2\1\1\1\0\2\0\0$ &$\2\1\1\1\2\1\1\1\2\1\0\0$ &$\2\1\1\1\2\1\1\1\1\1\0\0$ &$\2\1\1\1\1\1\1\1\0\2\0\0$ &$\2\1\1\1\1\1\1\1\2\1\0\0$ &$\2\1\1\1\1\1\1\1\1\1\0\2$\\
$\2\1\1\1\1\1\1\1\1\1\2\1$ &$\2\1\1\1\1\1\1\1\1\1\1\1$ &$\1\2\0\2\0\0\0\0\0\0\0\0$ &$\1\2\0\0\0\0\0\0\0\0\3\0$ &$\1\2\0\0\0\0\0\0\0\0\1\3$ &$\1\2\0\0\0\3\0\0\0\0\1\1$ &$\1\2\0\0\3\2\0\0\0\0\1\1$\\
$\1\2\0\0\3\1\0\0\0\0\1\1$ &$\1\2\2\1\0\0\0\0\0\0\0\0$ &$\1\2\1\1\0\0\0\0\2\0\0\0$ &$\1\2\1\1\0\0\0\0\1\2\0\0$ &$\1\2\1\1\0\0\2\2\1\1\0\0$ &$\1\2\1\1\0\0\2\1\1\1\0\0$ &$\1\2\1\1\0\0\1\2\1\1\0\0$\\
$\1\2\1\1\0\0\1\1\1\1\0\0$ &$\1\1\0\2\0\2\0\2\0\0\0\0$ &$\1\1\0\0\0\2\0\0\0\0\0\3$ &$\1\1\0\0\2\0\0\0\0\0\3\2$ &$\1\1\0\0\0\2\0\0\0\0\3\1$ &$\1\1\0\2\0\2\2\1\0\0\0\0$ &$\1\1\0\2\0\2\1\1\0\0\0\0$\\
$\1\1\0\2\2\1\0\2\0\0\0\0$ &$\1\1\0\0\2\1\0\0\0\0\0\3$ &$\1\1\0\0\2\1\0\0\0\0\3\1$ &$\1\1\0\2\2\1\2\1\0\0\0\0$ &$\1\1\0\2\2\1\1\1\0\0\0\0$ &$\1\1\0\0\3\2\0\0\0\0\0\2$ &$\1\1\0\0\3\2\0\0\0\0\2\1$\\
$\1\1\0\0\1\2\0\0\0\0\3\2$ &$\1\1\0\0\3\2\0\0\0\0\1\1$ &$\1\1\0\2\1\1\0\2\0\2\0\0$ &$\1\1\0\0\1\1\0\2\0\2\0\3$ &$\1\1\0\0\1\1\0\0\0\2\3\2$ &$\1\1\0\0\1\1\0\2\0\2\3\1$ &$\1\1\0\2\1\1\0\2\2\1\0\0$\\
$\1\1\0\0\1\1\0\2\2\1\0\3$ &$\1\1\0\0\1\1\0\0\2\1\3\2$ &$\1\1\0\0\1\1\0\2\2\1\3\1$ &$\1\1\2\2\1\1\0\0\1\1\0\2$ &$\1\1\2\2\1\1\0\0\1\1\2\1$ &$\1\1\2\0\1\1\0\0\1\1\3\2$ &$\1\1\2\2\1\1\2\0\1\1\1\1$\\
$\1\1\0\2\1\1\2\1\0\2\0\0$ &$\1\1\0\0\1\1\2\1\0\2\0\3$ &$\1\1\0\0\1\1\2\1\0\2\3\1$ &$\1\1\0\2\1\1\2\1\2\1\0\0$ &$\1\1\0\0\1\1\2\1\2\1\0\3$ &$\1\1\0\0\1\1\2\1\2\1\3\1$ &$\1\1\0\2\1\1\0\3\1\1\0\2$\\
$\1\1\0\2\1\1\0\3\1\1\2\1$ &$\1\1\0\2\1\1\2\3\1\1\1\1$ &$\1\1\2\2\1\1\1\2\1\1\1\1$ &$\1\1\0\2\1\1\1\1\0\2\0\0$ &$\1\1\0\0\1\1\1\1\0\2\0\3$ &$\1\1\0\0\1\1\1\1\0\2\3\1$ &$\1\1\0\2\1\1\1\1\2\1\0\0$\\
$\1\1\0\0\1\1\1\1\2\1\0\3$ &$\1\1\0\0\1\1\1\1\2\1\3\1$ &$\1\1\0\2\1\1\3\1\1\1\0\2$ &$\1\1\0\2\1\1\3\1\1\1\2\1$ &$\1\1\2\2\1\1\1\1\1\1\1\1$ &$\1\1\2\1\0\2\0\2\0\0\0\0$ &$\1\1\2\1\0\2\2\1\0\0\0\0$\\
$\1\1\2\1\0\2\1\1\0\0\0\0$ &$\1\1\2\1\2\1\0\2\0\0\0\0$ &$\1\1\2\1\2\1\2\1\0\0\0\0$ &$\1\1\2\1\2\1\1\1\0\0\0\0$ &$\1\1\2\1\1\1\0\2\0\2\0\0$ &$\1\1\2\1\1\1\0\2\2\1\0\0$ &$\1\1\2\1\1\1\0\2\1\1\0\2$\\
$\1\1\2\1\1\1\0\2\1\1\2\1$ &$\1\1\2\1\1\1\2\2\1\1\1\1$ &$\1\1\2\1\1\1\2\1\0\2\0\0$ &$\1\1\2\1\1\1\2\1\2\1\0\0$ &$\1\1\2\1\1\1\2\1\1\1\0\2$ &$\1\1\2\1\1\1\2\1\1\1\2\1$ &$\1\1\2\1\1\1\2\1\1\1\1\1$\\
$\1\1\2\1\1\1\1\2\1\1\1\1$ &$\1\1\2\1\1\1\1\1\0\2\0\0$ &$\1\1\2\1\1\1\1\1\2\1\0\0$ &$\1\1\2\1\1\1\1\1\1\1\0\2$ &$\1\1\2\1\1\1\1\1\1\1\2\1$ &$\1\1\2\1\1\1\1\1\1\1\1\1$ &$\1\1\1\2\1\1\0\0\1\1\2\2$\\
$\1\1\1\2\1\1\0\0\1\1\2\1$ &$\1\1\1\2\1\1\0\0\1\1\1\2$ &$\1\1\1\2\1\1\2\0\1\1\1\1$ &$\1\1\1\2\1\1\1\2\1\1\1\1$ &$\1\1\1\2\1\1\1\1\1\1\1\1$ &$\1\1\1\1\0\2\0\2\0\2\0\0$ &$\1\1\1\1\0\2\0\2\2\1\0\0$\\
$\1\1\1\1\0\0\0\0\3\2\0\0$ &$\1\1\1\1\0\2\2\2\1\1\0\0$ &$\1\1\1\1\0\2\2\1\0\2\0\0$ &$\1\1\1\1\0\2\2\1\2\1\0\0$ &$\1\1\1\1\0\2\2\1\1\1\0\0$ &$\1\1\1\1\0\2\1\2\1\1\0\0$ &$\1\1\1\1\0\2\1\1\0\2\0\0$\\
$\1\1\1\1\0\2\1\1\2\1\0\0$ &$\1\1\1\1\0\2\1\1\1\1\0\0$ &$\1\1\1\1\2\1\0\2\0\2\0\0$ &$\1\1\1\1\2\1\0\2\2\1\0\0$ &$\1\1\1\1\2\1\2\2\1\1\0\0$ &$\1\1\1\1\2\1\2\1\0\2\0\0$ &$\1\1\1\1\2\1\2\1\2\1\0\0$\\
$\1\1\1\1\2\1\2\1\1\1\0\0$ &$\1\1\1\1\2\1\1\2\1\1\0\0$ &$\1\1\1\1\2\1\1\1\0\2\0\0$ &$\1\1\1\1\2\1\1\1\2\1\0\0$ &$\1\1\1\1\2\1\1\1\1\1\0\0$ &$\1\1\1\1\1\1\0\2\0\2\0\0$ &$\1\1\1\1\1\1\0\2\2\1\0\0$\\
$\1\1\1\1\1\1\2\2\1\1\2\0$ &$\1\1\1\1\1\1\0\2\1\1\2\3$ &$\1\1\1\1\1\1\2\2\1\1\1\2$ &$\1\1\1\1\1\1\2\2\1\1\1\1$ &$\1\1\1\1\1\1\2\1\0\2\0\0$ &$\1\1\1\1\1\1\2\1\2\1\0\0$ &$\1\1\1\1\1\1\2\1\1\1\2\2$\\
$\1\1\1\1\1\1\2\1\1\1\2\1$ &$\1\1\1\1\1\1\2\1\1\1\1\2$ &$\1\1\1\1\1\1\2\1\1\1\1\1$ &$\1\1\1\1\1\1\1\2\1\1\2\0$ &$\1\1\1\1\1\1\1\2\1\1\1\2$ &$\1\1\1\1\1\1\1\2\1\1\1\1$ &$\1\1\1\1\1\1\1\1\0\2\0\0$\\
$\1\1\1\1\1\1\1\1\2\1\0\0$ &$\1\1\1\1\1\1\1\1\1\1\2\2$ &$\1\1\1\1\1\1\1\1\1\1\2\1$ &$\1\1\1\1\1\1\1\1\1\1\1\2$ &$\1\1\1\1\1\1\1\1\1\1\1\1$ &&
\end{tabular}
}
\label{table-rm-M}
\end{table}

\newpage

\begin{figure}[h]
  \includegraphics[scale=0.13]{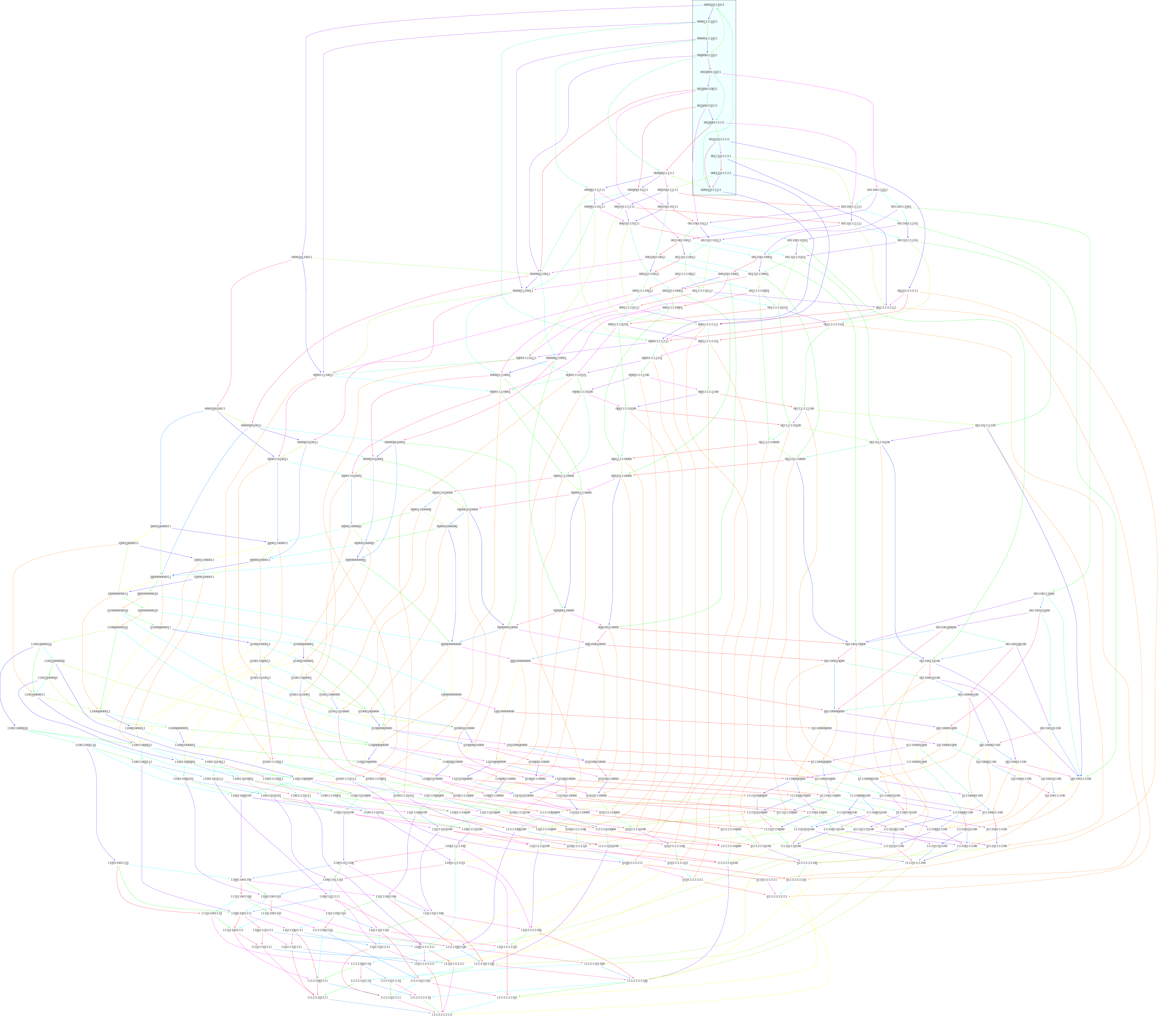}%
  \caption{The USO corresponding to the representation map $r_1$ of
    $C_H$ described in Table~\ref{table-rm-M}. The 12-vertex NSCC
      is in the blue box. Best viewed by zooming in.}%
  \label{fig-rm-ce-M}
\end{figure}

\newpage
\begin{table}[ht]
    \caption{The representation map $r_2$ for $C_H$
    obtained using Theorem~\ref{th-maximum} 
    by recursing on dimensions 12, 11, $\ldots$}
  \centering
{\small
\begin{tabular}{ccc ccc c}
$\0\2\0\0\0\0\0\0\0\0\0\0$ &
$\0\2\0\0\0\0\0\0\0\0\0\3$ &
$\0\2\0\0\0\0\0\0\0\0\3\0$ &
$\0\2\0\0\0\0\0\0\0\0\3\3$ &
$\0\0\0\0\0\2\0\3\0\0\0\0$ &
$\0\0\0\0\0\2\0\3\0\0\0\3$ &
$\0\0\0\0\0\2\0\3\0\0\3\1$ \\
$\0\0\0\0\0\2\3\1\0\0\0\0$ &
$\0\0\0\0\0\2\3\1\0\0\0\3$ &
$\0\0\0\0\0\2\3\1\0\0\3\1$ &
$\0\0\0\2\0\2\1\1\0\3\1\1$ &
$\0\0\2\0\0\0\1\1\3\2\1\1$ &
$\0\0\0\2\0\2\1\1\3\1\1\1$ &
$\0\2\0\0\0\3\0\0\0\0\0\0$ \\
$\0\2\0\0\0\3\0\0\0\0\0\3$ &
$\0\2\0\0\0\3\0\0\0\0\3\1$ &
$\0\0\0\0\2\1\0\3\0\0\0\0$ &
$\0\0\0\0\2\1\0\3\0\0\0\3$ &
$\0\0\0\0\2\1\0\3\0\0\3\1$ &
$\0\0\0\0\2\1\3\1\0\0\0\0$ &
$\0\0\0\0\2\1\3\1\0\0\0\3$ \\
$\0\0\0\0\2\1\3\1\0\0\3\1$ &
$\0\0\0\2\2\1\1\1\0\3\1\1$ &
$\0\0\0\0\0\3\1\1\3\2\1\1$ &
$\0\0\0\2\2\1\1\1\3\1\1\1$ &
$\2\0\0\0\3\2\0\0\0\0\1\1$ &
$\0\0\0\0\3\2\0\3\0\0\1\1$ &
$\0\0\0\0\3\2\3\1\0\0\1\1$ \\
$\0\0\0\0\3\2\1\1\3\0\1\1$ &
$\0\0\0\0\3\2\1\1\1\3\1\1$ &
$\0\2\0\0\3\1\0\0\0\0\0\0$ &
$\0\2\0\0\3\1\0\0\0\0\0\3$ &
$\0\2\0\0\3\1\0\0\0\0\3\1$ &
$\0\2\0\0\1\1\0\3\0\0\0\0$ &
$\0\2\0\0\1\1\0\3\0\0\0\3$ \\
$\0\2\0\0\1\1\0\3\0\0\3\1$ &
$\0\2\0\0\1\1\3\1\0\0\0\0$ &
$\0\2\0\0\1\1\3\1\0\0\0\3$ &
$\0\2\0\0\1\1\3\1\0\0\3\1$ &
$\0\2\0\0\1\1\1\1\0\3\0\0$ &
$\0\2\0\0\1\1\1\1\0\3\0\3$ &
$\0\2\0\0\1\1\1\1\0\3\3\1$ \\
$\0\0\0\0\3\1\1\1\3\2\1\1$ &
$\0\2\0\0\1\1\1\1\3\1\0\0$ &
$\0\2\0\0\1\1\1\1\3\1\0\3$ &
$\0\2\0\0\1\1\1\1\3\1\3\1$ &
$\0\2\0\3\0\0\0\2\0\0\0\0$ &
$\0\2\0\3\0\0\2\1\0\0\0\0$ &
$\0\0\0\3\0\2\1\1\0\0\0\0$ \\
$\0\0\0\3\0\2\1\1\0\0\0\3$ &
$\0\0\0\3\0\2\1\1\0\0\3\1$ &
$\0\0\2\1\0\2\1\1\0\3\1\1$ &
$\0\0\2\1\0\2\1\1\3\1\1\1$ &
$\0\0\0\3\2\1\1\1\0\0\0\0$ &
$\0\0\0\3\2\1\1\1\0\0\0\3$ &
$\0\0\0\3\2\1\1\1\0\0\3\1$ \\
$\0\0\2\1\2\1\1\1\0\3\1\1$ &
$\0\0\2\1\2\1\1\1\3\1\1\1$ &
$\0\0\0\3\3\2\1\1\1\1\1\1$ &
$\0\0\0\3\1\1\1\1\0\2\0\0$ &
$\0\0\0\3\1\1\1\1\0\2\0\3$ &
$\0\0\0\3\1\1\1\1\0\2\3\1$ &
$\0\0\0\3\1\1\1\1\2\1\0\0$ \\
$\0\0\0\3\1\1\1\1\2\1\0\3$ &
$\0\0\0\3\1\1\1\1\2\1\3\1$ &
$\0\2\0\3\1\1\1\1\1\1\0\0$ &
$\0\2\0\3\1\1\1\1\1\1\0\3$ &
$\0\2\0\3\1\1\1\1\1\1\3\1$ &
$\0\0\3\2\0\0\1\1\0\2\1\1$ &
$\0\0\3\2\0\0\1\1\2\1\1\1$ \\
$\0\0\1\2\0\0\1\1\3\2\1\1$ &
$\0\0\3\2\2\0\1\1\1\1\1\1$ &
$\0\0\3\2\1\2\1\1\1\1\1\1$ &
$\0\2\3\2\1\1\1\1\1\1\1\1$ &
$\0\2\3\1\0\0\0\2\0\0\0\0$ &
$\0\2\1\1\0\0\0\2\0\3\0\0$ &
$\0\0\1\1\0\0\2\0\3\2\0\0$ \\
$\0\2\1\1\0\0\0\2\3\1\0\0$ &
$\0\2\3\1\0\0\2\1\0\0\0\0$ &
$\0\2\1\1\0\0\2\1\0\3\0\0$ &
$\0\2\1\1\0\0\2\1\3\1\0\0$ &
$\0\0\1\1\0\0\3\2\0\2\0\0$ &
$\0\0\1\1\0\0\3\2\2\1\0\0$ &
$\0\0\1\1\0\0\1\2\3\2\0\0$ \\
$\2\0\1\1\0\0\3\2\1\1\0\0$ &
$\0\0\3\1\0\2\1\1\0\0\0\0$ &
$\0\0\3\1\0\2\1\1\0\0\0\3$ &
$\0\0\3\1\0\2\1\1\0\0\3\1$ &
$\0\0\1\1\0\2\1\1\0\3\0\0$ &
$\0\0\1\1\0\2\1\1\0\3\0\3$ &
$\0\0\1\1\0\2\1\1\0\3\3\1$ \\
$\0\0\1\1\0\0\1\1\3\2\0\0$ &
$\0\0\1\1\0\0\1\1\3\2\0\3$ &
$\0\0\1\1\0\0\1\1\3\2\3\1$ &
$\0\0\1\1\0\2\1\1\3\1\0\0$ &
$\0\0\1\1\0\2\1\1\3\1\0\3$ &
$\0\0\1\1\0\2\1\1\3\1\3\1$ &
$\0\0\3\1\2\1\1\1\0\0\0\0$ \\
$\0\0\3\1\2\1\1\1\0\0\0\3$ &
$\0\0\3\1\2\1\1\1\0\0\3\1$ &
$\0\0\1\1\2\1\1\1\0\3\0\0$ &
$\0\0\1\1\2\1\1\1\0\3\0\3$ &
$\0\0\1\1\2\1\1\1\0\3\3\1$ &
$\0\0\1\1\2\1\1\1\3\1\0\0$ &
$\0\0\1\1\2\1\1\1\3\1\0\3$ \\
$\0\0\1\1\2\1\1\1\3\1\3\1$ &
$\0\0\3\1\3\2\1\1\1\1\1\1$ &
$\0\0\3\1\1\1\1\1\0\2\0\0$ &
$\0\0\3\1\1\1\1\1\0\2\0\3$ &
$\0\0\3\1\1\1\1\1\0\2\3\1$ &
$\0\0\3\1\1\1\1\1\2\1\0\0$ &
$\0\0\3\1\1\1\1\1\2\1\0\3$ \\
$\0\0\3\1\1\1\1\1\2\1\3\1$ &
$\0\2\3\1\1\1\1\1\1\1\0\0$ &
$\0\2\3\1\1\1\1\1\1\1\0\3$ &
$\0\2\3\1\1\1\1\1\1\1\3\1$ &
$\2\1\0\0\0\0\0\0\0\0\0\0$ &
$\2\1\0\0\0\0\0\0\0\0\0\3$ &
$\2\1\0\0\0\0\0\0\0\0\3\0$ \\
$\2\1\0\0\0\0\0\0\0\0\3\3$ &
$\0\3\0\0\0\2\0\3\0\0\0\0$ &
$\0\3\0\0\0\2\3\1\0\0\0\0$ &
$\2\1\0\0\0\3\0\0\0\0\0\0$ &
$\2\1\0\0\0\3\0\0\0\0\0\3$ &
$\2\1\0\0\0\3\0\0\0\0\3\1$ &
$\0\3\0\0\2\1\0\3\0\0\0\0$ \\
$\0\3\0\0\2\1\3\1\0\0\0\0$ &
$\2\1\0\0\3\1\0\0\0\0\0\0$ &
$\2\1\0\0\3\1\0\0\0\0\0\3$ &
$\2\1\0\0\3\1\0\0\0\0\3\1$ &
$\2\1\0\0\1\1\0\3\0\0\0\0$ &
$\2\1\0\0\1\1\0\3\0\0\0\3$ &
$\2\1\0\0\1\1\0\3\0\0\3\1$ \\
$\2\1\0\0\1\1\3\1\0\0\0\0$ &
$\2\1\0\0\1\1\3\1\0\0\0\3$ &
$\2\1\0\0\1\1\3\1\0\0\3\1$ &
$\2\1\0\0\1\1\1\1\0\3\0\0$ &
$\2\1\0\0\1\1\1\1\0\3\0\3$ &
$\2\1\0\0\1\1\1\1\0\3\3\1$ &
$\2\1\0\0\1\1\1\1\3\1\0\0$ \\
$\2\1\0\0\1\1\1\1\3\1\0\3$ &
$\2\1\0\0\1\1\1\1\3\1\3\1$ &
$\2\1\0\3\0\0\0\2\0\0\0\0$ &
$\2\1\0\3\0\0\2\1\0\0\0\0$ &
$\0\3\0\3\0\2\1\1\0\0\0\0$ &
$\0\3\0\3\2\1\1\1\0\0\0\0$ &
$\0\3\0\3\1\1\1\1\0\2\0\0$ \\
$\0\3\0\3\1\1\1\1\2\1\0\0$ &
$\2\1\0\3\1\1\1\1\1\1\0\0$ &
$\2\1\0\3\1\1\1\1\1\1\0\3$ &
$\2\1\0\3\1\1\1\1\1\1\3\1$ &
$\2\1\3\2\1\1\1\1\1\1\1\1$ &
$\2\1\3\1\0\0\0\2\0\0\0\0$ &
$\2\1\1\1\0\0\0\2\0\3\0\0$ \\
$\0\3\1\1\0\0\0\0\3\2\0\0$ &
$\2\1\1\1\0\0\0\2\3\1\0\0$ &
$\2\1\3\1\0\0\2\1\0\0\0\0$ &
$\2\1\1\1\0\0\2\1\0\3\0\0$ &
$\2\1\1\1\0\0\2\1\3\1\0\0$ &
$\0\3\3\1\0\2\1\1\0\0\0\0$ &
$\0\3\1\1\0\2\1\1\0\3\0\0$ \\
$\0\3\1\1\0\2\1\1\3\1\0\0$ &
$\0\3\3\1\2\1\1\1\0\0\0\0$ &
$\0\3\1\1\2\1\1\1\0\3\0\0$ &
$\0\3\1\1\2\1\1\1\3\1\0\0$ &
$\0\3\3\1\1\1\1\1\0\2\0\0$ &
$\0\3\3\1\1\1\1\1\2\1\0\0$ &
$\2\1\3\1\1\1\1\1\1\1\0\0$ \\
$\2\1\3\1\1\1\1\1\1\1\0\3$ &
$\2\1\3\1\1\1\1\1\1\1\3\1$ &
$\3\2\0\0\0\0\0\0\0\0\2\0$ &
$\3\2\0\0\0\0\0\0\0\0\1\0$ &
$\3\2\0\0\0\0\0\0\0\0\1\3$ &
$\3\2\0\0\0\3\0\0\0\0\1\1$ &
$\1\2\0\0\3\2\0\0\0\0\1\1$ \\
$\3\2\0\0\3\1\0\0\0\0\1\1$ &
$\3\2\0\3\0\0\0\0\0\0\0\0$ &
$\3\2\3\1\0\0\0\0\0\0\0\0$ &
$\3\2\1\1\0\0\0\0\3\0\0\0$ &
$\3\2\1\1\0\0\0\0\1\3\0\0$ &
$\3\2\1\1\0\0\0\3\1\1\0\0$ &
$\1\2\1\1\0\0\3\2\1\1\0\0$ \\
$\3\2\1\1\0\0\3\1\1\1\0\0$ &
$\1\1\0\0\2\0\0\0\0\0\0\0$ &
$\1\1\0\0\2\0\0\0\0\0\0\3$ &
$\1\1\0\0\2\0\0\0\0\0\3\0$ &
$\1\1\0\0\2\0\0\0\0\0\3\3$ &
$\3\1\0\0\0\2\0\3\0\0\0\0$ &
$\3\1\0\0\0\2\3\1\0\0\0\0$ \\
$\1\1\0\0\2\3\0\0\0\0\0\0$ &
$\1\1\0\0\2\3\0\0\0\0\0\3$ &
$\1\1\0\0\2\3\0\0\0\0\3\1$ &
$\3\1\0\0\2\1\0\3\0\0\0\0$ &
$\3\1\0\0\2\1\3\1\0\0\0\0$ &
$\1\1\0\0\1\2\0\0\0\0\0\0$ &
$\1\1\0\0\1\2\0\0\0\0\0\3$ \\
$\1\1\0\0\1\2\0\0\0\0\3\0$ &
$\1\1\0\0\1\2\0\0\0\0\3\3$ &
$\1\1\0\0\1\1\0\0\0\0\0\0$ &
$\1\1\0\0\1\1\0\0\0\0\0\3$ &
$\1\1\0\0\1\1\0\0\0\0\3\0$ &
$\1\1\0\0\1\1\0\0\0\0\3\3$ &
$\1\1\0\0\1\1\0\0\0\3\0\0$ \\
$\1\1\0\0\1\1\0\0\0\3\0\3$ &
$\1\1\0\0\1\1\0\0\0\3\3\0$ &
$\1\1\0\0\1\1\0\0\0\3\3\3$ &
$\1\1\0\0\1\1\0\0\3\1\0\0$ &
$\1\1\0\0\1\1\0\0\3\1\0\3$ &
$\1\1\0\0\1\1\0\0\3\1\3\0$ &
$\1\1\0\0\1\1\0\0\3\1\3\3$ \\
$\1\1\0\0\1\1\0\3\0\2\0\0$ &
$\1\1\0\0\1\1\0\3\0\2\0\3$ &
$\1\1\0\0\1\1\0\3\0\2\3\1$ &
$\1\1\0\0\1\1\0\3\2\1\0\0$ &
$\1\1\0\0\1\1\0\3\2\1\0\3$ &
$\1\1\0\0\1\1\0\3\2\1\3\1$ &
$\1\1\0\2\1\1\0\3\1\1\0\0$ \\
$\1\1\0\2\1\1\0\3\1\1\0\3$ &
$\1\1\0\2\1\1\0\3\1\1\3\1$ &
$\1\1\0\2\1\1\3\2\1\1\1\1$ &
$\1\1\0\0\1\1\3\1\0\2\0\0$ &
$\1\1\0\0\1\1\3\1\0\2\0\3$ &
$\1\1\0\0\1\1\3\1\0\2\3\1$ &
$\1\1\0\0\1\1\3\1\2\1\0\0$ \\
$\1\1\0\0\1\1\3\1\2\1\0\3$ &
$\1\1\0\0\1\1\3\1\2\1\3\1$ &
$\1\1\0\2\1\1\3\1\1\1\0\0$ &
$\1\1\0\2\1\1\3\1\1\1\0\3$ &
$\1\1\0\2\1\1\3\1\1\1\3\1$ &
$\1\1\0\3\0\2\0\2\0\0\0\0$ &
$\1\1\0\3\0\2\2\1\0\0\0\0$ \\
$\3\1\0\3\0\2\1\1\0\0\0\0$ &
$\1\1\0\3\2\1\0\2\0\0\0\0$ &
$\1\1\0\3\2\1\2\1\0\0\0\0$ &
$\3\1\0\3\2\1\1\1\0\0\0\0$ &
$\1\1\0\3\1\1\0\2\0\2\0\0$ &
$\1\1\0\3\1\1\0\2\2\1\0\0$ &
$\1\1\2\3\1\1\0\0\1\1\0\0$ \\
$\1\1\2\3\1\1\0\0\1\1\0\3$ &
$\1\1\2\3\1\1\0\0\1\1\3\1$ &
$\1\1\0\3\1\1\2\1\0\2\0\0$ &
$\1\1\0\3\1\1\2\1\2\1\0\0$ &
$\1\1\2\1\1\1\0\3\1\1\0\0$ &
$\1\1\2\1\1\1\0\3\1\1\0\3$ &
$\1\1\2\1\1\1\0\3\1\1\3\1$ \\
$\1\1\2\1\1\1\3\2\1\1\1\1$ &
$\3\1\0\3\1\1\1\1\0\2\0\0$ &
$\3\1\0\3\1\1\1\1\2\1\0\0$ &
$\1\1\2\1\1\1\3\1\1\1\0\0$ &
$\1\1\2\1\1\1\3\1\1\1\0\3$ &
$\1\1\2\1\1\1\3\1\1\1\3\1$ &
$\1\1\3\0\1\1\0\0\1\1\0\0$ \\
$\1\1\3\0\1\1\0\0\1\1\0\3$ &
$\1\1\3\0\1\1\0\0\1\1\3\0$ &
$\1\1\3\0\1\1\0\0\1\1\3\3$ &
$\1\1\3\2\1\1\3\0\1\1\1\1$ &
$\1\1\3\2\1\1\1\3\1\1\1\1$ &
$\1\1\3\1\0\2\0\2\0\0\0\0$ &
$\1\1\1\1\0\2\0\2\0\3\0\0$ \\
$\3\1\1\1\0\0\0\0\3\2\0\0$ &
$\1\1\1\1\0\2\0\2\3\1\0\0$ &
$\1\1\3\1\0\2\2\1\0\0\0\0$ &
$\1\1\1\1\0\2\2\1\0\3\0\0$ &
$\1\1\1\1\0\2\2\1\3\1\0\0$ &
$\1\1\1\1\0\2\3\2\1\1\0\0$ &
$\3\1\3\1\0\2\1\1\0\0\0\0$ \\
$\3\1\1\1\0\2\1\1\0\3\0\0$ &
$\3\1\1\1\0\2\1\1\3\1\0\0$ &
$\1\1\3\1\2\1\0\2\0\0\0\0$ &
$\1\1\1\1\2\1\0\2\0\3\0\0$ &
$\1\1\1\1\2\1\0\2\3\1\0\0$ &
$\1\1\3\1\2\1\2\1\0\0\0\0$ &
\end{tabular}
}
\label{table-rm-J}
\end{table}

\newpage

\begin{figure}[h]
  \includegraphics[scale=0.11]{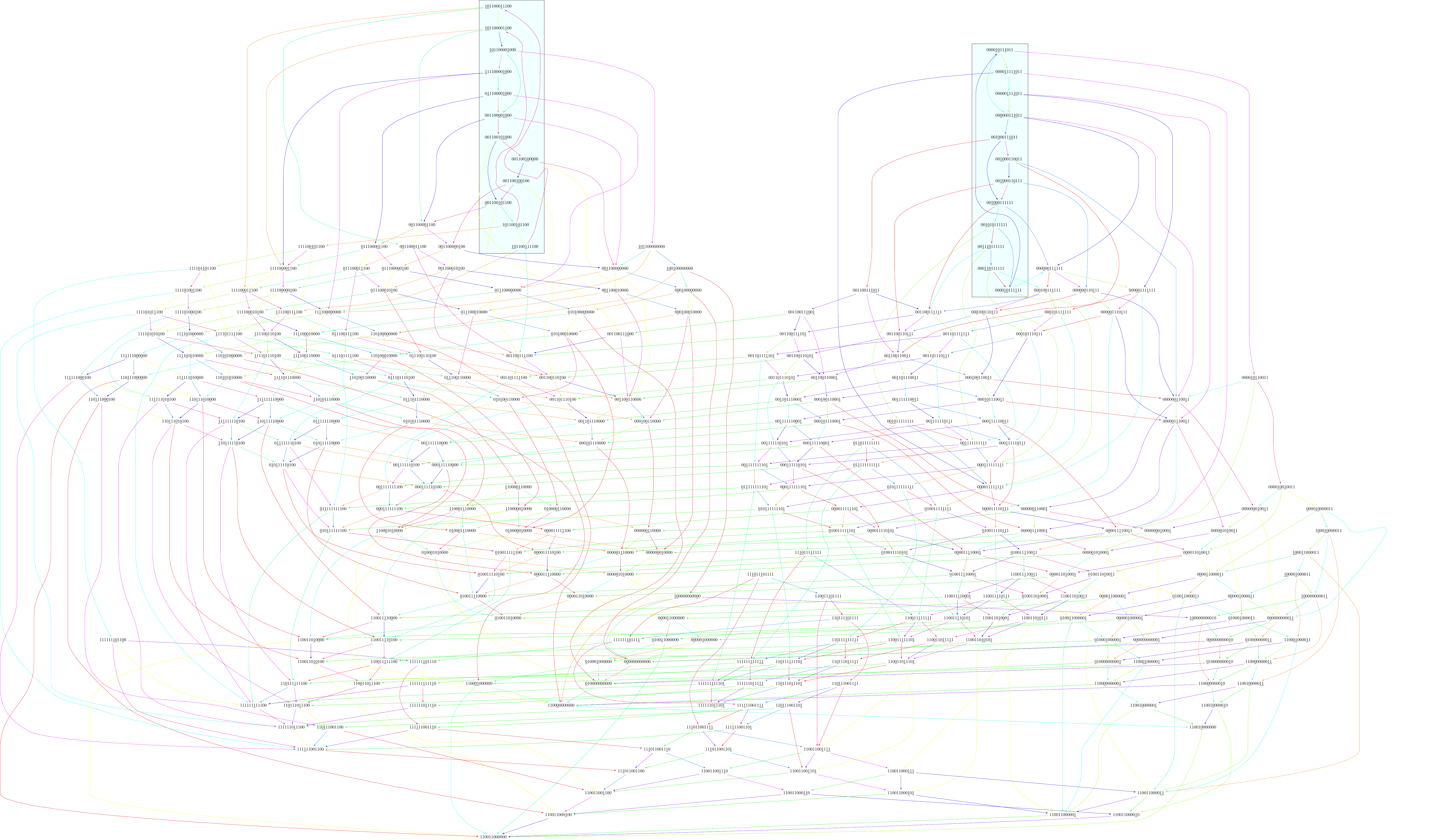}%
  \caption{The USO corresponding to the representation map $r_2$ of
    $C_H$ described in Table~\ref{table-rm-J}.
    This map has two NSCCs of size 12 (in blue).}
  \label{fig-rm-ce-J}
\end{figure}

\begin{figure}[h]
  \includegraphics[scale=0.13]{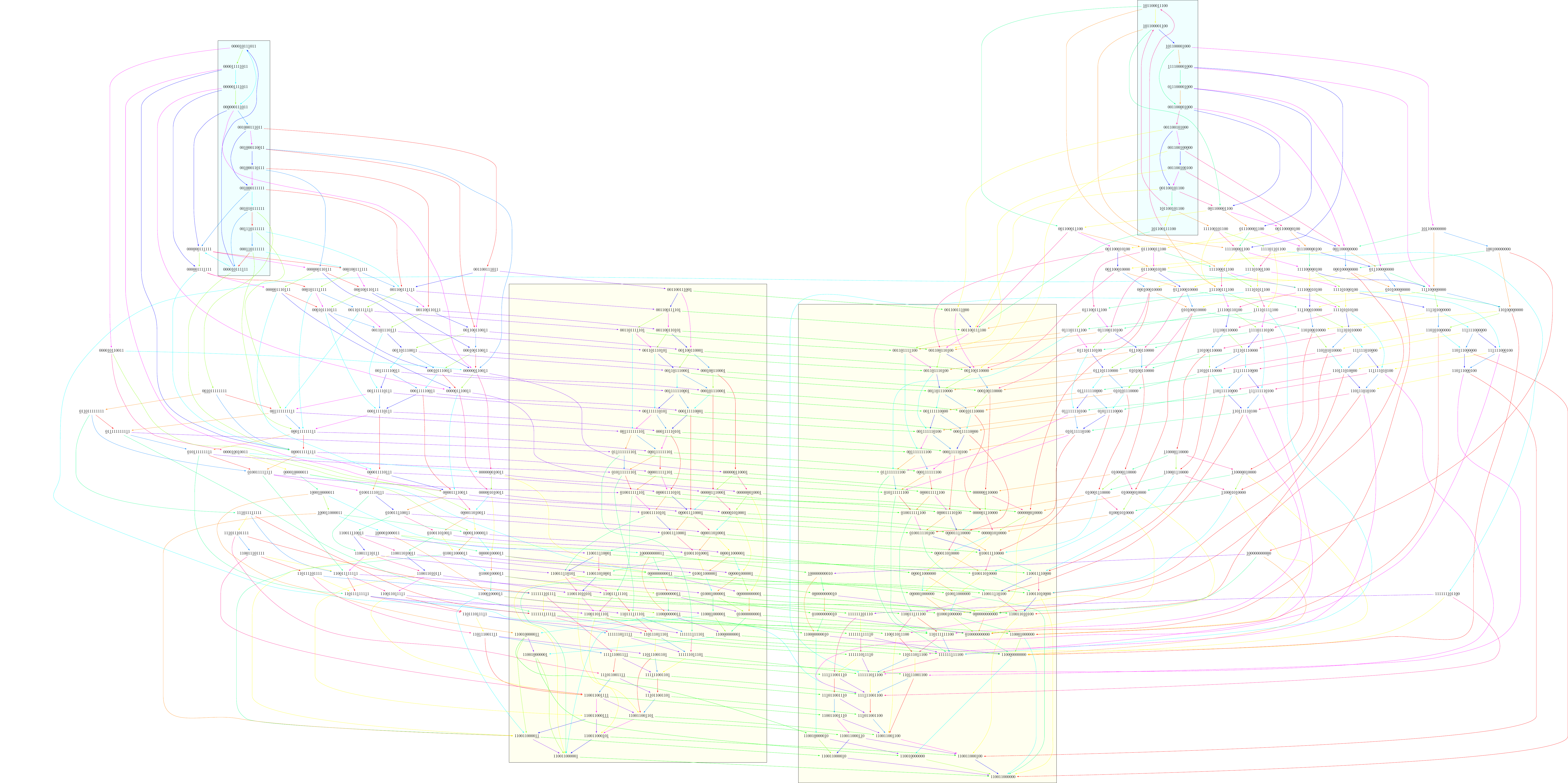}%
  \caption{Another representation of the USO corresponding to $r_2$
    illustrating the top recursion (on dimension 12) of the algorithm
    from Theorem~\ref{th-maximum}: $0C^x$ and $1C^x$ are represented
    by yellow boxes.}%
  \label{fig-rm-ce-CD}
\end{figure}

\end{document}